\documentclass[12pt,letterpaper]{article}
\usepackage[margin=1.25in]{geometry}
\usepackage{setspace}
\usepackage{graphicx}
\usepackage[authoryear,round,sort,comma]{natbib}
\usepackage{times}
\providecommand{\bibname}{References}

\def\startlocaldefs{}
\def\endlocaldefs{}
\newenvironment{frontmatter}{}{}

\newcommand{\runtitle}[1]{}
\makeatletter
\global\let\@author\@empty
\renewcommand{\author}[2][]{\g@addto@macro\@author{#2\\[0.25ex]}}
\newcommand{\address}[2][]{\g@addto@macro\@author{{\small #2}\\[0.25ex]}}
\makeatother

\newif\ifkeywordfirst
\newenvironment{keyword}{\par\noindent\textbf{Keywords:} \keywordfirsttrue}{\par\medskip}
\newcommand{\kwd}[1]{\ifkeywordfirst\keywordfirstfalse\else\unskip; \fi #1\ignorespaces}
\AtBeginDocument{\setlength{\textheight}{32\baselineskip}}

\usepackage[english]{babel}
\usepackage{mathrsfs}
\usepackage{mathtools}
\usepackage{enumitem}
\usepackage{amsthm}
\usepackage{amssymb}
\usepackage{calc}
\usepackage{cancel}
\usepackage{mathdots}
\usepackage{microtype}
\usepackage{stackrel}
\PassOptionsToPackage{version=3}{mhchem}
\usepackage{mhchem}
\usepackage{esint}
\usepackage{threeparttable,booktabs}
\usepackage{longtable, array}
\usepackage[table]{xcolor}
\usepackage{resizegather}
\usepackage{tikz}
\usetikzlibrary{arrows.meta, positioning, calc}
\resizegathersetup{warningthreshold=0.02}
\definecolor{rowgray}{gray}{0.93}
\allowdisplaybreaks[4]

\usepackage[colorlinks,citecolor=blue,linkcolor=blue,urlcolor=blue]{hyperref}

\makeatletter
\def\abstractname@afterskip{\space}
\def\keywordname@size{\normalfont}
\def\keyword@skip{30\p@}
\renewcommand{\bibname}{References}
\addto\captionsenglish{\renewcommand{\bibname}{References}}
\let\put@numberlines@box\@empty
\makeatother

\startlocaldefs
\numberwithin{equation}{section}

\global\long\def\R{\mathbb{R}}%
\global\long\def\N{\mathbb{N}}%
\global\long\def\e{\mathbb{E}}%
\global\long\def\Var{\mathrm{Var}}%
\global\long\def\bs#1{\boldsymbol{#1}}%
\global\long\def\trans{\top}%
\global\long\def\VC{\mathrm{VC}}%
\global\long\def\assign{:=}%

\providecommand{\assumptionname}{Assumption}
\providecommand{\corollaryname}{Corollary}

\providecommand{\examplename}{Example}

\providecommand{\lemmaname}{Lemma}

\providecommand{\propositionname}{Proposition}
\providecommand{\remarkname}{Remark}

\providecommand{\theoremname}{Theorem}

\theoremstyle{definition}
\newtheorem{example}{\protect\examplename}[section]
\theoremstyle{plain}
\newtheorem{assumption}{\protect\assumptionname}[section]
\theoremstyle{plain}
\newtheorem{thm}{\protect\theoremname}[section]
\theoremstyle{definition}
\newtheorem{rem}{\protect\remarkname}[section]
\theoremstyle{plain}
\newtheorem{cor}{\protect\corollaryname}[section]
\theoremstyle{plain}
\newtheorem{lem}{\protect\lemmaname}[section]
\theoremstyle{plain}
\newtheorem{prop}{\protect\propositionname}[section]

\renewcommand{\hat}{\widehat}
\renewcommand{\tilde}{\widetilde}

\endlocaldefs

\begin{document}
\begin{frontmatter}
\hypersetup{linkcolor=black}
\renewcommand{\thefootnote}{\fnsymbol{footnote}}
\begin{center}
{\LARGE\bfseries Data-Automated Policy Learning for Nonlinear Welfare\par}
\vspace{1.2em}
{\large Chunrong Ai\textsuperscript{a}\footnote{Email: \href{mailto:chunrongai@cuhk.edu.cn}{\texttt{chunrongai@cuhk.edu.cn}}.}\quad Zeqi Wu\textsuperscript{b}\footnote{Email: \href{mailto:wuzeqi@ruc.edu.cn}{\texttt{wuzeqi@ruc.edu.cn}}.}\quad Zheng Zhang\textsuperscript{b}\footnote{Email: \href{mailto:zhengzhang@ruc.edu.cn}{\texttt{zhengzhang@ruc.edu.cn}}.}\par}
\vspace{0.6em}
{\small \textsuperscript{a}\,School of Management and Economics, The Chinese University of Hong Kong, Shenzhen\par}
{\small \textsuperscript{b}\,Institute of Statistics and Big Data, Renmin University of China\par}
\end{center}
\renewcommand{\thefootnote}{\arabic{footnote}}
\setcounter{footnote}{0}
\hypersetup{linkcolor=blue}
\vspace{1em}

\begin{abstract}
This paper explores policy learning from observational data, focusing on a nonlinear welfare criterion in a binary treatment setting. The nonlinear criterion is inspired by scenarios where policymakers prioritize specific population segments. We model this criterion using a utility function that encompasses potential outcomes and intermediate parameters, with the latter capturing higher moments of the outcome distributions. When formulated in the context of observational data, both the intermediate parameters and the welfare criterion depend on the propensity score, which we estimate using machine-learning techniques. To address bias in machine learning estimates, we introduce a novel reweighting-based debiasing approach that offers a promising alternative to traditional orthogonality-based methods. To tackle the complexities of infinite-dimensional policy spaces, we employ sieve approximations and $K$-fold cross-validation for model selection, thereby fully automating the policy-learning process. Despite these complexities, we demonstrate that both the welfare regret and the average welfare regret of our proposed policy learning method satisfy an oracle inequality, thereby providing theoretical guarantees on the performance of the estimated policy relative to the best possible policy. This finding extends the existing results from linear to nonlinear welfare criteria, from finite-dimensional to infinite-dimensional policy spaces, and from a known propensity score to a machine-learned one.
\end{abstract}

\begin{keyword}
\kwd{Policy Learning}
\kwd{Oracle Inequality}
\kwd{Sieve Approximation}
\kwd{Machine Learning}
\kwd{Welfare Criterion}
\end{keyword}
\end{frontmatter}

\section{Introduction}

There is an increasing body of literature on policy learning. Most existing studies focus on a linear welfare criterion and a binary treatment setting. In this context, the binary variable $T$ indicates program participation status: $T=1$ if the individual participates and $T=0$ if not. The potential outcome associated with participation status $t$ is represented as $Y^{*}(t)\in\mathcal{Y}\subset\mathbb{R}$, and $\bs X\in\mathcal{X}\subset\mathbb{R}^{d}$ denotes individual characteristics. A policy $\pi:\mathcal{X}\to\{0,1\}$ decides whether to assign an individual with attributes $\bs{X}$ to the program. The potential outcome under the policy is given by:
\[
Y^{*}(\pi(\bs X))=\pi(\bs X)Y^{*}(1)+(1-\pi(\bs X))Y^{*}(0).
\]
Policymakers specify a welfare criterion $W(\pi)$ and a policy space $\Pi_{\infty}$, aiming to learn the optimal policy defined by
\[
\pi^{*}=\arg \max_{\pi \in \Pi_{\infty}} W(\pi).
\]
A common choice for the welfare is the expectation of the potential outcome:
\[
W(\pi)=\e[Y^{*}(\pi(\bs X))]=\e[\pi(\bs X)Y^{*}(1)+(1-\pi(\bs X))Y^{*}(0)],
\]
which is clearly linear with respect to the policy.

A risk-averse policymaker might prefer to utilize the average utility of potential outcomes, represented as
\[
W(\pi)=\e[U(Y^{*}(\pi(\bs X)))]=\e[\pi(\bs X)U(Y^{*}(1))+(1-\pi(\bs X))U(Y^{*}(0))].
\]
However, this extension is only superficial: redefining the potential outcome as $U(Y^*(t))$ demonstrates that the resulting welfare remains linear in policy. By contrast, many significant applications lead to welfare criteria that are genuinely nonlinear in policy. For instance, in income-inequality literature, policymakers aim to reduce disparities through targeted interventions, and the welfare criterion is often a measure of inequality, such as the Gini coefficient \citep{gastwirth1971General,gastwirth1972Estimation}:
\[
W(\pi)= -\frac{2M^{-1}\sum_{r=1}^{M}\alpha_{r}\beta^{*}(\pi;\alpha_{r})}{\e [Y^{*}(\pi(\bs X))]} + 1,\qquad \alpha_{r}=\frac{r}{M+1},
\]
where $M$ is a finite positive integer and $\beta^{*}(\pi;\alpha)$ denotes the $\alpha$-quantile of $Y^{*}(\pi(\bs X))$. Since quantiles are nonlinear in relation to policy choices, the welfare criterion is likewise nonlinear. Nonlinearity also appears in alternative inequality measures, such as the relative standing of specific population segments:
\[
W(\pi)=\frac{\e\left[Y^{*}(\pi(\bs X))\mid Y^{*}(\pi(\bs X))\leq\beta^{*}(\pi;0.5)\right]}{\e\left[Y^{*}(\pi(\bs X))\right]},
\]
or the relative status of a defined subpopulation:
\[
W(\pi)=\frac{\e\left[Y^{*}(\pi(\bs X))1(\text{gender}=\text{``female''})\right]}{\e\left[Y^{*}(\pi(\bs X))\right]},
\]
and in the analysis of upper-tail (90/50) and lower-tail (50/10) ratios \citep{autor2008trends}:
\[
W(\pi)=-\frac{\beta^{*}(\pi;0.9)}{\beta^{*}(\pi;0.5)}\text{ and }W(\pi) = -\frac{\beta^{*}(\pi;0.5)}{\beta^{*}(\pi;0.1)}.
\]

Nonlinear welfare criteria naturally arise in the fields of risk management and public health. In risk management, firms aim to control the risk of significant losses. Let the binary variable $T\in\{0,1\}$ represent an investment decision, where $T=1$ indicates a one-period investment in the asset and $T=0$ indicates no investment. Let $Y^{*}(1)$ denote the corresponding one-period payoff (or return), while we set $Y^{*}(0)=0$. Given characteristics $\bs X$ (e.g., volatility, momentum, or fundamentals), an investment strategy $\pi:\mathcal X\to\{0,1\}$ induces the realized payoff: 
\[
Y^{*}(\pi(\bs X))=\pi(\bs X)Y^{*}(1)+(1-\pi(\bs X))Y^{*}(0)=\pi(\bs X)Y^{*}(1).
\]
We denote $L^{*}(\pi)=-Y^{*}(\pi(\bs X))$ and let $\beta^{*}(\pi;\alpha)$ represent the $\alpha$-quantile of $L^{*}(\pi)$. The welfare criterion can be expressed as a measure of risk, such as Value-at-Risk (VaR), defined as $W(\pi)=-\beta^{*}(\pi;\alpha)$, or Conditional Value-at-Risk (CVaR), given by $W(\pi)=-(1-\alpha)^{-1}\e\!\left[L^{*}(\pi)\cdot 1\!\left(L^{*}(\pi)\ge \beta^{*}(\pi;\alpha)\right)\right]$ \citep{rockafellar2000Optimization}, or the spectral risk measure (e.g., \citealp{acerbi2002Spectral}):   
\[
W(\pi)=-\sum_{r=1}^{M}\phi(\alpha_{r})\beta^{*}(\pi;\alpha_{r}),\qquad \alpha_{r}=\frac{r}{M+1},
\]
with a weighting function $\phi(\alpha)$ (see \citealp{dowd2006Var}). All of these risk measures are nonlinear in relation to policy.

In public health, policymakers aim to reduce the incidence of severe post-discharge utilization through targeted programs. For example, in the United States, Medicare's Hospital Readmissions Reduction Program (HRRP) specifically targets 30-day unplanned readmissions by penalizing hospitals with higher-than-average readmission rates \citep{zuckerman2016readmissions,ryan2017valuebased,khera2018hrrp}. Research shows that readmission-related utilization is highly right-skewed: a small fraction of patients accounts for a disproportionate share of readmissions and hospital use \citep{fouayzi2022highfrequency,manning2005generalized}. These distributional features motivate policymakers to prioritize the upper tail of the outcome distribution rather than focusing solely on the average. In this context, let $T$ represent the intensity of post-discharge care for a patient, where $T=1$ indicates enhanced transitional care (e.g., intensive monitoring and follow-up) and $T=0$ signifies usual care. The variable $Y^{*}(t)\geq 0$ reflects the readmission burden under treatment $t$, such as the total number of inpatient days due to unplanned readmissions within 30 days post-discharge, with lower values being more desirable. Given patient covariates $\bs X$ (e.g., age, gender, BMI, and comorbidity indices), a policy $\pi:\mathcal X\to\{0,1\}$ assigns care intensity and determines the burden $Y^{*}(\pi(\bs X))$. The welfare criterion being considered is a tail-sensitive measure of the expected burden among the worst-off $(1-\alpha)$ fraction of patients:
\[
W(\pi)
=
-\frac{1}{1-\alpha}\e\!\left[
Y^{*}(\pi(\bs X))\cdot
1\!\left\{Y^{*}(\pi(\bs X))\ge \beta^{*}(\pi;\alpha)\right\}
\right],
\]
where $\beta^{*}(\pi;\alpha)$ denotes the $\alpha$-quantile of $Y^{*}(\pi(\bs X))$ and is nonlinear in relation to the policy.

All the aforementioned examples emphasize the need for a nonlinear welfare criterion that directly depends on the policy through potential outcomes and indirectly through intermediate parameters $\bs{\beta}^*(\pi)$ (such as the quantiles mentioned earlier). Specifically, we model this with: 
\begin{equation}
    W(\pi)=\e\left[U(Y^{*}(\pi(\bs X)),\bs X,\bs{\beta}^{*}(\pi))\right],\label{eq:general welfare}
\end{equation}
where $U(\cdot,\cdot,\cdot)$ is a known utility function. 

The policy space can be complex and infinite-dimensional, making the optimal policy $\pi^{*}$ challenging to compute. To overcome this issue, we leverage the sieve literature to approximate the policy space using a sequence of finite-dimensional sieve classes $\{\Pi_{\ell}: \ell=1,2,\ldots\}$. Within each class, we determine the best policy as
\[\pi^{*}_{\ell}=\arg\max_{\pi\in \Pi_{\ell}}W(\pi).
\]
Since the welfare criterion is unknown, we estimate it from a training sample, say $I$, by $\widehat{W}_{I}(\pi)$ and find the best policy using: 
\[
\widehat{\pi}_{\ell,I}=\arg\max_{\pi\in \Pi_{\ell}}\widehat{W}_{I}(\pi).
\]
We then apply $K$-fold cross-validation to select the optimal policy approximation space $\Pi_{\widehat{\ell}}$ (see (\ref{eq:def of l_hat-general})) and subsequently estimate the optimal policy as $\widehat{\pi}$ (see (\ref{eq:CV criterion-general})). Despite these extensions, we establish the following oracle inequality for the average welfare regret:
\begin{equation}
\begin{aligned}
    & \e\left[W(\pi^{*})-W(\widehat{\pi})\right]\\  \leq \ &\inf_{\ell=1,2,\ldots}\left\{ \underbrace{W(\pi^{*})-\max_{\pi\in\Pi_{\ell}}W(\pi)}_{\mathrm{approximation \ error}}+\underbrace{\max_{\pi\in\Pi_{\ell}}W(\pi)-\e\left[W(\widehat{\pi}_{\ell,I})\right]}_{\mathrm{estimation \ error}}+\frac{\log\ell}{\sqrt{N}}\right\} +\sqrt{\frac{C}{N}}.
\end{aligned}
\label{eq:oracle inequality}
\end{equation}

\textit{Related Literature.} Our work builds on the policy-learning literature initially developed by \citet{manski2004Statistical}. A significant portion of this literature examines treatment choices under linear welfare criteria, in which an average or utilitarian objective defines the policy value. Early contributions to this field include studies by \citet{hirano2009asymptotics,stoye2009minimax,stoye2012minimax,bhattacharya2012inferring,tetenov2012statistical,qian2011Performance,zhao2012Estimating}, and more recent works by \citet{kitagawa2018Who,athey2021Policy,mbakop2021Model,luedtke2016Statistical,crippa2025Regret,liu2025nonparametric,fang2025semiparametric,fang2025model}. With the exception of \citet{mbakop2021Model}, all these studies assume a finite-dimensional policy space (e.g., $\Pi_{\infty}=\Pi_{\ell}$ with $\ell$ fixed or allowed to grow with sample size), which allows them to avoid the complexities of policy-space approximation and data-driven model selection. They either assume a known propensity score or estimate it using machine-learning methods, then apply double debiasing to remove the machine-learning bias. In both cases, they ultimately establish an explicit upper bound on the welfare regret:
\begin{equation}
\e\left[W(\pi^{*})-W(\widehat{\pi}_{\ell,I})\right]\leq C\sqrt{\frac{\mathrm{VC}(\Pi_{\ell})}{N}} \label{eq:first VC},
\end{equation}
where $\VC(\Pi_{\ell})$ represents the Vapnik-Chervonenkis (VC) dimension of the policy class. In contrast, \citet{mbakop2021Model} addresses an infinite-dimensional policy space with a known propensity score, thus avoiding the need for debiasing. They present an upper bound on the average welfare regret that is similar to ours, though still under a linear welfare criterion. Our contribution extends this literature to encompass a nonlinear welfare criterion, a machine-learned propensity score, and an infinite-dimensional policy space.

A smaller body of literature explores policy learning for nonlinear welfare criteria within a finite-dimensional policy class, where the dimension may increase with sample size (that is, $\Pi_{\infty}=\Pi_{\ell}$, with $\ell$ allowed to grow with sample size). Notable examples include \citet{wang2018QuantileOptimal} on quantile-optimal treatment regimes, \citet{chen2025quantileoptimalpolicylearningunmeasured} addressing a quantile-based welfare criterion with an unmeasured confounder, \citet{fan2025policylearningalphaexpectedwelfare} focusing on a conditional value-at-risk criterion, \citet{kitagawa2021EqualityMinded} discussing an equality-minded social welfare criterion, and \citet{terschuur2025locally} examining nonlinear welfare criteria defined through U-statistics. These studies assume a finite-dimensional policy space, treat the propensity score as either known or derived via machine learning, and employ double-debiasing techniques to mitigate machine-learning bias. Despite the nonlinearity, they manage to derive a similar, explicit upper bound on the average welfare regret, akin to the upper bound mentioned in equation (\ref{eq:first VC}). We aim to extend this existing literature to encompass infinite-dimensional policy spaces.

To learn the optimal policy from observational data, it is essential to estimate both the intermediate parameters and the welfare criterion. Since both estimates rely on the unknown propensity score, we use machine-learning algorithms to estimate it. It is well recognized that machine-learned propensity scores introduce bias in both the intermediate parameter and the welfare-criterion estimates, which in turn affects the (average) welfare regret and slows its convergence rate. To correct for this machine-learning bias, existing literature applies double-debiasing procedures based on a Neyman orthogonality condition (see \citet{athey2021Policy,robins1994Estimation,chernozhukov2018Double} for linear welfare criteria and \citet{fan2025policylearningalphaexpectedwelfare,terschuur2025locally} for nonlinear criteria). Our procedure includes an additional step, estimating the intermediate parameters, so we need to debias both these parameters and the welfare criterion estimates simultaneously. We propose a reweighting method inspired by covariate-balancing approaches \citep{imai2014Covariate,chan2016Globally,ai2021Unified}, adapted here for bias correction. Covariate-balancing methods have shown strong performance in finite samples, and we expect our debiasing procedure to exhibit similar effectiveness. To our knowledge, this reweighting-based debiasing technique is novel in the literature and serves as a valuable alternative to double debiasing based on Neyman orthogonality.

The remainder of the paper is organized as follows. Section~\ref{sec:A-General-Learning} presents a data-automated optimal policy learning procedure that employs a generic welfare estimate alongside $K$-fold cross-validation to identify the best policy subclass, while establishing an oracle inequality for both the average welfare regret and the welfare regret. Section~\ref{sec:general framework} formally defines the model, expressing the unknown parameters and the welfare criterion in terms of the observed data, based on the assumptions of unconfoundedness and overlap. Section~\ref{sec:empirical welfare} introduces a machine-learning propensity-score estimator, along with a novel reweighting debiasing procedure for estimating the intermediate parameters and the welfare criterion. Section~\ref{sec:asymptotic properties} verifies that the proposed estimator for the welfare criterion meets the high-level conditions outlined in Section~\ref{sec:A-General-Learning}. Section~\ref{sec:empirical application} applies the proposed methodology to data from the National Job Training Partnership Act (JTPA) Study. Finally, Section~\ref{sec:conclusion} summarizes the findings. All omitted proofs are included in the Appendix.

\section{A Data-Automated Learning Procedure\label{sec:A-General-Learning}}

We will outline a data-driven policy-learning procedure that relies on a generic welfare estimator and cross-validation. Let $\left\{ \left(Y_{i},\bs X_{i},T_{i}\right)\right\} _{i=1}^{N}$ represent an independent and identically distributed (i.i.d.) sample. Throughout the paper, we use $I\subset\{1,2,\ldots,N\}$ to index a generic training subsample and $|I|$ to denote its sample size. Let $\widehat{W}_{I}(\pi)$ be a generic estimator of $W(\pi)$ calculated from the training subsample $\left\{ \left(Y_{i},\bs X_{i},T_{i}\right)\right\} _{i\in I}$. We assume that $\widehat{W}_{I}(\pi)$ is pointwise $\sqrt{\left|I\right|}$-consistent for $W(\pi)$ and that its estimation error satisfies an exponential probability bound, as formalized in Assumption~\ref{assu:general assu on hatW}.

\begin{assumption}
\label{assu:general assu on hatW}
For a fixed policy $\pi\in\Pi_{\infty}$, there exist finite constants $C_{1},\ldots,C_{4}>0$, independent of $\delta$, $I$, and $\pi$, such that the following inequality holds for all $\delta>0$ and $\left|I\right| >C_{4}$: 
\[
P\left(\left|\widehat{W}_{I}(\pi)-W(\pi)\right|\geq\delta+\frac{C_{1}}{\sqrt{\left|I\right|}}\right)\leq C_{2}\exp(-C_{3}\left|I\right|\delta^{2}).
\] 
 
\end{assumption}
In applications, users must verify that their welfare estimators satisfy this high-level condition. We will present a welfare estimator and confirm that it indeed satisfies Assumption~\ref{assu:general assu on hatW}. 

With $\widehat{W}_{I}(\pi)$ established, a natural policy learning strategy is to maximize it over $\Pi_{\infty}$. However, for an infinite-dimensional $\Pi_{\infty}$, such a strategy can be computationally intractable and prone to overfitting. Following the work of \citet{mbakop2021Model}, we approximate $\Pi_{\infty}$ by a nested sequence of finite-dimensional policy subclasses $\Pi_{\ell}\subset\Pi_{\ell+1}\subset\cdots\subset\Pi_{\infty}$.\footnote{Throughout the paper, we take this approximating sequence to satisfy $\VC(\Pi_{\ell})<\infty$ for every $\ell\geq1$ and $W(\pi^{*})-\max_{\pi\in\Pi_{\ell}}W(\pi)\to0$ as $\ell\to\infty$.} Within each subclass $\Pi_{\ell}$, we estimate the best policy by 
\begin{equation}
\widehat{\pi}_{\ell,I}:=\underset{\pi\in\Pi_{\ell}}{\arg\max}\ \widehat{W}_{I}(\pi).\label{eq:def of each optimal pi-general}
\end{equation}
We determine the best subclass using a $K$-fold cross-validation (CV) procedure, as outlined in various studies (\citet{Hall1983Large}, \citet{Stone1974Cross-Validatory}, \citet{lecue2012Oracle}, and \citet{gyorfi2002DistributionFree}). Specifically, for a fixed integer $K\geq2$, we partition the index set $\{1,\ldots,N\}$ into $K$ disjoint folds of equal size. For each pair $(k,\ell)$, let $I_{k}$ denote the $k$th fold and $I_{-k}=\{1,\ldots,N\}\setminus I_{k}$ its complement. We utilize the training subsample $I_{-k}$ to learn the welfare  $\widehat{W}_{I_{-k}}(\pi)$ and the candidate policy $\widehat{\pi}_{\ell,I_{-k}}$. The holdout sample $I_{k}$ is then used to compute $\widehat{W}_{I_{k}}(\pi)$ and to evaluate the performance of the candidate policy on the holdout sample using $\widehat{W}_{I_{k}}(\widehat{\pi}_{\ell,I_{-k}})$. The $K$-fold CV procedure selects the best subclass according to the formula: 
\begin{equation}
\widehat{\ell}=\underset{\ell=1,2,\ldots}{\arg\max}\ \left\{ \frac{1}{K}\sum_{k=1}^{K}\widehat{W}_{I_{k}}(\widehat{\pi}_{\ell,I_{-k}})-\frac{\log\ell}{\sqrt{N}}\right\}, \label{eq:def of l_hat-general}
\end{equation}
where the penalty term $\log\ell/\sqrt{N}$ helps to prevent the selection of very large policy classes. The identified policy is then defined as: 
\begin{equation}
\widehat{\pi}:=\widehat{\pi}_{\widehat{\ell},I_{-\widehat{k}}},\ \text{where }\widehat{k}:=\underset{1\leq k\leq K}{\arg\max}\ \widehat{W}_{I_{k}}(\widehat{\pi}_{\widehat{\ell},I_{-k}}).\label{eq:CV criterion-general}
\end{equation}
The following theorem establishes the oracle inequality (\ref{eq:oracle inequality}) with $I=I_{-k}$ for any $k$. 

\begin{thm}
\label{thm:oracle inequality-general}
Assuming that the conditions outlined in Assumption \ref{assu:general assu on hatW} hold and that $C>0$ is a finite constant, the learned optimal policy $\widehat{\pi}$ defined in (\ref{eq:CV criterion-general}) satisfies the oracle inequality (\ref{eq:oracle inequality}) for sufficiently large $N$.

\end{thm}

\begin{rem}
The study by \citet{mbakop2021Model} utilized a single holdout sample to identify the best policy class, which introduces randomness due to relying on that single sample. The $K$-fold cross-validation (CV) procedure mitigates this randomness by averaging over multiple holdout samples.
\end{rem}
\begin{rem}
\label{rem:finite L}
The penalty term ``$\log\ell/\sqrt{N}$'' in (\ref{eq:def of l_hat-general}) serves as a regularizer that prevents the selection of very large policy classes (as noted in \cite{mbakop2021Model}). This term can be omitted when the number of candidate subclasses is finite and may even increase with $N$. For example, one might use the following equations: 
\[
\widehat{\ell}:=\underset{\ell=1,\ldots,L_{N}}{\arg\max}\sum_{k=1}^{K}\widehat{W}_{I_{k}}(\widehat{\pi}_{\ell,I_{-k}}) \text{ and } \widehat{k}:=\underset{1\leq k\leq K}{\arg\max} \ \widehat{W}_{I_{k}}(\widehat{\pi}_{\widehat{\ell},I_{-k}})
\]
where $L_{N}\geq2$ is a sequence of positive integers. In this scenario, the learned optimal policy $\widehat{\pi}$ fulfills the condition:
\begin{align*}
    &\e\left[W(\pi^{*})-W(\widehat{\pi})\right]\\
      \leq &  \ \inf_{\ell=1,\ldots,L_{N}}\Biggl\{
    \underbrace{W(\pi^{*})-\max_{\pi\in\Pi_{\ell}}W(\pi)}_{\mathrm{approximation\ error}}
    +\underbrace{\max_{\pi\in\Pi_{\ell}}W(\pi)-\e\left[W(\widehat{\pi}_{\ell,I_{-1}})\right]}_{\mathrm{estimation\ error}}
    \Biggr\} +C\sqrt{\frac{\log L_{N}}{N}}
\end{align*}
for sufficiently large $N$; the proof is given in Appendix~\ref{sec:proofs-general-learning}. The bound compares this policy with the best tradeoff among the first $L_{N}$ policy classes, but it may not necessarily be the best overall.  \end{rem}
\begin{rem}
    Once the best subclass $\Pi_{\widehat{\ell}}$ has been selected, it is tempting to re-learn the policy using the full sample by maximizing $\widehat{W}_{\{1,\ldots,N\}}(\pi)$ over $\pi\in\Pi_{\widehat{\ell}}$. However, without additional stability-type conditions, as discussed in \citep{bousquet2002stability}, this learned policy may not satisfy the oracle inequality (\ref{eq:oracle inequality}), because the policy selected from a subsample might not be the best after retraining on the full sample, as noted by \citep[Example~2.8]{lecue2012Oracle}.
    \end{rem}

The theorem generalizes the oracle inequality by providing a bound on the average welfare regret $\e[W(\pi^{*})-W(\widehat{\pi})]$. In practice, policymakers can only access a single sample, and are more concerned with the realized welfare regret $W(\pi^{*})-W(\widehat{\pi})$. We show that a similar oracle inequality holds with high probability.

\begin{cor}
\label{cor:oracle inequality prob bound}
Suppose Assumption \ref{assu:general assu on hatW} holds. Let $C,C^{\prime},C_{1},C_{2}>0$ be finite constants. The learned policy $\widehat{\pi}$, defined in (\ref{eq:CV criterion-general}) satisfies the inequality: 
\begin{align*}
    &W(\pi^{*})-W(\widehat{\pi})\\
   \leq & \inf_{\ell=1,2,\ldots}\left\{ \underbrace{W(\pi^{*})-\max_{\pi\in\Pi_{\ell}}W(\pi)}_{\mathrm{approximation\ error}}+\underbrace{\max_{\pi\in\Pi_{\ell}}W(\pi)-\frac{1}{K}\sum_{k=1}^{K}W(\widehat{\pi}_{\ell,I_{-k}})}_{\mathrm{estimation\ error}}+2\frac{\log\ell}{\sqrt{N}}\right\} +\sqrt{\frac{C+\delta}{N}}.
   \end{align*}
This holds with a probability of at least $1-C_{1}\exp(-C_{2}\delta)$ for all $\delta>0$ and sufficiently large $N$.
\end{cor}

Corollary~\ref{cor:oracle inequality prob bound} establishes a probability oracle inequality. The constant ``2'' in front of $\log\ell/\sqrt{N}$ arises from a union bound used to control the concentration of $\frac{1}{K}\sum_{k=1}^{K}\widehat{W}_{I_{k}}(\widehat{\pi}_{\ell,I_{-k}})$ around $\frac{1}{K}\sum_{k=1}^{K}W(\widehat{\pi}_{\ell,I_{-k}})$. This constant can be replaced with any value greater than $1$, with corresponding adjustments to $C,C_1,C_2$.

The oracle inequality provides insights into the quality of the learned policy only when both the approximation and estimation errors are minimal. We can quantify the estimation error under a strengthened Assumption~\ref{assu:general assu on hatW}.

\begin{assumption}
\label{assu:general ass uniform convergence}
For the training sample $I$ and a subclass of policies $\Pi\subset\Pi_{\infty}$ with VC dimension satisfying $\mathrm{VC}(\Pi)/|I|\to0$, we can assume that: 
\[
P\left(\sup_{\pi\in\Pi}\left|\widehat{W}_{I}(\pi)-W(\pi)\right|\geq\delta+C_{1}\sqrt{\frac{\mathrm{VC}(\Pi)}{\left|I\right|}}\right)\leq C_{2}\exp(-C_{3}\left|I\right|\delta^{2})
\]
holds for all $\delta>0$ and $\left|I\right|>C_{4}$, where $C_{1},\ldots,C_{4}>0$ are finite constants independent of $\delta$, $I$, and $\Pi$.
\end{assumption}

Assumption~\ref{assu:general ass uniform convergence} imposes a uniform convergence rate on the welfare estimator. It allows the VC dimension of the policy class, $\Pi$, to increase with the sample size at a rate that is slower than the sample size itself: $\mathrm{VC}(\Pi)/|I|\to0$. This rate condition aligns, up to logarithmic factors, with the minimax rate for policy learning under linear welfare criteria \citep{athey2021Policy,kitagawa2018Who}. Under this strengthened assumption, we derive an upper bound on the estimation error. 

\begin{cor}
\label{cor:estimation-error-bound-general}
Suppose Assumption~\ref{assu:general ass uniform convergence}
holds. Then there exist finite constants $C,C^{\prime}>0$ such that the learned policy $\widehat{\pi}$ defined in (\ref{eq:CV criterion-general}) satisfies, for sufficiently large $N$,
\begin{align*}
 & \e\left[W(\pi^{*})-W(\widehat{\pi})\right]\\
\leq & \inf_{\ell=1,2,\ldots}\left\{ \underbrace{W(\pi^{*})-\max_{\pi\in\Pi_{\ell}}W(\pi)}_{\mathrm{approximation\ error}}+C^{\prime}\sqrt{\frac{K}{K-1}}\sqrt{\frac{\mathrm{VC}(\Pi_{\ell})}{N}}+\frac{\log\ell}{\sqrt{N}}\right\} +\sqrt{\frac{C}{N}}.
\end{align*}
\end{cor}

To quantify the approximation error, however, we need to gather information about the policy space and its approximation. For illustrative purposes, we will compute the approximation error for three common policy classes and their approximations. Throughout this process, we will maintain the following Lipschitz condition on $W(\pi)$. 

\begin{assumption}
\label{ass:high-level}
There exists a finite constant $C_{W}>0$ such
that for any two policies $\pi_{1},\pi_{2}\in\Pi_{\infty}$, \(
\left|W(\pi_{1})-W(\pi_{2})\right|\le C_{W}P(\pi_{1}(\bs X)\neq\pi_{2}(\bs X))\).

\end{assumption}
\begin{example}[Monotone policies]
\label{exa:monotone policies}
In practice, shape restrictions on the policy class $\Pi_{\infty}$ capture constraints implied by economic theory or fairness considerations. A canonical example of this is monotonicity. After normalizing the supports, let $\bs X_i=(X_{i1},X_{i2})^{\trans}\in[0,1]^2$. The monotone policy class is written as:  
\(
\Pi_{\infty}=\left\{ \pi_f(x_1,x_2)=1\{x_2\le f(x_1)\}:f\text{ is non-increasing}\right\}
\).
We approximate this infinite-dimensional class using the monotone piecewise-linear sieve $\Pi_{\ell}$. Its elements are threshold rules defined by non-increasing, continuous, piecewise-linear boundaries with knots on a $2^{\ell}$-grid. Appendix~\ref{sec:approximation-errors} provides the formal definition, following the construction found in \citet[Example~3.2]{mbakop2021Model}. Under Assumption~\ref{ass:high-level}, if the conditional density of $X_1$ given $X_2$ is uniformly bounded, then $\VC(\Pi_{\ell})=O(2^{\ell})$ and the approximation error is $W(\pi^{*})-\max_{\pi\in\Pi_{\ell}}W(\pi)=O(2^{-\ell})$ (see Proposition~\ref{prop:monotone-approx-vc}).
\end{example}
In the next two examples, the sieve classes need not be subsets of $\Pi_{\infty}$. When these classes are used in the oracle bound, Assumptions~\ref{assu:general assu on hatW}--\ref{ass:high-level} are imposed on these sieve classes directly, so Theorem~\ref{thm:oracle inequality-general} continues to apply through the resulting approximation and estimation errors.
\begin{example}[Decision Trees]
\label{exa:Decision Trees}
Decision trees are another popular class of binary-valued policies \citep[e.g.,][]{athey2021Policy,zhou2023Offline}. One example policy space is a treatment rule represented as a smooth threshold in a single covariate, conditional on the remaining covariates. Specifically, with $\bs X=(\bs X_{-d},X_d)\in[0,1]^d$, the policy class is given by
\(
\Pi_{\infty}=\left\{ \pi_f(\bs x_{-d},x_d)=1\{x_d\le f(\bs x_{-d})\}:f\in C^{s}([0,1]^{d-1}),\ 0\le f\le1\right\}
\). 
Here $s\in\mathbb{N}_{+}$, and $C^{s}([0,1]^{d-1})$ denotes the class of functions with continuous partial derivatives up to order $s$. Let $\Pi_{\ell}$ be the class of binary decision trees of depth at most $\ell$. Each internal node selects a coordinate $j\in\{1,\ldots,d\}$ and a threshold $b\in\R$, routing observations according to whether $x_j<b$, with each leaf assigned a label in $\{0,1\}$. Under Assumption~\ref{ass:high-level}, if the conditional density of $X_d$ given $\bs X_{-d}$ is uniformly bounded, then $\VC(\Pi_{\ell})=O(2^{\ell}(\ell+\log d))$ and the approximation error is $W(\pi^{*})-\max_{\pi\in\Pi_{\ell}}W(\pi)=O\!\left(2^{-\left\lfloor(\ell-1)/(d-1)\right\rfloor}\right)$ (see Proposition~\ref{prop:dt-approx-vc}).
\end{example}
\begin{example}[Deep Neural Networks]
\label{exa:Neural Networks}
Deep neural networks (DNNs) provide a flexible framework for approximating complex decision boundaries. As in Example~\ref{exa:Decision Trees}, we also consider the smooth decision-boundary policy class, with $\bs X=(\bs X_{-d},X_d)\in[0,1]^d$, given by
\(
\Pi_{\infty}=\left\{ \pi_f(\bs x_{-d},x_d)=1\{x_d\le f(\bs x_{-d})\}:f\in C^{s}([0,1]^{d-1}),\ 0\le f\le1\right\}
\).
Let $\mathcal{F}_{\mathrm{DNN},\ell}$ denote the class of fully connected feedforward ReLU networks with width $\mathcal{H}_{\ell}$ and depth $\mathcal{D}_{\ell}$, where these architectural parameters grow with $\ell$. The sieve policy class is $\Pi_{\ell}=\{\pi_g(\bs x)=1\{g(\bs x)\ge0\}:g\in\mathcal{F}_{\mathrm{DNN},\ell}\}$. Under Assumption~\ref{ass:high-level}, if the conditional density of $X_d$ given $\bs X_{-d}$ is uniformly bounded, then Proposition~\ref{prop:nn-approx-vc} gives \(\VC(\Pi_{\ell})=O\{\mathcal{D}^{2}_{\ell}\mathcal{H}^{2}_{\ell}\allowbreak\log(\mathcal{D}_{\ell}\mathcal{H}^{2}_{\ell})\}\). The approximation error is \(W(\pi^{*})-\max_{\pi\in\Pi_{\ell}}W(\pi)\allowbreak=O\!\{(\mathcal{H}_{\ell}/\log\mathcal{H}_{\ell})^{-2s/(d-1)}\allowbreak(\mathcal{D}_{\ell}/\log\mathcal{D}_{\ell})^{-2s/(d-1)}\}\).
\end{example}

\section{Model}
\label{sec:general framework}

We now formally establish the model for the welfare criterion. We define the intermediate parameters $\bs{\beta}^{*}(\pi)\in\R^{p}$ as the minimizer of a sum of expected convex losses:
\[
\bs{\beta}^{*}(\pi):=\underset{\bs{\beta}=(\beta_{1},\ldots,\beta_{p})^{\trans}\in\R^{p}}{\arg\min}\sum_{j=1}^{p}\e\left[\mathcal{L}_{j}(Y^{*}(\pi(\bs X))-\beta_{j})\right],
\]
where $p\geq1$ is an integer and $\mathcal{L}_{1},\ldots,\mathcal{L}_{p}$ are known convex loss functions. Different selections of loss functions capture different distributional features. For example, when $p=2$, if we take $\mathcal{L}_{1}(v)=v^{2}/2$, and $\mathcal{L}_{2}(v)=v\cdot(0.5-1(v\leq0))$, the intermediate parameters correspond to the mean and the median: $\bs{\beta}^{*}(\pi)=\left(\e\left[Y^{*}(\pi(\bs X))\right],\text{median}[Y^{*}(\pi(\bs X))]\right)^{\trans}$. Similarly, with the check loss defined as $\mathcal{L}(v)=v\cdot(\alpha - 1(v\leq 0))$, the intermediate parameter corresponds to the $\alpha$-quantile.

Both the intermediate parameters and the welfare criterion are expressed in terms of potential outcomes. To reframe them using the observed data $(Y,\bs X,T)$, where $Y=Y^{*}(T)$ represents the observed outcome, we impose the Stable Unit Treatment Value Assumption (SUTVA) \citep{imbens2015Causal} and  the following conditions regarding the data-generating process \citep{athey2021Policy,kitagawa2018Who,mbakop2021Model}.

\begin{assumption}[Unconfoundedness and Overlap]\label{assu:unconfounded}
\begin{enumerate}[label=(\roman*)]
    \item \textbf{(Unconfoundedness)} Given the covariate $\bs X$, program participation $T$ is independent of the potential outcomes, meaning $Y^{*}(0),Y^{*}(1)\perp T\mid\bs X$.
    \item \textbf{(Overlap)} There exists a constant $0<\kappa<1/2$ such that $\kappa < e^{*}(\bs X) < 1-\kappa$ almost surely, where $e^{*}(\bs X):=\e\left[T\mid\bs X\right]$ is the propensity score.
\end{enumerate}
\end{assumption}

Under Assumption \ref{assu:unconfounded}, we can express the following optimization problem:
\begin{equation}
\bs{\beta}^{*}(\pi)=\underset{\bs{\beta}=(\beta_{1},\ldots,\beta_{p})^{\trans}\in\R^{p}}{\arg\min}\sum_{j=1}^{p}\e\left[\left\{ \frac{\pi(\bs X)T}{e^{*}(\bs X)}+\frac{\left(1-\pi(\bs X)\right)\left(1-T\right)}{1-e^{*}(\bs X)}\right\} \mathcal{L}_{j}(Y-\beta_{j})\right].\label{eq:def-of-beta-pi}
\end{equation}
Additionally, we define: 
\begin{align}
W(\pi)=\e\left[\left\{ \frac{\pi(\bs X)T}{e^{*}(\bs X)}+\frac{\left(1-\pi(\bs X)\right)\left(1-T\right)}{1-e^{*}(\bs X)}\right\} U(Y,\bs X,\bs{\beta}^{*}(\pi))\right].\label{eq:identification-of-W-pi}
\end{align}
The propensity score $e^{*}(\bs X)$ is unknown and is defined by $e^{*}(\bs X)=\e [T\mid\bs X]$, which can be found by solving the following optimization problem: 
\(
e^{*}(\bs X)=\underset{e(\cdot)}{\arg\min}\left\{ \e\left[(T-e(\bs{X}))^2\right]\right\}
\). 
However, a sample least-squares estimator based on this characterization may produce fitted values close to 0 or 1. We therefore use the following equivalent population formulation, understood over functions satisfying $0<e(\bs X)<1$: 
\begin{align}
e^{*}(\bs X) & =\underset{0<e(\cdot)<1}{\arg\min}\Biggl\{ \e\left[T\left(\frac{1}{e(\bs X)}-\frac{1}{e^{*}(\bs X)}\right)^{2}+(1-T)\left(\frac{1}{1-e(\bs X)}-\frac{1}{1-e^{*}(\bs X)}\right)^{2}\right]\Biggr\} \nonumber \\
  & =\underset{0<e(\cdot)<1}{\arg\min}\left\{ \e\left[\frac{T}{e(\bs X)^{2}}-\frac{2}{e(\bs X)}\right]+\e\left[\frac{1-T}{(1-e(\bs X))^{2}}-\frac{2}{1-e(\bs X)}\right]\right\}. \label{eq:propensity-objective}
\end{align}
Compared with the least-squares characterization, this criterion measures errors on the inverse-propensity scale that enters the inverse-probability-weighted (IPW) terms. Its population excess risk is a weighted sum of squared errors of $1/e(\bs X)$ and $1/\{1-e(\bs X)\}$, with the same minimizer $e^{*}$; in estimation, we minimize its sample analog over a bounded logistic DNN class to keep fitted propensity scores away from $0$ and $1$.

\section{Estimation of Welfare}\label{sec:empirical welfare}

We present an estimator for the welfare criterion based on a generic training sample $I$. Equations (\ref{eq:def-of-beta-pi})--(\ref{eq:propensity-objective}) suggest a three-step sequential estimation procedure. In the first step, we estimate the propensity score using a sample analog of (\ref{eq:propensity-objective}). We then substitute this estimate into a sample analog of (\ref{eq:def-of-beta-pi}) to estimate the intermediate parameters. Finally, we use both estimates to compute the welfare criterion using a sample analog of (\ref{eq:identification-of-W-pi}).

To estimate the propensity score from the training sample $I$, we utilize a deep neural network, denoting the estimate as $\widehat{e}_{I}(\cdot)$ (see (\ref{eq:logistic reg for propensity})). From the existing literature, it follows that $\left\Vert \widehat{e}_{I}-e^{*}\right\Vert_{P,2}= O_{P}\left(\left|I\right|^{-s_{e}/(2s_{e}+d)}\log^{3}\left|I\right|\right)$ (see the proof of Lemma~\ref{lem:converge of propensity}), where $s_{e}$ indicates the smoothness of $e^{*}(\bs X)$ (see Assumption~\ref{assu:Overlap assumption and smoothness}). 

It is well-documented that machine learning can induce bias, which then propagates to the welfare criterion through both direct bias in the IPW terms and indirect bias in estimates of intermediate parameters. Simply substituting $e^{*}$ with $\widehat{e}_{I}$ in the sample analog of equations (\ref{eq:def-of-beta-pi})--(\ref{eq:propensity-objective}) may introduce bias in the intermediate parameters and welfare estimates. This, in turn, leads to a violation of Assumption~\ref{assu:general assu on hatW}. To mitigate machine learning bias, we propose a weighted analog of equations (\ref{eq:def-of-beta-pi})--(\ref{eq:propensity-objective}): 
\begin{equation}
\widehat{\bs{\beta}}_{I}(\pi)=\underset{\bs{\beta}=(\beta_{1},\ldots,\beta_{p})^{\trans}\in\R^{p}}{\arg\min}\sum_{j=1}^{p}\frac{1}{\left|I\right|}\sum_{i\in I}\widehat{w}_{I,i}(\pi)\Bigl\{\frac{\pi(\bs X_{i})T_{i}}{\widehat{e}_{I}(\bs X_{i})}+\frac{(1-\pi(\bs X_{i}))(1-T_{i})}{1-\widehat{e}_{I}(\bs X_{i})}\Bigr\}\mathcal{L}_{j}(Y_{i}-\beta_{j}).\label{eq:def-of-beta-hat-pi-I-unknown-propensity}
\end{equation}
The welfare estimator is defined as: 
\begin{equation}
\widehat{W}_{I}(\pi)=\frac{1}{\left|I\right|}\sum_{i\in I}\widehat{w}_{I,i}(\pi)\Bigl\{\frac{\pi(\bs X_{i})T_{i}}{\widehat{e}_{I}(\bs X_{i})}+\frac{(1-\pi(\bs X_{i}))(1-T_{i})}{1-\widehat{e}_{I}(\bs X_{i})}\Bigr\} U(Y_{i},\bs X_{i},\widehat{\bs{\beta}}_{I}(\pi)).\label{eq:def-of-W-hat-pi-I-unknown-propensity}
\end{equation}
In this context, $\left\{ \widehat{w}_{I,i}(\pi):i\in I\right\}$ represents the calibrated weights. We show in Appendix~\ref{sec:app-nuisance-estimation} that these weights must satisfy the following conditions: 
\begin{equation}
\begin{aligned}
 & \frac{1}{\left|I\right|}\sum_{i\in I}(1-\pi(\bs X_{i}))\left\{ \frac{w_{i}(1-T_{i})}{1-\widehat{e}_{I}(\bs X_{i})}-1\right\} \mu_{j0}^{*}(\bs X_{i};\bs{\beta}^{*}(\pi))\\
 & +\frac{1}{\left|I\right|}\sum_{i\in I}\pi(\bs X_{i})\left\{ \frac{w_{i}T_{i}}{\widehat{e}_{I}(\bs X_{i})}-1\right\} \mu_{j1}^{*}(\bs X_{i};\bs{\beta}^{*}(\pi))=0,\quad j=0,1,\ldots,p.
\end{aligned}
\label{eq:calibrated-weight-constraints}
\end{equation}
For $t\in\{0,1\}$, we define: 
$\mu_{0t}^{*}(\bs x;\bs{\beta}):=\allowbreak\e\bigl[U(Y,\bs X,\bs{\beta})\allowbreak\mid\allowbreak\bs X=\bs x,T=t\bigr]$
and $\mu_{jt}^{*}(\bs x;\bs{\beta}):=\allowbreak\e\bigl[\mathcal{L}_{j}^{\prime}(Y-\beta_{j})\allowbreak\mid\allowbreak\bs X=\bs x,T=t\bigr]$
for $j=1,\ldots,p$, where $\mathcal{L}_{j}^{\prime}$ is understood as specified in Assumption~\ref{assu:regularity-conditions-on-L}.
In practice, $\bs{\beta}^{*}(\pi)$ and $\mu_{jt}^{*}(\cdot;\bs{\beta})$ are replaced with the initial estimator $\widehat{\bs{\beta}}_{I}^{\mathrm{init}}(\pi)$
and the conditional mean estimators $\widehat{\mu}_{I,jt}(\cdot;\widehat{\bs{\beta}}_{I}^{\mathrm{init}}(\pi))$, respectively,
as detailed in Appendix~\ref{sec:app-nuisance-estimation}.

Notice that the weights satisfying the equations (\ref{eq:calibrated-weight-constraints}) are generally not unique. We apply the entropy method to calibrate the weights as the solution to
\begin{equation}
\left(\widehat{w}_{I,i}(\pi):i\in I\right)=\underset{w_{i}>0:i\in I}{\arg\min}\sum_{i\in I}\left(w_{i}\log w_{i}-w_{i}\right)\text{ subject to }(w_{i}:i\in I)\text{ satisfying }(\ref{eq:calibrated-weight-constraints}).\label{eq:balancing for weights}
\end{equation}
In this context, the objective function $D(w)=w\log w-w$ measures the distance of $w$ from $1$, ensuring that the calibrated weights $\widehat{w}_{I,i}(\pi)$, $i\in I$, are unique and always non-negative. 

\begin{rem}
    From a computational perspective, the equation (\ref{eq:balancing for weights}) is a convex program with $p+1$ linear constraints, making it straightforward to compute. In particular, the problem can be solved efficiently through its dual formulation using standard convex-optimization solvers, as referenced in \citep{boyd2004Convex}.
    \end{rem}

\section{Properties of the Empirical Welfare Criterion}\label{sec:asymptotic properties} 

Having constructed the welfare criterion estimator, we will now verify that it meets the high-level condition outlined in Section~\ref{sec:A-General-Learning}. We require the following conditions. 

\begin{assumption}
\label{assu:Overlap assumption and smoothness} Assume that $\mathcal{X}=[0,1]^{d}$. Let $M>1$ be a finite constant. The true propensity score $e^{*}(\bs x)$ satisfies the following condition: 
\[
\log\frac{e^{*}(\bs x)}{1-e^{*}(\bs x)}\in C^{s_{e}}\left([0,1]^{d}\right):=\left\{ f:\max_{\bs{\alpha}\in\N^{d},\left\Vert \bs{\alpha}\right\Vert _{1}\leq s_{e}}\underset{\bs x\in[0,1]^{d}}{\sup}\left|\partial^{\bs{\alpha}}f(\bs x)\right|\leq M\right\} ,
\]
where $s_{e}>d/2$ is an integer, $\left\Vert \bs{\alpha}\right\Vert _{1}:=\alpha_{1}+\cdots+\alpha_{d}$ and $\partial^{\bs{\alpha}}f$ represents the partial derivative of $f$.
\end{assumption}
\begin{assumption}
[Regularity assumptions on $\mathcal{L}_{j}$]\label{assu:regularity-conditions-on-L} For
each $j=1,\ldots,p$, let $\beta_{j}^{*}(\pi)$ be the $j$th component of $\bs{\beta}^{*}(\pi)$. Let $c_{0}>0$ be a finite constant. The following conditions hold for any $j=1,\ldots,p$.
\begin{enumerate}[label=(\roman*)]
\item $\mathcal{L}_{j}(v)$ is convex on $\R$. There exists a non-decreasing function $\mathcal{L}_{j}^{\prime}(v):\R\to\R$ such that $\int_{a}^{b}\mathcal{L}_{j}^{\prime}(v)dv=\mathcal{L}_{j}(b)-\mathcal{L}_{j}(a)$ for any $a,b\in\R$. We refer to $\mathcal{L}_{j}^{\prime}(v)$ as the ``derivative'' of $\mathcal{L}_{j}(v)$.
\item Define
\(
Q_{j}(\beta;\pi):=\e\Bigl[\left\{ \frac{\pi(\bs X)T}{e^{*}(\bs X)}+\frac{\left(1-\pi(\bs X)\right)\left(1-T\right)}{1-e^{*}(\bs X)}\right\} \mathcal{L}_{j}(Y-\beta)\Bigr]
\). 
Let $\underline{Q}^{\prime\prime}$ and $Q_{lip}^{\prime\prime}$ be two positive and finite constants. $Q_{j}(\beta;\pi)$ is twice differentiable with respect to $\beta$, and we denote its second-order derivative by $Q_{j}^{\prime\prime}(\beta;\pi)$. It holds that 
$Q_{j}^{\prime\prime}(\beta_{j}^{*}(\pi);\pi)\geq\underline{Q}^{\prime\prime}$ uniformly over $\pi\in\Pi_{\infty}$ and 
$\left|Q_{j}^{\prime\prime}(\beta;\pi)-Q_{j}^{\prime\prime}(\beta_{j}^{*}(\pi);\pi)\right|\leq Q_{lip}^{\prime\prime}\cdot\left|\beta-\beta_{j}^{*}(\pi)\right|$ for all $\beta$ satisfying $\left|\beta-\beta_{j}^{*}(\pi)\right|\leq c_{0}$
and all $\pi\in\Pi_{\infty}$.
\item $\sup_{\pi\in\Pi_{\infty}}\sup_{\beta:\left|\beta-\beta_{j}^{*}(\pi)\right|\leq c_{0}}\left|\mathcal{L}_{j}^{\prime}(Y-\beta)\right|\leq M/4$ almost surely, where $0<M<\infty$ is a constant that may depend on $\underline{Q}^{\prime\prime}$ and $Q_{lip}^{\prime\prime}$.
\end{enumerate}
\end{assumption}
\begin{assumption}
[Regularity assumptions on $U$]\label{assu:regularity-assumptions-on-U} Denote $\Psi(\bs{\beta};\pi)$ by
\(
\e\bigl[\bigl\{ \pi(\bs X)T/e^{*}(\bs X)+\allowbreak(1-\pi(\bs X))(1-T)/(1-e^{*}(\bs X))\bigr\} \allowbreak U(Y,\bs X,\bs{\beta})\bigr]
\). 
The following conditions hold.
\begin{enumerate}[label=(\roman*)]
\item There exists a finite constant $M>0$ such that the bound
$\sup_{\pi\in\Pi_{\infty}}\sup_{\bs{\beta}:\left\Vert \bs{\beta}-\bs{\beta}^{*}(\pi)\right\Vert \leq c_{0}}\allowbreak\left|U(Y,\bs X,\bs{\beta})\right|\leq M/4$
holds almost surely.
\item The function class $\mathcal{U}:=\{(Y,\bs X)\mapsto U(Y,\bs X,\bs{\beta}):\allowbreak\left\Vert \bs{\beta}-\bs{\beta}^{*}(\pi)\right\Vert \leq c_{0},\allowbreak\pi\in\Pi_{\infty}\}$
satisfies $\sup_{Q}\log N\left(\frac{M}{4}\epsilon,\mathcal{U},\left\Vert \cdot\right\Vert _{Q,2}\right)\leq\nu\log\left(a/\epsilon\right)\text{ for all }0<\epsilon<1,$
where $a,\nu>0$ are finite constants, and $\sup_{Q}$ is taken over
all finitely discrete measures.
\item Given any policy $\pi\in\Pi_{\infty}$, $\Psi(\bs{\beta};\pi)$ is differentiable
with respect to $\bs{\beta}$, and we denote its gradient by $\nabla\Psi(\bs{\beta};\pi)$.
There exists a finite constant $\overline{\Psi}^{\prime}\geq0$ such that $\left\Vert \nabla\Psi(\bs{\beta};\pi)\right\Vert \leq\overline{\Psi}^{\prime}$
for all $\bs{\beta}$ satisfying $\left\Vert \bs{\beta}-\bs{\beta}^{*}(\pi)\right\Vert \leq c_{0}$ and all $\pi\in\Pi_{\infty}$.
\end{enumerate}
\end{assumption}
\begin{assumption}
\label{assu:smoothness for mu} Let $s_{\mu}>d/2$ be an integer. It
holds that 
\begin{align*}
 & \left\{ \bs X\mapsto\mu^{*}_{jt}(\bs X;\bs{\beta}):\left\Vert \bs{\beta}-\bs{\beta}^{*}(\pi)\right\Vert \leq c_{0},\pi\in\Pi_{\infty},t\in\{0,1\},j=0,\ldots,p\right\} \\
\subset & \  C^{s_{\mu}}\left([0,1]^{d}\right):=\left\{ f:\max_{\bs{\alpha}\in\N^{d},\left\Vert \bs{\alpha}\right\Vert _{1}\leq s_{\mu}}\underset{\bs x\in[0,1]^{d}}{\sup}\left|\partial^{\bs{\alpha}}f(\bs x)\right|\leq M\right\} 
\end{align*}
for some constant $M>0$.
\end{assumption}
\begin{assumption}
\label{assu:regularity on mu} There exist finite constants $L_{\mu}>0$ and
$c_{\xi}>0$ such that the following conditions hold.
\begin{enumerate}[label=(\roman*)]
\item     For any $\bs{\beta}_{1},\bs{\beta}_{2}\in\mathbb{R}^{p}$, $j=0,\ldots,p$, and $t\in\{0,1\}$: $\left\Vert \mu_{jt}^{*}(\bs X;\bs{\beta}_{1})-\mu_{jt}^{*}(\bs X;\bs{\beta}_{2})\right\Vert _{P,2}\leq L_{\mu}\left\Vert \bs{\beta}_{1}-\bs{\beta}_{2}\right\Vert$.
\item For $t\in\{0,1\}$, define $\bs{\xi}_{t}^{*}(\bs X;\pi):=\left(\mu_{jt}^{*}(\bs X;\bs{\beta}^{*}(\pi)):j=0,\ldots,p\right)^{\trans}$. Then
\[
\inf_{\pi\in\Pi_{\infty}}\lambda_{\min}\Bigl\{\e\bigl[(1-\pi(\bs X))\bs{\xi}_{0}^{*}(\bs X;\pi)\bs{\xi}_{0}^{*}(\bs X;\pi)^{\trans}+\pi(\bs X)\bs{\xi}_{1}^{*}(\bs X;\pi)\bs{\xi}_{1}^{*}(\bs X;\pi)^{\trans}\bigr]\Bigr\}\geq c_{\xi}.
\]
\end{enumerate}
\end{assumption}
\begin{assumption}
\label{assu:regularity-assumptions-on-nuisance} The training sample size satisfies $|I|\to\infty$ as $N\to\infty$.
Throughout this assumption, $\widehat{e}_{I}$, $\widehat{\bs{\beta}}_{I}^{\mathrm{init}}(\pi)$, and $\widehat{\mu}_{I,jt}$, $j=0,\ldots,p$ and $t\in\{0,1\}$, refer to the estimators defined in the respective equations (\ref{eq:logistic reg for propensity}), (\ref{eq:beta-init-pi-I}), and (\ref{eq:DNN reg for L})--(\ref{eq:DNN reg for U}).
\begin{enumerate}[label=(\roman*)]
\item Almost surely, $\left|\log\{\widehat{e}_{I}(\bs X)/(1-\widehat{e}_{I}(\bs X))\}\right|\leq M$, and $\left|\widehat{\mu}_{I,jt}(\bs X;\bs{\beta})\right|\leq M$ uniformly over $\pi\in\Pi_{\infty}$,
$\left\Vert \bs{\beta}-\bs{\beta}^{*}(\pi)\right\Vert \leq c_{0}$, $j=0,\ldots,p$, and $t=0,1$, where $c_{0}$ is from Assumption~\ref{assu:regularity-conditions-on-L}.
\item The DNN classes $\mathcal{F}_{\mathrm{DNN}}(\mathcal{H}_{e},\mathcal{D}_{e})$ and $\mathcal{F}_{\mathrm{DNN}}(\mathcal{H}_{\mu},\mathcal{D}_{\mu})$
defined in (\ref{eq:DNN class}) and used in (\ref{eq:DNN reg for L})--(\ref{eq:DNN reg for U}) are constructed as follows. The first uses $\mathcal{H}_{e}\mathcal{D}_{e}\asymp\left|I\right|^{d/(4s_{e}+2d)}(\log\left|I\right|)^{2}$, and the second uses $\mathcal{H}_{\mu}\mathcal{D}_{\mu}\asymp\left|I\right|^{d/(4s_{\mu}+2d)}(\log\left|I\right|)^{2}$. In addition, their minimum diverges to infinity, and their logarithms are $O(\log |I|)$.
\end{enumerate}
\end{assumption}
Assumption~\ref{assu:Overlap assumption and smoothness} is a smoothness condition that is commonly recognized in the literature on deep neural network estimation \citep{farrell2021Deepa,jiao2023Deep,schmidt-hieber2020Nonparametric}, as well as in the broader nonparametric-estimation literature \citep{chen2007Large}. The conditions outlined in Assumption~\ref{assu:regularity-conditions-on-L} are for estimating $\bs{\beta}^{*}(\pi)$, allowing for non-smooth objectives such as $\mathcal{L}_{j}(v)=v(0.5-1(v\leq0))$. Similarly, the conditions in Assumption~\ref{assu:regularity-assumptions-on-U} are for estimating $W(\pi)$ and are satisfied by many utility functions $U$. Both assumptions are well-established in the literature \citep{vaart1998Asymptotic,vaart1996Weak}.

It is important to note that Assumptions~\ref{assu:regularity-conditions-on-L} and~\ref{assu:regularity-assumptions-on-U} hold uniformly over $\pi\in\Pi_{\infty}$. The uniform condition is critical for Assumption~\ref{assu:general ass uniform convergence}, which requires uniform convergence of $\widehat{W}_{I}(\pi)$. If only Assumption~\ref{assu:general assu on hatW} is necessary, a pointwise-in-$\pi$ version of these conditions is sufficient. Assumption~\ref{assu:smoothness for mu} is a smoothness condition on the conditional mean function $\mu_{jt}^{*}(\bs X_{i};\bs{\beta})$, which is analogous to Assumption~\ref{assu:Overlap assumption and smoothness}. Assumption~\ref{assu:regularity on mu} imposes Lipschitz continuity on $\mu_{jt}^{*}(\bs X_{i};\bs{\beta})$ in relation to $\bs{\beta}$ and includes a population non-singularity condition involving $\mu_{jt}^{*}(\bs X;\bs{\beta}^{*}(\pi))$, $j=0,\ldots,p$ and $t\in\{0,1\}$. The Lipschitz condition is satisfied by both $\mathcal{L}_{j}(v)=v(0.5-1(v\leq0))$ and $\mathcal{L}_{j}(v)=v^{2}/2$. The non-singularity condition eliminates linear redundancy among the functions $\mu_{jt}^{*}(\bs X;\bs{\beta}^{*}(\pi))$; if this condition fails, redundant functions can be removed before applying the calibration step. Assumption~\ref{assu:regularity-assumptions-on-nuisance}(i) imposes boundedness on the nuisance estimates, while Assumption~\ref{assu:regularity-assumptions-on-nuisance}(ii) restricts the width and depth of the deep neural networks. This latter requirement is familiar in the deep neural network estimation literature \citep[see, e.g.,][]{farrell2021Deepa,schmidt-hieber2020Nonparametric,jiao2023Deep}.

Under these sufficient conditions, we show that the proposed welfare criterion estimator meets the high-level assumptions outlined in Section~\ref{sec:A-General-Learning}. 

\begin{thm}
\label{thm:oracle holdout unknown propensity} Suppose that Assumptions
\ref{assu:unconfounded}, \ref{assu:Overlap assumption and smoothness},
\ref{assu:regularity-conditions-on-L},
\ref{assu:regularity-assumptions-on-U},
\ref{assu:smoothness for mu}, \ref{assu:regularity on mu}, and
\ref{assu:regularity-assumptions-on-nuisance}
hold. Additionally, suppose that $\left|\widehat{W}_{I}(\pi)\right|\leq C$ almost surely for
any $\pi\in\Pi_{\infty}$ and $I\subset\{1,\ldots,N\}$, where $C>0$
is a finite constant. Under these conditions, the debiased welfare criterion estimator
$\widehat{W}_{I}(\pi)$, defined in (\ref{eq:def-of-W-hat-pi-I-unknown-propensity}),
satisfies Assumptions~\ref{assu:general assu on hatW} and \ref{assu:general ass uniform convergence}.
Furthermore, the welfare function $W(\pi)$, defined in (\ref{eq:identification-of-W-pi}),
satisfies Assumption~\ref{ass:high-level}. 
\end{thm}
Combined with Theorem~\ref{thm:oracle inequality-general}, Theorem~\ref{thm:oracle holdout unknown propensity} implies that the average welfare regret of the proposed policy-learning procedure adheres to the oracle inequality stated in~(\ref{eq:oracle inequality}); combined with Corollary~\ref{cor:oracle inequality prob bound}, it also yields the corresponding high-probability welfare regret bound.

\section{Empirical Application} \label{sec:empirical application}

To illustrate the practical value of the proposed policy learning procedure, we apply it to data from the National Job Training Partnership Act (JTPA) Study. This large-scale randomized controlled trial was commissioned by the U.S. Department of Labor to evaluate the effectiveness of publicly funded job-training programs. This dataset has become a benchmark in the policy evaluation and policy learning literature \citep{crippa2025Regret,ai2026data,abadie2002Instrumental,liu2025nonparametric, kitagawa2018Who, mbakop2021Model}.\footnote{The sample we use is taken from the supplementary materials of \cite{mbakop2021Model}, available at \url{https://onlinelibrary.wiley.com/doi/10.3982/ECTA16437}.}

Our analysis uses a sample of $N=11{,}008$ individuals. For each individual, we observe two baseline covariates: years of education ($X_{1}$) and pre-program earnings ($X_{2}$). The outcome of interest, $Y_i$, is the total earnings over the 30-month period following random assignment. Let $T\in\{0,1\}$ denote the randomized treatment assignment (the training offer), which means that the propensity score $e^{*}(\bs X)=P(T=1\mid \bs X)$ is constant and equal to $2/3$. While the true propensity score is known in this sample, we deliberately treat it as unknown to demonstrate the applicability of the method in situations where assignment probabilities are unavailable, partially observed, or require estimation.

We focus on a specific class of monotone, interpretable allocation rules, guided by the principle that, \emph{all else being equal}, individuals with lower socioeconomic status (such as less education or lower earnings) should be (weakly) prioritized for training. Let $\mathcal X_1$ and $\mathcal X_2$ represent the supports of education and pre-program earnings, respectively. We define the policy space as: 
\[
\Pi_{\infty}
=
\left\{
\pi: \mathcal{X}_1 \times \mathcal{X}_2 \to \{0,1\}
:
\pi(x_1, x_2) = 1\!\left(f(x_1) \geq x_2\right)
\text{ for some non-increasing } f
\right\}.
\]
The policy $\pi(x_1,x_2)=1\{x_2\le f(x_1)\}$ assigns an individual to training whenever their pre-program earnings fall below an education-specific cutoff $f(x_1)$. The restriction that $f$ is non-increasing ensures that the cutoff is (weakly) higher for individuals with less education, making the earnings criterion more lenient for them. This rule is transparent: it can be represented as a treatment region in the $(x_1,x_2)$-plane or, equivalently, as an estimated cutoff curve $\widehat f(x_1)$.

Previous studies (e.g., \citealp{mbakop2021Model}) have optimized this class of rules using a linear welfare criterion that maximizes average outcomes $\e[Y^{*}(\pi(\bs{X}))]$. However, such an objective neglects distributional concerns: a policy designed to maximize average income may inadvertently increase income disparities. To address this trade-off between \textit{efficiency} (aggregate income) and \textit{equity} (income dispersion), we adopt a nonlinear welfare criterion that penalizes outcome dispersion. Specifically, we aim to maximize the ratio of the mean outcome to its standard deviation, known as the inverse coefficient of variation:
\begin{equation}
    W_{\mathrm{ICV}}(\pi) = \frac{\e\left[Y^{*}(\pi(\bs{X}))\right]}{\sqrt{\Var\left(Y^{*}(\pi(\bs{X}))\right)}},
\end{equation}
where $\Var\left(Y^{*}(\pi(\bs{X}))\right) = \e\left[Y^{*}(\pi(\bs{X}))^2\right] - \left(\e\left[Y^{*}(\pi(\bs{X}))\right]\right)^2$. This objective is rooted in the axiomatic literature on inequality measurement (e.g., \citealp{atkinson1970Measurement}), which emphasizes that social welfare evaluations should balance efficiency (mean outcomes) against equity (distributional fairness). Maximizing $W_{\mathrm{ICV}}(\pi)$ is equivalent to minimizing the coefficient of variation, a scale-invariant measure of inequality that penalizes dispersion relative to the mean.

To reformulate this objective within our framework, we express it using the auxiliary parameters $\bs{\beta}^{*}(\pi)$ and the utility function $U(\cdot)$ introduced in Section~\ref{sec:general framework}. Since earnings are non-negative in our sample and the policy mean is positive for the policies considered here, maximizing $W_{\mathrm{ICV}}(\pi)$ is equivalent to maximizing its square. Simple algebra shows that maximizing $W_{\mathrm{ICV}}(\pi)^2$ is equivalent to maximizing the negative ratio of the second moment to the squared first moment:
\begin{equation} \label{eq:obj_transformed}
    \pi^* = \underset{\pi \in \Pi_{\infty}}{\arg\max}
    \left\{
    - \frac{\e\left[Y^*(\pi(\bs X))^2\right]}{\left(\e\left[Y^*(\pi(\bs X))\right]\right)^2}
    \right\}.
\end{equation}
This problem fits directly into our general framework, with a single auxiliary parameter ($p=1$). We define $\beta^{*}_1(\pi)$ to be the population mean of the potential outcome, corresponding to the quadratic loss $\mathcal{L}_1(v) = v^2/2$:
\[
    \beta_1^*(\pi) := \underset{\beta \in \R}{\arg\min}\ \e\left[ \frac{1}{2}\left( Y^*(\pi(\bs X)) - \beta \right)^2 \right]
    = \e\left[ Y^*(\pi(\bs X)) \right].
\]
The utility function is given by \(U(y, \bs x, \beta_1) := -y^2/\beta_1^2\).
With a slight abuse of notation, we denote the welfare criterion as:   
\[
    W(\pi)=\e\left[ U(Y^*(\pi(\bs X)), \bs X, \beta_1^*(\pi)) \right]
    = - \frac{\e[Y^*(\pi(\bs X))^2]}{(\e[Y^*(\pi(\bs X))])^2}.
\]
We learn the optimal policy $\pi^*$ from observational data by maximizing $\widehat{W}_{I}(\pi)$ as defined in equation (\ref{eq:def-of-W-hat-pi-I-unknown-propensity}) over the sieve approximating sequence described in Example~\ref{exa:monotone policies}. We use a $5$-fold cross-validation procedure (\ref{eq:CV criterion-general}) to select the best subclass.\footnote{Following \cite{mbakop2021Model}, this application contains only five candidate subclasses, $\Pi_1,\ldots,\Pi_{5}$.
}
To determine the best policy within each policy subclass, we employ the Strategic Monte Carlo Optimization (SMCO) algorithm as outlined in \citet{chen2026Optimization}. This algorithm demonstrates that, under suitable conditions, it converges to a local optimum from a single starting point and to a global optimum as the number of starting points increases. Consequently, we run SMCO from multiple starting points that are generated quasi-uniformly over a unit hypercube. This approach provides space-filling exploration of the parameter domain and enhances the robustness of the nonconvex search.\footnote{Following the subclass construction in \citet{mbakop2021Model}, the $\ell$th subclass can be indexed, after an appropriate reparametrization, by a vector $\bs{\theta}=(\theta_1,\ldots,\theta_{2^{\ell-1}+1})$ lying in a simplex-type set: $\theta_j\ge 0$ and $\sum_{j=1}^{2^{\ell-1}+1}\theta_j \le 1$. We construct an explicit mapping from the $(2^{\ell-1}+1)$-dimensional unit hypercube onto this set and run SMCO over the hypercube, which lets us generate quasi-uniform starting points while enforcing the simplex constraint by construction.}
 
 The welfare estimate $\widehat W_I(\pi)$ as described in equation (\ref{eq:def-of-W-hat-pi-I-unknown-propensity}) relies on several nuisance components, including the estimated propensity score $\widehat e_I(\cdot)$ and the estimated conditional mean functions $\widehat \mu_{I,jt}(\cdot;\bs\beta)$. We compute these estimates using deep neural networks, as detailed in Appendix~\ref{sec:app-nuisance-estimation}. The architecture of the network, including its depth and width, is determined through cross-validation on the same training sample used to train $\widehat W_I(\pi)$.

Figure~\ref{fig:best policy in Pi1 and Pi5} displays the best policies found in the simplest ($\Pi_{1}$) and the most complex ($\Pi_{5}$) subclasses of the approximating sequence.

\begin{figure}[h]
    \includegraphics[width=0.5\textwidth]{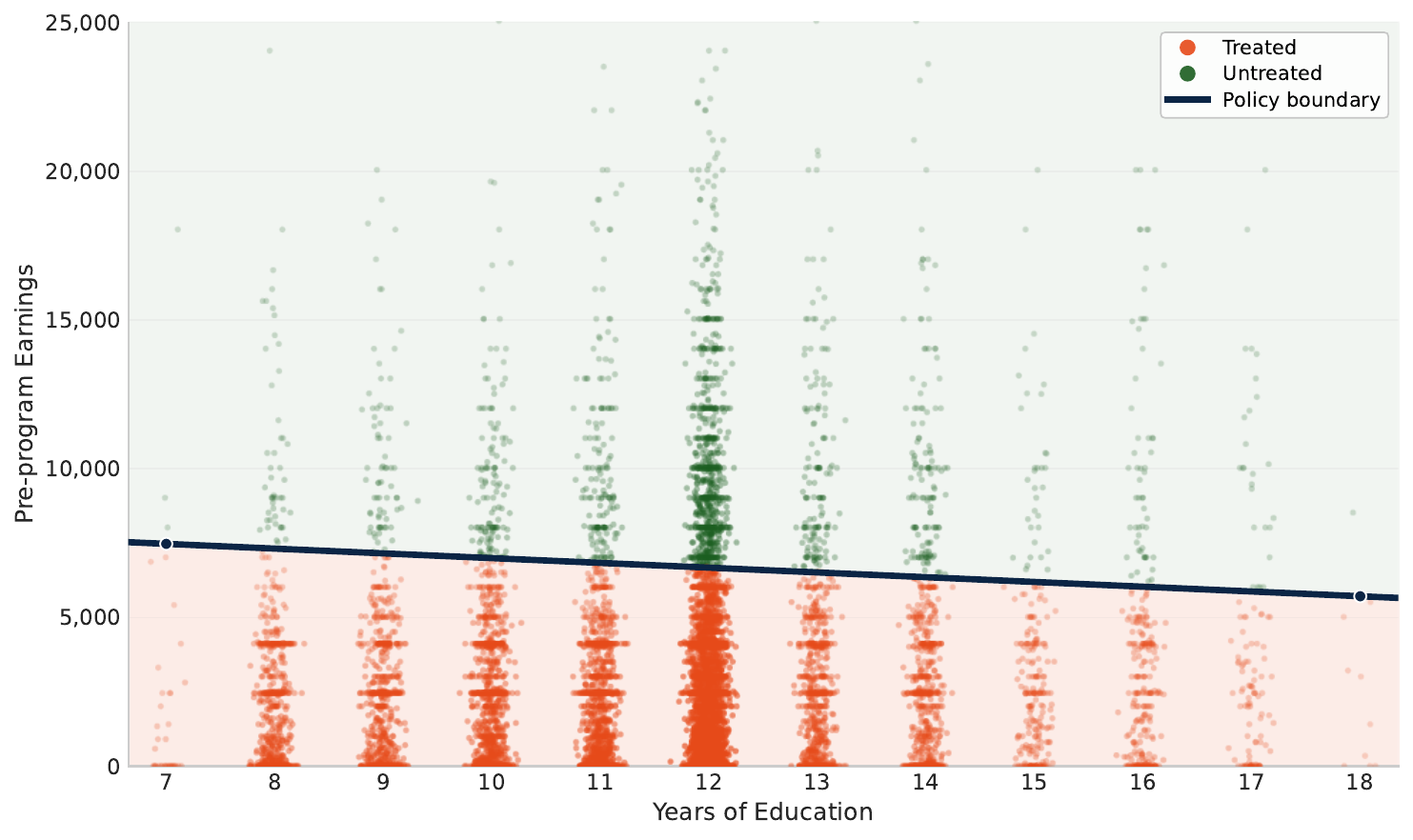}\includegraphics[width=0.5\textwidth]{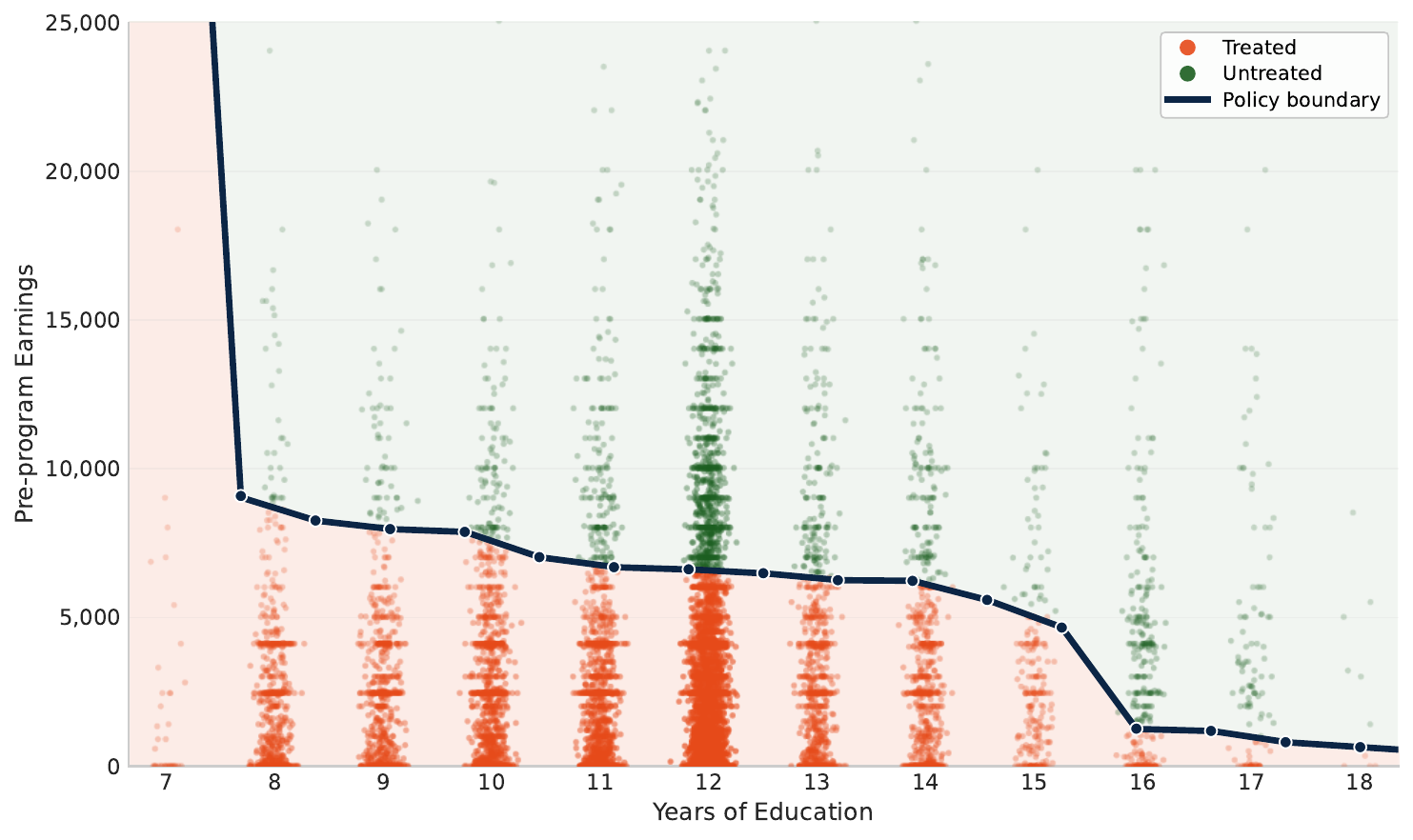}
    
\caption{\label{fig:best policy in Pi1 and Pi5}The best policy found in the
    simplest ($\Pi_{1}$) and most complicated ($\Pi_{5}$) classes. The $x$-axis represents years of education, while the $y$-axis indicates pre-program earnings. The red and green shaded areas represent individuals assigned to the treatment and control groups, respectively, under the estimated optimal policy. Left
    panel: best policy learned within $\Pi_{1}$; right panel: best policy
    learned within $\Pi_{5}$.}
    
    \end{figure}

    Our $5$-fold cross-validation procedure identifies $\Pi_{1}$ as the best subclass, and we denote the learned optimal policy as $\widehat{\pi}_{\mathrm{Nonlin}}$. We compare $\widehat{\pi}_{\mathrm{Nonlin}}$ with the benchmark policy $\widehat{\pi}_{\mathrm{Lin}}$, which was derived from penalized welfare maximization under the linear criterion $\e[Y^{*}(\pi(\bs{X}))]$ as discussed in \cite{mbakop2021Model}. To evaluate each policy, we compute the mean and standard deviation of its associated potential outcomes based on the full sample, using a de-biased estimator as described in Section~\ref{sec:empirical welfare}. For a given policy $\widehat{\pi}$, these estimators are defined as follows: 
    \[
    \widehat{\mathrm{Mean}}(\widehat{\pi}) = \frac{1}{N}\sum_{i=1}^{N}\widehat{w}_{i}(\widehat{\pi})\left\{\frac{\widehat{\pi}(\bs X_{i})T_{i}}{\widehat{e}(\bs X_{i})}+\frac{(1-\widehat{\pi}(\bs X_{i}))(1-T_{i})}{1-\widehat{e}(\bs X_{i})}\right\} Y_{i},
    \]
    and
    \[
    \widehat{\mathrm{SD}}(\widehat{\pi}) = \sqrt{\frac{1}{N}\sum_{i=1}^{N}\widehat{w}_{i}(\widehat{\pi})\left\{\frac{\widehat{\pi}(\bs X_{i})T_{i}}{\widehat{e}(\bs X_{i})}+\frac{(1-\widehat{\pi}(\bs X_{i}))(1-T_{i})}{1-\widehat{e}(\bs X_{i})}\right\} Y_{i}^{2} - \left(\widehat{\mathrm{Mean}}(\widehat{\pi})\right)^{2}},
    \]
    where $\widehat{e}(\bs X)$ is the DNN estimate of the propensity score using the full sample. The weights $\widehat{w}_{i}(\widehat{\pi})$ are determined by the following optimization problem: 
    \[
    \begin{cases}
        (\widehat{w}_{1}(\widehat{\pi}),\ldots,\widehat{w}_{N}(\widehat{\pi}))=\underset{w_{i}>0:i=1,\ldots,N}{\arg\min}\sum_{i=1}^{N}\left(w_{i}\log w_{i}-w_{i}\right)\text{ subject to}\\
    0=\frac{1}{N}\sum_{i=1}^{N}(1-\widehat{\pi}(\bs X_{i}))\left\{ \frac{w_{i}(1-T_{i})}{1-\widehat{e}(\bs X_{i})}-1\right\} \widehat{\e}\left[Y\mid\bs X=\bs X_i,T=0\right]\\
    \quad+\frac{1}{N}\sum_{i=1}^{N}\widehat{\pi}(\bs X_{i})\left\{ \frac{w_{i}T_{i}}{\widehat{e}(\bs X_{i})}-1\right\} \widehat{\e}\left[Y\mid\bs X=\bs X_i,T=1\right],\\
    0=\frac{1}{N}\sum_{i=1}^{N}(1-\widehat{\pi}(\bs X_{i}))\left\{ \frac{w_{i}(1-T_{i})}{1-\widehat{e}(\bs X_{i})}-1\right\} \widehat{\e}\left[Y^{2}\mid\bs X=\bs X_i,T=0\right]\\
    \quad+\frac{1}{N}\sum_{i=1}^{N}\widehat{\pi}(\bs X_{i})\left\{ \frac{w_{i}T_{i}}{\widehat{e}(\bs X_{i})}-1\right\} \widehat{\e}\left[Y^{2}\mid\bs X=\bs X_i,T=1\right],
    \end{cases}
    \]
    where $\widehat{\e}[Y\mid\bs X=\bs x,T=t]$ and $\widehat{\e}[Y^{2}\mid\bs X=\bs x,T=t]$ are the DNN estimates of $\e[Y\mid\bs X=\bs x,T=t]$ and $\e[Y^{2}\mid\bs X=\bs x,T=t]$, respectively.
    
    The empirical results highlight the trade-off inherent in our method. The benchmark policy $\widehat{\pi}_{\mathrm{Lin}}$ yields a mean outcome of $16{,}201.57$, with a standard deviation of $16{,}763.98$. In contrast, the policy $\widehat{\pi}_{\mathrm{Nonlin}}$, estimated under our nonlinear welfare criterion, yields a mean outcome of $16{,}132.78$ and a standard deviation of $16{,}617.18$. Relative to the benchmark, $\widehat{\pi}_{\mathrm{Nonlin}}$ reduces the mean outcome by approximately $0.42\%$ and the standard deviation by about $0.88\%$, reflecting our goal of achieving lower outcome dispersion under the nonlinear welfare criterion.

\section{Conclusion}\label{sec:conclusion}

This paper presents a data-driven policy-learning procedure that utilizes observational data to address a nonlinear welfare criterion within an infinite-dimensional policy space. The proposed learning procedure expands the existing literature on policy learning by moving from a linear (utilitarian) welfare criterion to a nonlinear one, transitioning from finite-dimensional to infinite-dimensional policy spaces, and shifting focus from a known propensity score to an unknown one. Additionally, we introduce a novel reweighting-based debiasing method, providing a valuable alternative to the current double debiasing approach. We applied this procedure to the JTPA study, where we found a balance between efficiency and equity.

However, a significant challenge remains in the computational aspect: determining the best policy $\widehat{\pi}_{\ell,I}$ is fundamentally a nonconvex optimization problem. Due to the nonlinearity of the welfare criterion and the structure of the policy space, multiple local optima may arise. Future research should aim to develop more efficient optimization techniques, such as tighter convex relaxations or advanced heuristic search algorithms, to tackle these computational challenges.

\appendix
\setlength{\abovedisplayskip}{2pt plus 1pt minus 1pt}
\setlength{\belowdisplayskip}{2pt plus 1pt minus 1pt}
\setlength{\abovedisplayshortskip}{0pt plus 1pt minus 1pt}
\setlength{\belowdisplayshortskip}{1pt plus 1pt minus 1pt}
\setlength{\jot}{0pt}
\begin{center}
{\LARGE\bfseries Appendix}
\end{center}
The appendices include the following components: DNN estimators, weight calibration, calculations of approximation errors, and proofs of the main results. Throughout the appendix, we write ``w.p.'' as shorthand for ``with probability.''

\textit{Notation}.
The notation used throughout this paper, and its appendices, is as follows. Let $\N:=\{0,1,2,\ldots\}$ and $\N_{+}:=\{1,2,\ldots\}$. For a scalar $a$, $\lfloor a\rfloor$ denotes the integer part of $a$. The indicator of an event $A$ is denoted by $1(A)$ or $1\{A\}$. For any set $D$, $|D|$ denotes its cardinality.

For a column vector $\bs x=(x_{1},\ldots,x_{d})^{\trans}\in\R^{d}$, the Euclidean norm is defined as $\|\bs x\|=(\bs x^{\trans}\bs x)^{1/2}$. The notation $\bs x_{-d}=(x_{1},\ldots,x_{d-1})^{\trans}$ separates the last coordinate from the remaining covariates. For a matrix $A$, $\|A\|$ denotes the operator norm, which is the largest singular value. When $A$ is symmetric and positive semidefinite, $\lambda_{\min}(A)$ indicates its smallest eigenvalue, which is also its smallest singular value. For a multi-index $\bs{\alpha}=(\alpha_{1},\ldots,\alpha_{d})\in\N^{d}$, we define $\|\bs{\alpha}\|_{1}:=\alpha_{1}+\cdots+\alpha_{d}$ and $\partial^{\bs{\alpha}}f$ as the partial derivative. For an integer $s\ge0$ and a domain $D$, $C^{s}(D)$ represents the class of functions that have continuous partial derivatives up to order $s$, equipped with the usual maximum sup-norm over those derivatives.

Let $(\Omega,\mathcal{F},P)$ be the underlying probability space, and let $\e$ denote the expectation under the probability measure $P$. For an integrable function $f$, we define $Pf:=\int f\,dP=\e[f]$. When there is no risk of confusion, we use $\bs Z$ to denote the observed data triple consisting of $Y$, $\bs X$, and $T$. For a training set $I\subset\{1,\ldots,N\}$, let $m=|I|$ represent the size of the training sample. We define the empirical measure based on the training sample as $P_{m}f:=m^{-1}\sum_{i\in I}f(\bs Z_{i})$.

For any random vector or matrix $\bs X$, we define the norm $\|\bs X\|_{P,q}:=(\int\|\bs X\|^{q}dP)^{1/q}$ for $q\ge1$. For any function $f:\mathcal{X}\subset\R^{d}\to\R^{d'}$, we define $\|f\|_{P,q}:=\|f(\bs X)\|_{P,q}$ and $\|f\|_{\infty}:=\sup_{\bs x\in\mathcal{X}}\|f(\bs x)\|$. For two sequences of random vectors $X_{n}$ and $Y_{n}$, we say $X_{n}=o_{P}(\|Y_{n}\|)$ if $P(\|X_{n}\|/\|Y_{n}\|>\epsilon)\to0$ for any $\epsilon>0$. Similarly, $X_{n}=O_{P}(\|Y_{n}\|)$ means that for any $\epsilon>0$ there exists $M<\infty$ such that $\limsup_{n\to\infty}P(\|X_{n}\|/\|Y_{n}\|\ge M)<\epsilon$. For two positive, non-random sequences $a_{n}$ and $b_{n}$, we write $a_{n}\lesssim b_{n}$ if $a_{n}\le Cb_{n}$ for some finite constant $C>0$ that is independent of $n$. We write $a_{n}\asymp b_{n}$ if both $a_{n}\lesssim b_{n}$ and $b_{n}\lesssim a_{n}$ hold.

For a function class $\mathcal{G}$ and a norm $\|\cdot\|_{Q,q}$, where $Q$ is a probability measure, the covering number $N(\epsilon,\mathcal{G},\|\cdot\|_{Q,q})$ is defined as the minimal cardinality of a finite set $\{g_{1},\ldots,g_{N_{\epsilon}}\}$ such that $\sup_{g\in\mathcal{G}}\inf_{1\le i\le N_{\epsilon}}\|g-g_{i}\|_{Q,q}<\epsilon$. Whenever $\sup_{Q}$ appears in entropy bounds, it refers to the supremum over finitely discrete probability measures. Following \citet{chernozhukov2014Gaussian}, when a class has a bounded envelope, we say it is VC-type with index $v$ if, for some constants $A>1$ and $C>0$, \(\sup_{Q}\log N(C\epsilon,\mathcal{G},\|\cdot\|_{Q,2})\leq Cv\log(A/\epsilon)\) for all \(0<\epsilon<1\). For a policy class $\Pi$, $\VC(\Pi)$ denotes its Vapnik--Chervonenkis dimension. For a function class $\mathcal{G}$ and a function $h$, we define $h\circ\mathcal{G}:=\{h\circ g:g\in\mathcal{G}\}$.

\section{DNN estimators}
\label{sec:app-nuisance-estimation}

This section presents DNN estimators for the propensity score and the conditional mean functions, as well as initial estimators of the intermediate parameters. Let
\begin{equation}
\begin{aligned}
& \mathcal{F}_{\mathrm{DNN}}(\mathcal{H},\mathcal{D})\\
:= & \Bigl\{\bs x\in\R^{d}\mapsto W_{\mathcal{D}}\bs{\sigma}\Big(\cdots\bs{\sigma}\Big(W_{3}\bs{\sigma}\Big(W_{2}\bs{\sigma}\Big(W_{1}\bs{\sigma}(W_{0}\bs x+b_{0})+b_{1}\Big)+b_{2}\Big)+b_{3}\Big)+\cdots\Big)+b_{\mathcal{D}}:\\
 & \quad W_{l}\in\R^{d_{l+1}\times d_{l}},b_{l}\in\R^{d_{l+1}},0\leq l\leq\mathcal{D},d_{0}=d,d_{\mathcal{D}+1}=1,\max\{d_{1},\ldots,d_{\mathcal{D}}\}\leq\mathcal{H}\Bigl\}
\end{aligned}
\label{eq:DNN class}
\end{equation}
denote a class of deep neural networks (DNNs) with depth $\mathcal{D}$ and width $\mathcal{H}$. These networks utilize fully connected feedforward architectures with the nonsmooth rectified linear unit (ReLU) activation function defined as $\bs{\sigma}(z)=\max\{z,0\}$, which is applied to each component of $z$ if $z$ is a vector.

To estimate the propensity score and ensure that the resulting estimates fall within the range $(0,1)$, we use a logistic transformation of the DNN class and estimate the propensity score function using a DNN and the subsample $I$ as follows:
\begin{equation}
  \begin{aligned}\widehat{e}_{I} & :=\underset{e\in\text{logistic}\circ\left\{ f\in\mathcal{F}_{\mathrm{DNN}}(\mathcal{H}_{e},\mathcal{D}_{e}):\|f\|_{\infty}\le M\right\} }{\arg\min}\sum_{i\in I}\left\{ \frac{T_{i}}{e^{2}(\bs X_{i})}-\frac{2}{e(\bs X_{i})}+\frac{1-T_{i}}{\left(1-e(\bs X_{i})\right)^{2}}-\frac{2}{1-e(\bs X_{i})}\right\} \end{aligned}
  \label{eq:logistic reg for propensity}
  \end{equation}
where $\mathcal{H}_{e}$ and $\mathcal{D}_{e}$ represent the width and depth of the DNN, respectively, and $\text{logistic}(x)=1/\{1+\exp(-x)\}$.

We estimate the conditional mean functions $\mu_{0t}^{*}(\bs x;\bs{\beta}):=\e[U(Y,\bs X,\bs{\beta})\mid\bs X=\bs x,T=t]$ and $\mu_{jt}^{*}(\bs x;\bs{\beta}):=\e[\mathcal{L}_{j}^{\prime}(Y-\beta_{j})\mid\bs X=\bs x,T=t]$, $j=1,\ldots,p$, using DNNs as follows, where $\mathcal{L}_{j}^{\prime}$ is understood in the context of Assumption~\ref{assu:regularity-conditions-on-L}:
\begin{align}
\widehat{\mu}_{I,jt}(\cdot;\bs{\beta})&:=\underset{f\in\mathcal{F}_{\mathrm{DNN}}(\mathcal{H}_{\mu},\mathcal{D}_{\mu})}{\arg\min}\frac{1}{\left|I\right|}\sum_{i\in I}1(T_{i}=t)(\mathcal{L}_{j}^{\prime}(Y_{i}-\beta_{j})-f(\bs X_{i}))^{2},\  1\leq j \leq p,\label{eq:DNN reg for L}\\
\widehat{\mu}_{I,0t}(\cdot;\bs{\beta})&:=\underset{f\in\mathcal{F}_{\mathrm{DNN}}(\mathcal{H}_{\mu},\mathcal{D}_{\mu})}{\arg\min}\frac{1}{\left|I\right|}\sum_{i\in I}1(T_{i}=t)(U(Y_{i},\bs X_{i},\bs{\beta})-f(\bs X_{i}))^{2}.\label{eq:DNN reg for U}
\end{align}
where $\mathcal{H}_{\mu}$ and $\mathcal{D}_{\mu}$ represent the width and depth of the DNN, respectively. Finally, we estimate the intermediate parameters, $\widehat{\bs{\beta}}_{I}^{\mathrm{init}}(\pi)$, as follows:
\begin{equation}
\widehat{\bs{\beta}}_{I}^{\mathrm{init}}(\pi)=\underset{\bs{\beta}=(\beta_{1},\ldots,\beta_{p})^{\trans}\in\R^{p}}{\arg\min}\sum_{j=1}^{p}\frac{1}{\left|I\right|}\sum_{i\in I}\Bigl\{\frac{\pi(\bs X_{i})T_{i}}{\widehat{e}_{I}(\bs X_{i})}+\frac{(1-\pi(\bs X_{i}))(1-T_{i})}{1-\widehat{e}_{I}(\bs X_{i})}\Bigr\}\mathcal{L}_{j}(Y_{i}-\beta_{j}).\label{eq:beta-init-pi-I}
\end{equation}

\subsection{Weights calibration}

To obtain the calibrated weights as defined in (\ref{eq:calibrated-weight-constraints}), we perform a first-order Taylor expansion of $\widehat{W}_{I}(\pi)-W(\pi)$ in $\widehat{\bs{\beta}}_{I}(\pi)$ around $\bs{\beta}^{*}(\pi)$:
\[
\begin{aligned} & \widehat{W}_{I}(\pi)-W(\pi)\\
= & \frac{1}{\left|I\right|}\sum\limits_{i\in I}\widehat{w}_{I,i}(\pi)\biggl\{\frac{(1-\pi(\bs X_{i}))(1-T_{i})}{1-\widehat{e}_{I}(\bs X_{i})}+\frac{\pi(\bs X_{i})T_{i}}{\widehat{e}_{I}(\bs X_{i})}\biggr\} U(Y_{i},\bs X_{i},\bs{\beta}^{*}(\pi))-W(\pi)\\
 & \ +\frac{1}{\left|I\right|}\sum\limits_{i\in I}\widehat{w}_{I,i}(\pi)\biggl\{\frac{(1-\pi(\bs X_{i}))(1-T_{i})}{1-\widehat{e}_{I}(\bs X_{i})}+\frac{\pi(\bs X_{i})T_{i}}{\widehat{e}_{I}(\bs X_{i})}\biggr\}\frac{\partial U(Y_{i},\bs X_{i},\bs{\beta}^{*}(\pi))}{\partial\bs{\beta}}\bigl\{\widehat{\bs{\beta}}_{I}(\pi)-\bs{\beta}^{*}(\pi)\bigr\}\\
 & \ +o(\|\widehat{\bs{\beta}}_{I}(\pi)-\bs{\beta}^{*}(\pi)\|).
\end{aligned}
\]
We refer to the first two displayed summands as the first and second terms, respectively. The bias in machine learning propagates directly through the estimated propensity score in the first term, and indirectly through $\widehat{\bs{\beta}}_{I}(\pi)-\bs{\beta}^{*}(\pi)$ in the second term. Decomposing the first term as $A_1+A_2$, where
\[
\begin{aligned}
A_{1}
&:=\frac{1}{\left|I\right|}\sum\limits_{i\in I}\widehat{w}_{I,i}(\pi)
\biggl\{\frac{(1-\pi(\bs X_{i}))(1-T_{i})}{1-\widehat{e}_{I}(\bs X_{i})}
+\frac{\pi(\bs X_{i})T_{i}}{\widehat{e}_{I}(\bs X_{i})}\biggr\}\\
&\qquad\times\bigl\{U(Y_{i},\bs X_{i},\bs{\beta}^{*}(\pi))-\mu_{0T_{i}}^{*}(\bs X_{i};\bs{\beta}^{*}(\pi))\bigr\},\\
A_{2}
&:=\frac{1}{\left|I\right|}\sum\limits_{i\in I}\widehat{w}_{I,i}(\pi)
\biggl\{\frac{(1-\pi(\bs X_{i}))(1-T_{i})}{1-\widehat{e}_{I}(\bs X_{i})}
+\frac{\pi(\bs X_{i})T_{i}}{\widehat{e}_{I}(\bs X_{i})}\biggr\}\\
&\qquad\times\mu_{0T_{i}}^{*}(\bs X_{i};\bs{\beta}^{*}(\pi))-W(\pi),
\end{aligned}
\]
we have by the definition of $\mu_{0t}^{*}$ that $\e[U(Y_{i},\bs X_{i},\bs{\beta}^{*}(\pi))-\mu_{0T_{i}}^{*}(\bs X_{i};\bs{\beta}^{*}(\pi))\mid T_{i},\bs X_{i}]=0$. Hence $A_{1}=O_{P}(|I|^{-1/2})$ by applying a central limit theorem. Now, we will rewrite the term $A_{2}$ as:
\[
\begin{aligned}
A_{2}&=\underbrace{\begin{aligned}[t]
\frac{1}{\left|I\right|}\sum\limits_{i\in I}\biggl[&(1-\pi(\bs X_{i}))\bigl\{\frac{\widehat{w}_{I,i}(\pi)(1-T_{i})}{1-\widehat{e}_{I}(\bs X_{i})}-1\bigr\}\mu_{00}^{*}(\bs X_{i};\bs{\beta}^{*}(\pi))\\
&+\pi(\bs X_{i})\bigl\{\frac{\widehat{w}_{I,i}(\pi)T_{i}}{\widehat{e}_{I}(\bs X_{i})}-1\bigr\}\mu_{01}^{*}(\bs X_{i};\bs{\beta}^{*}(\pi))\biggr]
\end{aligned}}_{:=A_{21}}\\
 & +\underbrace{\frac{1}{\left|I\right|}\sum\limits_{i\in I}\bigl\{(1-\pi(\bs X_{i}))\mu_{00}^{*}(\bs X_{i};\bs{\beta}^{*}(\pi))+\pi(\bs X_{i})\mu_{01}^{*}(\bs X_{i};\bs{\beta}^{*}(\pi))-W(\pi)\bigr\}}_{:=A_{22}}.
\end{aligned}
\]
Once again, $A_{22}=O_{P}(|I|^{-1/2})$ follows from a central limit theorem, indicating that the term $A_{21}$ represents the direct bias. To reduce this direct bias, we can adjust the weights to ensure that the direct bias becomes zero:
\[
\begin{aligned}
0 & =\frac{1}{\left|I\right|}\sum\limits_{i\in I}\biggl[(1-\pi(\bs X_{i}))\bigl\{\frac{\widehat{w}_{I,i}(\pi)(1-T_{i})}{1-\widehat{e}_{I}(\bs X_{i})}-1\bigr\}\mu_{00}^{*}(\bs X_{i};\bs{\beta}^{*}(\pi))\\
 & \qquad+\pi(\bs X_{i})\bigl\{\frac{\widehat{w}_{I,i}(\pi)T_{i}}{\widehat{e}_{I}(\bs X_{i})}-1\bigr\}\mu_{01}^{*}(\bs X_{i};\bs{\beta}^{*}(\pi))\biggr].
\end{aligned}
\]
By applying similar arguments to the first-order condition of (\ref{eq:def-of-beta-hat-pi-I-unknown-propensity}), we can determine the weights needed to eliminate the indirect bias for each $j=1,\ldots,p$:
\[
\begin{aligned}0= & \frac{1}{\left|I\right|}\sum\limits_{i\in I}(1-\pi(\bs X_{i}))\bigl\{\frac{\widehat{w}_{I,i}(\pi)(1-T_{i})}{1-\widehat{e}_{I}(\bs X_{i})}-1\bigr\}\mu_{j0}^{*}(\bs X_{i};\bs{\beta}^{*}(\pi))\\
 & +\frac{1}{\left|I\right|}\sum\limits_{i\in I}\pi(\bs X_{i})\bigl\{\frac{\widehat{w}_{I,i}(\pi)T_{i}}{\widehat{e}_{I}(\bs X_{i})}-1\bigr\}\mu_{j1}^{*}(\bs X_{i};\bs{\beta}^{*}(\pi)).
\end{aligned}
\]
Both $\mu_{jt}^{*}(\cdot;\bs{\beta})$ and $\bs{\beta}^{*}(\pi)$ are unknown, so we replace them with the initial estimator $\widehat{\bs{\beta}}_{I}^{\mathrm{init}}(\pi)$ as defined in (\ref{eq:beta-init-pi-I}) and with the DNN estimators as defined in (\ref{eq:DNN reg for L})--(\ref{eq:DNN reg for U}), and compute the calibrated weights from the minimization problem
\[
\begin{aligned}
  &\left(\widehat{w}_{I,i}(\pi):i\in I\right)=\underset{w_{i}>0:i\in I}{\arg\min}\sum_{i\in I}\left(w_{i}\log w_{i}-w_{i}\right)\\
  &\text{s.t.}\quad 0=\frac{1}{\left|I\right|}\sum\limits_{i\in I}(1-\pi(\bs X_{i}))\bigl\{\frac{w_{i}(1-T_{i})}{1-\widehat{e}_{I}(\bs X_{i})}-1\bigr\}\widehat{\mu}_{I,j0}(\bs X_{i};\widehat{\bs{\beta}}_{I}^{\mathrm{init}}(\pi))\\
&\quad+\frac{1}{\left|I\right|}\sum\limits_{i\in I}\pi(\bs X_{i})\bigl\{\frac{w_{i}T_{i}}{\widehat{e}_{I}(\bs X_{i})}-1\bigr\}\widehat{\mu}_{I,j1}(\bs X_{i};\widehat{\bs{\beta}}_{I}^{\mathrm{init}}(\pi)),\quad j=0,1,\ldots,p.
\end{aligned}
\]

\section{Calculation of approximation errors}
\label{sec:approximation-errors}
This section calculates the bounds on the approximation-error rates for example policy spaces, such as~\ref{exa:monotone policies}--\ref{exa:Neural Networks},
under the high-level Assumption~\ref{ass:high-level}.

Let $d_{\Delta}(\pi_{1},\pi_{2}):=P\bigl(\pi_{1}(\bs X)\neq\pi_{2}(\bs X)\bigr)$. It follows from Assumption~\ref{ass:high-level} that
$\inf_{\pi\in\Pi_{\ell}}|W(\pi^{*})-W(\pi)|\le C_{W}\inf_{\pi\in\Pi_{\ell}}d_{\Delta}(\pi^{*},\pi)$,
where $\pi^{*}:=\arg\max_{\pi\in\Pi_{\infty}}W(\pi)$. As a result, we obtain
\begin{equation}
W(\pi^{*})-\max_{\pi\in\Pi_{\ell}}W(\pi)\le \inf_{\pi\in\Pi_{\ell}}\left|W(\pi^{*})-W(\pi)\right|\le C_{W}\inf_{\pi\in\Pi_{\ell}}d_{\Delta}(\pi,\pi^{*}).\label{eq:approx-error-transfer}
\end{equation}

We study examples of monotone policies, decision trees, and neural networks from Section~\ref{sec:A-General-Learning}. In each example, we derive an explicit
bound for $\inf_{\pi\in\Pi_{\ell}}d_{\Delta}(\pi,\pi^{*})$ and then use
\eqref{eq:approx-error-transfer} to determine the corresponding welfare
approximation rate. We will start with the example of monotone policies.

\subsection{Example~\ref{exa:monotone policies}: monotone policies}

\subsubsection{Monotone policy class and the sieve}

In the case of monotone policies, it is convenient to normalize the supports such that $\bs X_{i}=(X_{i1},X_{i2})^{\trans}\in[0,1]^{2}$. Consider the monotone class defined as:
\[
\Pi_{\infty}=\Pi^{\mathrm{mon}}_{\infty}=\left\{ \pi_{f}(x_{1},x_{2})=1\{f(x_{1})\geq x_{2}\}:f:[0,1]\to[0,1]\ \text{is non-increasing}\right\} .
\]
This class has an infinite VC dimension; see \citet{devroye1996probabilistic}
and \citet{mbakop2021Model}.

The sieve used by \citet[Example 3.2]{mbakop2021Model} is constructed from
triangular basis functions. In their example, the sieve is defined as:
$G=\{(x_{1},x_{2}):x_{2}\ge f(x_{1}),\ f\text{ increasing}\}$. Using our notation, the same construction can be achieved by reversing the direction of
monotonicity, which simply changes the sign of the linear
inequality involving the coefficient vector.

Fix an integer $J\ge1$. For $j=0,\ldots,J$, define
\[
\psi_{J,j}(x)=\begin{cases}
1-\left|Jx-j\right|, & x\in\bigl[(j-1)/J,(j+1)/J\bigr]\cap[0,1],\\[0.3em]
0, & \text{otherwise}.
\end{cases}
\]
These are the standard ``hat'' or triangular basis functions. Define
the coefficient vector $\bs{\theta}_{J}=(\theta_{J,0},\ldots,\theta_{J,J})^{\trans}$
and the first-difference matrix
\[
D_{J}=\begin{bmatrix}-1 & 1 & 0 & \cdots & 0 & 0\\
0 & -1 & 1 & \cdots & 0 & 0\\
\vdots & \vdots & \vdots & \ddots & \vdots & \vdots\\
0 & 0 & 0 & \cdots & -1 & 1
\end{bmatrix}\in\R^{J\times(J+1)}.
\]
For our non-increasing class, the relevant coefficient restriction
is $D_{J}\bs{\theta}_{J}\le0$, i.e., $\theta_{J,0}\ge\theta_{J,1}\ge\cdots\ge\theta_{J,J}$.

We define the piecewise-linear function class and the associated policy class as
\[
\begin{aligned}
\mathcal{F}^{\mathrm{mon}}_{J}&:=\left\{ f_{\bs{\theta}}(x_{1}):=\sum^{J}_{j=0}\theta_{J,j}\psi_{J,j}(x_{1}):
\bs{\theta}_{J}\in[0,1]^{J+1},\ D_{J}\bs{\theta}_{J}\le0\right\},\\
\Pi^{\mathrm{mon}}_{J}&:=\left\{ \pi_{\bs{\theta}}(x_{1},x_{2})=1\{x_{2}\le f_{\bs{\theta}}(x_{1})\}:f_{\bs{\theta}}\in\mathcal{F}^{\mathrm{mon}}_{J}\right\}.
\end{aligned}
\]

For a sieve index $\ell$, we define $J_{\ell}:=2^{\ell}$ and $\Pi_{\ell}:=\Pi^{\mathrm{mon}}_{J_{\ell}}$.
This refinement is similar to that discussed in the literature \citet[Example 3.2]{mbakop2021Model},
with an inconsequential shift in the index.

\begin{rem}
The basis $\{\psi_{J,j}\}^{J}_{j=0}$ interpolates the coefficient
vector at the grid points: $f_{\bs{\theta}}(j/J)=\theta_{J,j}$. On the interval $[j/J,(j+1)/J]$, we have $f_{\bs{\theta}}(x)=\theta_{J,j}\{1-(Jx-j)\}+\theta_{J,j+1}(Jx-j)$,
which shows that $f_{\bs{\theta}}$ is linear on that interval. Therefore, $\mathcal{F}^{\mathrm{mon}}_{J}$ consists precisely of continuous piecewise-linear functions with knots at $j/J$, where the values at the knots are non-increasing.
\end{rem}

\subsubsection{Approximation error rate}

We will derive the approximation error rate in $d_{\Delta}(\pi,\pi^{*})$ and then convert it to a welfare approximation error rate using \eqref{eq:approx-error-transfer}.

\begin{assumption}[Regularity for the monotone example]
\label{ass:monotone-density} The optimal policy, $\pi^{*}$, belongs to the set $\Pi^{\mathrm{mon}}_{\infty}$, which implies that $\pi^{*}(x_{1},x_{2})=1\{x_{2}\le f^{*}(x_{1})\}$ for some non-increasing $f^{*}:[0,1]\to[0,1]$. Additionally, for every
$x_{2}\in[0,1]$, the conditional distribution of $X_{1}\mid X_{2}=x_{2}$
is absolutely continuous and has density bounded by
$A_{X}<\infty$.
\end{assumption}
\begin{prop}[Monotone-sieve approximation error rate and VC dimension]
\label{prop:monotone-approx-vc} Let $\Pi_{\ell}=\Pi^{\mathrm{mon}}_{2^{\ell}}$, and suppose that Assumptions~\ref{ass:high-level} and \ref{ass:monotone-density} hold. For every $\ell\ge1$, we have $\VC(\Pi_{\ell})\le2^{\ell}+3$ and \(W(\pi^{*})-\max_{\pi\in\Pi_{\ell}}W(\pi)\le C_{W}A_{X}\,2^{-\ell}\).
\end{prop}
\begin{proof}
To define the parameter $J$, we set $J=J_{\ell}=2^{\ell}$. By construction, every function $f\in\mathcal{F}^{\mathrm{mon}}_{J}$
can be expressed as: $f(x_{1})=\sum^{J}_{j=0}\theta_{j}\psi_{J,j}(x_{1})$, which indicates that $\mathcal{F}^{\mathrm{mon}}_{J}$ is contained in the linear span of the basis functions $\{\psi_{J,0},\ldots,\psi_{J,J}\}$. These basis functions are linearly independent. Specifically, if $\sum^{J}_{j=0}a_{j}\psi_{J,j}(x)=0$ for
all $x\in[0,1]$, then evaluating this equation at the grid points $x=i/J$ leads
to $a_{i}=0$ for every $i=0,\ldots,J$, because $\psi_{J,j}(i/J)=1\{i=j\}$.
Consequently, the dimension of the ambient linear space is $J+1$.

Next, $\Pi^{\mathrm{mon}}_{J}$ represents the class of subgraphs generated
by $\mathcal{F}^{\mathrm{mon}}_{J}$, as we define $\pi_{\bs{\theta}}(x_{1},x_{2})=1\{x_{2}\le f_{\bs{\theta}}(x_{1})\}$.
Thus, $\Pi^{\mathrm{mon}}_{J}$ is a subclass of the subgraph
class generated by the $(J+1)$-dimensional linear span of the functions $\{\psi_{J,0},\ldots,\psi_{J,J}\}$. According to Lemma~2.6.15 from \citet{vaart1996Weak}, the VC dimension of the subgraph class of a finite-dimensional vector space of measurable functions is bounded by the dimension plus $2$. Therefore, $\VC(\Pi_{\ell})=\VC(\Pi^{\mathrm{mon}}_{J})\le(J+1)+2=J+3=2^{\ell}+3$.

To prove the approximation error bound, we define $J=J_{\ell}=2^{\ell}$
and consider the grid $\xi_{j}:=\frac{j}{J}$, $j=0,1,\ldots,J$.
We define the coefficient vector by sampling the true boundary given by $\theta^{*}_{j}:=f^{*}(\xi_{j})$,
$j=0,1,\ldots,J$. Since $f^{*}$ is non-increasing, it follows that $\theta^{*}_{0}\ge\theta^{*}_{1}\ge\cdots\ge\theta^{*}_{J}$. This implies that $D_{J}\bs{\theta}^{*}_{J}\le0$, where $\bs{\theta}^{*}_{J}=(\theta^{*}_{0},\ldots,\theta^{*}_{J})^{\top}$. Thus, we can define the linear interpolant $\tilde{f}_{\ell}(x_{1}):=\sum^{J}_{j=0}\theta^{*}_{j}\psi_{J,j}(x_{1})$, which belongs to $\mathcal{F}^{\mathrm{mon}}_{J}$. Consequently, we define $\tilde{\pi}_{\ell}(x_{1},x_{2}):=1\{x_{2}\le\tilde{f}_{\ell}(x_{1})\}\in\Pi_{\ell}$.
By \eqref{eq:approx-error-transfer}, it is sufficient to bound $d_{\Delta}(\tilde{\pi}_{\ell},\pi^{*})$.

For $i=1,\ldots,J$, we define the rectangles $M_{i}:=[\xi_{i-1},\xi_{i}]\times[f^{*}(\xi_{i}),f^{*}(\xi_{i-1})]$.
Since $f^{*}$ is non-increasing, each $M_{i}$ forms a rectangle (which may have zero height). On the interval $[\xi_{i-1},\xi_{i}]$, both the graph of $f^{*}$ and the graph of its linear interpolant $\tilde{f}_{\ell}$
lie within $M_{i}$. This is due to the fact that the first graph respects monotonicity, while the second graph, $\tilde{f}_{\ell}$, is represented by the straight line segment connecting $(\xi_{i-1},f^{*}(\xi_{i-1}))$ and $(\xi_{i},f^{*}(\xi_{i}))$.

Thus, the area where $\pi^{*}$ and $\tilde{\pi}_{\ell}$ differ is included within the union of these rectangles: \(\{(x_{1},x_{2}):\tilde{\pi}_{\ell}(x_{1},x_{2})\neq\pi^{*}(x_{1},x_{2})\}\subseteq\bigcup^{J}_{i=1}M_{i}\).
Therefore, we have $d_{\Delta}(\tilde{\pi}_{\ell},\pi^{*})\le\sum^{J}_{i=1}P_{X}(M_{i})$. Let $M_{i}=M_{1i}\times M_{2i}$, where  $M_{1i}:=[\xi_{i-1},\xi_{i}]$ and $M_{2i}:=(f^{*}(\xi_{i}),f^{*}(\xi_{i-1})]$, with boundary points omitted because they have zero probability under the conditional-density assumption. Based on the conditional-density
assumption, we have:
\begin{align*}
P((X_{1},X_{2})\in M_{i}) & =\int_{M_{2i}}P(X_{1}\in M_{1i}\mid X_{2}=x_{2})\,dP_{X_{2}}(x_{2})\\
&\le\int_{M_{2i}}A_{X}(\xi_{i}-\xi_{i-1})\,dP_{X_{2}}(x_{2})=A_{X}J^{-1}P(X_{2}\in M_{2i}).
\end{align*}
Summing over $i$, we obtain $d_{\Delta}(\tilde{\pi}_{\ell},\pi^{*})\le A_{X}J^{-1}\sum^{J}_{i=1}P(X_{2}\in M_{2i})$.
Since $f^{*}$ is non-increasing, the intervals $M_{2i}$ are disjoint
and their union is contained within $[0,1]$. Therefore, we have $\sum^{J}_{i=1}P_{X_{2}}(M_{2i})\le1$.
Thus, $d_{\Delta}(\tilde{\pi}_{\ell},\pi^{*})\le A_{X}J^{-1}=A_{X}2^{-\ell}$.
Given that $\tilde{\pi}_{\ell}\in\Pi_{\ell}$, this implies $\inf_{\pi\in\Pi_{\ell}}d_{\Delta}(\pi,\pi^{*})\le d_{\Delta}(\tilde{\pi}_{\ell},\pi^{*})\leq A_{X}2^{-\ell}$.
The bound for the approximation error $W(\pi^{*})-\max_{\pi\in\Pi_{\ell}}W(\pi)$
follows immediately from \eqref{eq:approx-error-transfer}. This
completes the proof.
\end{proof}

\subsection{Examples~\ref{exa:Decision Trees} and~\ref{exa:Neural Networks}: A Smooth Decision Boundary Class}

We will now consider a class of policies characterized by a smooth decision boundary.
Let $d\ge2$ and $\bs X=(\bs X_{-d},X_{d})\in[0,1]^{d-1}\times[0,1]=[0,1]^{d}$.
For a positive number $M>0$ and $s\in\mathbb{N}_{+}$, we define the class of policies as follows:
\[
\begin{aligned}
\Pi^{\mathrm{SDB}}(M,s)&:=\left\{ \pi_{f}(\bs x_{-d},x_{d})=1\{x_{d}\le f(\bs x_{-d})\}: f\in C^{s}([0,1]^{d-1}),\ \|f\|_{C^{s}}\le M,\ 0\le f\le1\right\},\\
\|f\|_{C^{s}}&:=\max_{\bs{\alpha}\in\mathbb{N}^{d-1},\,\|\bs{\alpha}\|_{1}\le s}
\sup_{\bs x_{-d}\in[0,1]^{d-1}}\left|\partial^{\bs{\alpha}}f(\bs x_{-d})\right|.
\end{aligned}
\]
We will show that this class has an infinite VC dimension for
$d\ge2$.

\begin{lem}[Infinite VC dimension of the fixed-radius smooth decision-boundary
class]
\label{lem:infinite-vc-sdb-fixed} For every $s\in\mathbb{N}_{+}$
and every $M>0$, the class $\Pi^{\mathrm{SDB}}(M,s)$ has an infinite
VC dimension.
\end{lem}
\begin{proof}
Let $n$ be any integer with $n\geq 1$. We will demonstrate that $\Pi^{\mathrm{SDB}}(M,s)$ can
shatter $n$ points in the space $[0,1]^{d}$. Define the points as follows:  $\bs u_{j}:=\left(\frac{j}{n+1},\frac{1}{2},\ldots,\frac{1}{2}\right)^{\top}\in(0,1)^{d-1}$,
$j=1,\ldots,n$. Let $r_{n}:=\frac{1}{n+1}$. The points $\bs u_{1},\ldots,\bs u_{n}$
are distinct, and the distance, measured by the supremum norm $\left\Vert \bs u_{j+1}-\bs u_{j}\right\Vert _{\infty}$,
between consecutive points is exactly $r_{n}$. Choose a nonnegative
function $\psi\in C^{\infty}(\mathbb{R}^{d-1})$ such that $0\le\psi\le1$, $\psi(\bs 0)=1$, and $\mathrm{supp}(\psi)\subset\mathcal{B}(\bs 0,1/4)$,
where $\mathcal{B}(\bs 0,1/4)$ denotes the closed Euclidean ball
of radius $1/4$ centered at $\bs 0$. Define
$C_{\psi}:=\max\{1,\max_{\bs{\alpha}\in\mathbb{N}^{d-1},\,\|\bs{\alpha}\|_{1}\le s}\sup_{\bs u\in\mathbb{R}^{d-1}}|\partial^{\bs{\alpha}}\psi(\bs u)|\}<\infty$.
For each $j=1,\ldots,n$, define the rescaled bump function as $\varphi_{j}(\bs u):=\psi((\bs u-\bs u_{j})/r_{n})$, $\bs u\in[0,1]^{d-1}$.
Since the support of $\psi$ is contained in $\mathcal{B}(\bs 0,1/4)$, we have $\mathrm{supp}(\varphi_{j})\subset\mathcal{B}(\bs u_{j},r_{n}/4)$. Given that the centers $\bs u_{j}$ are spaced $r_{n}$ apart, these supports are pairwise disjoint. Therefore, at every point $\bs u\in[0,1]^{d-1}$, at most one of the functions $\varphi_{1}(\bs u),\ldots,\varphi_{n}(\bs u)$ can be nonzero.

Next, set $b_{M}:=\min\{M,1\}/2$ and $\varepsilon_{n}:=\min\{M,1\}r^{s}_{n}/(8C_{\psi})$. For any label vector $\bs{\eta}=(\eta_{1},\ldots,\eta_{n})\in\{0,1\}^{n}$, define
$f_{\bs{\eta}}(\bs u):=b_{M}+\varepsilon_{n}\sum^{n}_{j=1}(2\eta_{j}-1)\varphi_{j}(\bs u)$, $\bs u\in[0,1]^{d-1}$.
We now verify that $f_{\bs{\eta}}\in C^{s}([0,1]^{d-1})$, $\|f_{\bs{\eta}}\|_{C^{s}}\le M$, and $0\le f_{\bs{\eta}}\le1$.

First, since the supports of the $\varphi_{j}$'s are pairwise disjoint,
at each $\bs u$ at most one term in the sum contributes, so $b_{M}-\varepsilon_{n}\le f_{\bs{\eta}}(\bs u)\le b_{M}+\varepsilon_{n}$ for all $\bs u\in[0,1]^{d-1}$. Because $r^{s}_{n}\le1$, we have $\varepsilon_{n}=\frac{\min\{M,1\}}{8C_{\psi}}r^{s}_{n}\le\frac{\min\{M,1\}}{8}$, and thus $0\le b_{M}-\varepsilon_{n}\le f_{\bs{\eta}}(\bs u)\le b_{M}+\varepsilon_{n}\le\frac{5}{8}\min\{M,1\}\le1$,
so, $0\le f_{\bs{\eta}}\le1$. In particular, $\sup_{\bs u\in[0,1]^{d-1}}|f_{\bs{\eta}}(\bs u)|\le\frac{5}{8}\min\{M,1\}\le M$.

Next, let $\bs{\alpha}\in\mathbb{N}^{d-1}$ satisfy $1\le\|\bs{\alpha}\|_{1}\le s$.
Since the supports are disjoint, at each $\bs u$ at most one term is active. According to the chain rule,
$\partial^{\bs{\alpha}}\varphi_{j}(\bs u)=r^{-\|\bs{\alpha}\|_{1}}_{n}(\partial^{\bs{\alpha}}\psi)((\bs u-\bs u_{j})/r_{n})$; hence
\[
\begin{aligned}
\sup_{\bs u\in[0,1]^{d-1}}\left|\partial^{\bs{\alpha}}f_{\bs{\eta}}(\bs u)\right|&\le\varepsilon_{n}\sup_{1\le j\le n}\sup_{\bs u\in[0,1]^{d-1}}\left|\partial^{\bs{\alpha}}\varphi_{j}(\bs u)\right|\\
&\leq\varepsilon_{n}C_{\psi}r^{-\|\bs{\alpha}\|_{1}}_{n}=\frac{\min\{M,1\}}{8}r^{s-\|\bs{\alpha}\|_{1}}_{n}\le\frac{\min\{M,1\}}{8}\le M,
\end{aligned}
\]
where the last line uses $\|\bs{\alpha}\|_{1}\le s$ and $r_{n}\le1$.
Since $\|f_{\bs{\eta}}\|_{C^{s}}\le M$, it follows that $f_{\bs{\eta}}$ belongs
to the function class defining $\Pi^{\mathrm{SDB}}(M,s)$. Consequently,
the corresponding policy is defined as $\pi_{\bs{\eta}}(\bs x_{-d},x_{d}):=1\{x_{d}\le f_{\bs{\eta}}(\bs x_{-d})\}$, which also
belongs to $\Pi^{\mathrm{SDB}}(M,s)$.

Define the $n$ points in $[0,1]^{d}$ as $\bs{\xi}_{j}:=(\bs u_{j},b_{M})$, for
$j=1,\ldots,n$. Since $\varphi_{j}(\bs u_{j})=1$ and $\varphi_{i}(\bs u_{j})=0$
for $i\neq j$, we obtain that $f_{\bs{\eta}}(\bs u_{j})=b_{M}+\varepsilon_{n}(2\eta_{j}-1)$
for every $j=1,\ldots,n$. Therefore, $\pi_{\bs{\eta}}(\bs{\xi}_{j})=1\{b_{M}\le f_{\bs{\eta}}(\bs u_{j})\}=1\{b_{M}\le b_{M}+\varepsilon_{n}(2\eta_{j}-1)\}=\eta_{j}$.
Every labeling of the set $\{\bs{\xi}_{1},\ldots,\bs{\xi}_{n}\}$
can be achieved by a policy in $\Pi^{\mathrm{SDB}}(M,s)$. Since $n\ge1$
is arbitrary, the class $\Pi^{\mathrm{SDB}}(M,s)$ can shatter arbitrarily
large finite sets. Thus, $\VC\bigl(\Pi^{\mathrm{SDB}}(M,s)\bigr)=\infty$.
\end{proof}

To obtain explicit rates of approximation error, we make the following
assumption:
\begin{assumption}
\label{ass:sdb} There exist $s\in\mathbb{N}_{+}$, $M>0$, and a function
$f^{*}\in C^{s}([0,1]^{d-1})$ with $\left\Vert f^{*}\right\Vert _{C^{s}}\le M$
and $0\le f^{*}\le1$ such that $\pi^{*}(\bs x_{-d},x_{d})=1\{x_{d}\le f^{*}(\bs x_{-d})\}$.
Furthermore, for every $\bs x_{-d}\in[0,1]^{d-1}$, the conditional
distribution of $X_{d}\mid\bs X_{-d}=\bs x_{-d}$ is absolutely continuous
with a density that is bounded above by a constant $A_{X}<\infty$, uniformly across all $\bs x_{-d}$.
\end{assumption}
The next lemma is a crucial tool that transforms a uniform approximation
of $f^{*}$ into a $d_{\Delta}$-approximation of $\pi^{*}$.
\begin{lem}[Policy distance controlled by decision-boundary sup-norm error]
\label{lem:sdb-to-policy} Assuming that Assumption~\ref{ass:sdb} holds, for any measurable functions $f,g:[0,1]^{d-1}\to[0,1]$, let $\pi_{f}(\bs x_{-d},x_{d})=1\{x_{d}\le f(\bs x_{-d})\}$ and $\pi_{g}(\bs x_{-d},x_{d})=1\{x_{d}\le g(\bs x_{-d})\}$. Then \(d_{\Delta}(\pi_{f},\pi_{g})\le A_{X}\,\e[|f(\bs X_{-d})-g(\bs X_{-d})|]\leq A_{X}\,\left\Vert f-g\right\Vert _{\infty}\).
\end{lem}
\begin{proof}
Condition on $\bs X_{-d}$. The two policies differ when
$X_{d}$ lies between $f(\bs X_{-d})$ and $g(\bs X_{-d})$. Hence, by the uniform bound $A_{X}$ on the conditional density of $X_{d}\mid\bs X_{-d}=\bs x_{-d}$,
\[
\begin{aligned}
d_{\Delta}(\pi_{f},\pi_{g})&=\e\!\left[P\{\min(f,g)(\bs X_{-d})<X_{d}\le\max(f,g)(\bs X_{-d})\mid\bs X_{-d}\}\right]\\
&\leq A_{X}\e\left[\left|f(\bs X_{-d})-g(\bs X_{-d})\right|\right]\leq A_{X}\left\Vert f-g\right\Vert _{\infty}.
\end{aligned}
\]
This completes the proof.
\end{proof}
For Examples~\ref{exa:Decision Trees} and~\ref{exa:Neural Networks}, we impose the smooth decision-boundary condition $\pi^{*}\in\Pi^{\mathrm{SDB}}(M,s)$ for some fixed $M$ and $s$, and use decision-tree and neural-network sieves as finite-dimensional approximation classes.

\subsubsection{Example~\ref{exa:Decision Trees}: decision trees}

Fix an integer $\ell\ge1$. A binary decision tree with a depth of at most
$\ell$ is defined as follows: Each internal node is assigned a coordinate
index $j\in\{1,\ldots,d\}$ and a threshold $b\in\mathbb{R}$. An
observation $\bs x\in[0,1]^{d}$ reaching that node is sent to the
left child if $x_{j}<b$ and to the right child otherwise. Each leaf node
is assigned an action label in $\{0,1\}$. The tree defines a binary
policy $\pi:\mathcal{X}\to\{0,1\}$ by guiding $\bs x$ from the root
to a terminal leaf according to these splitting rules, and then outputting
the label of the terminal leaf. Let $\Pi^{\mathrm{DT}}_{\ell}$ denote
the class of all such policies whose depth is at most $\ell$.

This class is sufficiently robust to represent the piecewise-constant threshold
approximants developed below. It also facilitates the VC-dimension
analysis, as every tree in the class $\Pi^{\mathrm{DT}}_{\ell}$ can be
embedded into a complete binary tree of depth $\ell$. This is achieved by padding
premature leaves with dummy descendants without altering the induced
policy.

\begin{prop}[Approximation error and VC dimension for the decision-tree sieve]
\label{prop:dt-approx-vc} Assuming that $d\geq2$ and Assumptions~\ref{ass:high-level}
and \ref{ass:sdb} are satisfied, let $\Pi_{\ell}=\Pi^{\mathrm{DT}}_{\ell}$.
Then, there exists a constant $C>0$, independent of $\ell$, such
that for every $\ell\ge1$,
\[
W(\pi^{*})-\max_{\pi\in\Pi_{\ell}}W(\pi)\le C\cdot2^{-\left\lfloor \frac{\ell-1}{d-1}\right\rfloor },\qquad
\VC(\Pi_{\ell})\le C\cdot2^{\ell}(\ell+\log d).
\]
\end{prop}
\begin{proof}
Let $J_{\ell}:=\left\lfloor \frac{\ell-1}{d-1}\right\rfloor $ and $h_{\ell}:=2^{-J_{\ell}}$. According to Assumption~\ref{ass:sdb}, we have
$s\in\mathbb{N}_{+}$ and $f^{*}\in C^{s}([0,1]^{d-1})$ with $\|f^{*}\|_{C^{s}}\le M$.
In particular, every first-order partial derivative of $f^{*}$ is bounded by $M$. Thus, for any $\bs u,\bs v\in[0,1]^{d-1}$, \(f^{*}(\bs u)-f^{*}(\bs v)=\int^{1}_{0}\nabla f^{*}\!\bigl(\bs v+t(\bs u-\bs v)\bigr)^{\top}(\bs u-\bs v)\,dt\).
Consequently, $|f^{*}(\bs u)-f^{*}(\bs v)|\le\int^{1}_{0}\sum^{d-1}_{j=1}|\partial_{j}f^{*}(\bs v+t(\bs u-\bs v))|\,|u_{j}-v_{j}|\,dt\le(d-1)M\|\bs u-\bs v\|_{\infty}$.
Thus, $f^{*}$ is Lipschitz continuous with Lipschitz constant $L_{f}:=(d-1)M$.

We now construct a piecewise-constant approximation to $f^{*}$ on
a dyadic partition of $[0,1]^{d-1}$. The term ``dyadic'' refers to the property that
every side length is a negative power of $2$. Specifically, for
each multi-index $\bs k=(k_{1},\ldots,k_{d-1})$ where $k_{r}\in\{0,1,\ldots,2^{J_{\ell}}-1\}$,
let $I_{r,k_{r}}=[k_{r}/2^{J_{\ell}},(k_{r}+1)/2^{J_{\ell}})$ if $k_{r}<2^{J_{\ell}}-1$ and $I_{r,k_{r}}=[k_{r}/2^{J_{\ell}},1]$ otherwise. Define $Q_{\bs k}:=\prod^{d-1}_{r=1}I_{r,k_{r}}$. Each set $Q_{\bs k}$ is a $(d-1)$-dimensional
cube whose side length is exactly $h_{\ell}=2^{-J_{\ell}}$; we refer to
such a set as a dyadic cube. The collection of all these cubes forms
a partition of $[0,1]^{d-1}$, which we denote by $\mathcal{Q}_{\ell}$.
Thus every point $\bs u\in[0,1]^{d-1}$ belongs to exactly one cube
$Q\in\mathcal{Q}_{\ell}$. A two-dimensional illustration of the dyadic
partition, a terminal cube $Q$, and the representative point $\bs u_{Q}$
is provided in Figure~\ref{fig:dt-dyadic-partition}.

\begin{figure}[t]
\centering
\begin{tikzpicture}[scale=3.3,>=Stealth]
  \draw[->] (-0.05,0) -- (1.08,0) node[right] {$x_1$};
  \draw[->] (0,-0.05) -- (0,1.08) node[above] {$x_2$};

  \draw[thick] (0,0) rectangle (1,1);
  \foreach \x in {0.25,0.5,0.75} {\draw[thin] (\x,0) -- (\x,1);}
  \foreach \y in {0.25,0.5,0.75} {\draw[thin] (0,\y) -- (1,\y);}

  \fill[gray!18] (0.50,0.25) rectangle (0.75,0.50);

  \node at (0.57,0.44) {$Q$};

  \fill (0.645,0.36) circle (0.012);

  \draw[->] (0.55,0.31) -- (0.638,0.355);
  \node at (0.49,0.285) {$\bs u_Q$};

  \draw[<->] (0.50,0.215) -- (0.75,0.215);
  \node[below] at (0.625,0.215) {$h_{\ell}$};

  \draw[->] (0.80,0.42) -- (0.98,0.62);
  \node[anchor=west,align=left] at (1.00,0.66)
    {$f_{\ell}(\bs u)=f^*(\bs u_Q)$\\ for all $\bs u\in Q$};

  \node[below left] at (0,0) {$0$};
  \node[below] at (1,0) {$1$};
  \node[left] at (0,1) {$1$};
\end{tikzpicture}
\caption{The dyadic partition of the square $[0,1]^2$ is depicted for the case where $J_{\ell}=2$ (illustrated with $d-1=2$). Each terminal cube has a side length of $h_{\ell}=2^{-J_{\ell}}$. Within each cube $Q$, the approximation $f_{\ell}$ is constant and equals $f^*(\bs u_Q)$ at a representative point $\bs u_Q\in Q$. The representative point $\bs u_Q$ is indicated by the black dot.}
\label{fig:dt-dyadic-partition}
\end{figure}
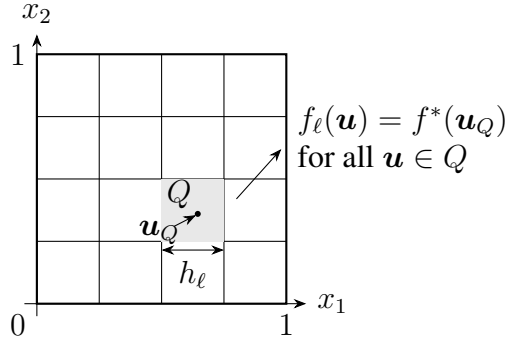

For each cube $Q\in\mathcal{Q}_{\ell}$, choose one representative
point $\bs u_{Q}\in Q$ and define $f_{\ell}(\bs u):=\sum_{Q\in\mathcal{Q}_{\ell}}f^{*}(\bs u_{Q})1\{\bs u\in Q\}$.
Since every $\bs u$ belongs to exactly one cube, the function $f_{\ell}$ is well-defined and constant on each cube $Q$. If $\bs u\in Q$, then $\|\bs u-\bs u_{Q}\|_{\infty}\le h_{\ell}$
because $Q$ has a side length of $h_{\ell}$. Therefore, we can express the difference as follows:  $|f_{\ell}(\bs u)-f^{*}(\bs u)|=|f^{*}(\bs u_{Q})-f^{*}(\bs u)|\le L_{f}h_{\ell}=L_{f}2^{-J_{\ell}}$.
Thus, $\|f_{\ell}-f^{*}\|_{\infty}\le L_{f}2^{-J_{\ell}}$.

Next, we define the policy $\pi_{\ell}(\bs x_{-d},x_{d}):=1\{x_{d}\le f_{\ell}(\bs x_{-d})\}$.
We now demonstrate that $\pi_{\ell}\in\Pi_{\ell}$ by explicitly
constructing a decision tree of a depth of at most $\ell$ that implements
this policy. Figure~\ref{fig:dt-tree-construction} illustrates this
construction in the case $d=3$ and $J_{\ell}=1$: the first $d-1$
coordinates determine the terminal cube containing $\bs x_{-d}$,
and the final split on $x_{d}$ compares $x_{d}$ with the cube-specific
threshold.

\begin{figure}[t]

\centering
\begin{tikzpicture}[>=Stealth, x=1cm, y=1cm,
  split/.style={draw, rounded corners=2pt, minimum width=16mm, minimum height=8mm, align=center},
  leaf/.style={draw, circle, minimum size=7mm, inner sep=0pt},
  edge label/.style={font=\scriptsize, fill=white, inner sep=1pt}]

  \node[split] (r) at (0,0) {$x_1<1/2$};
  \node[split] (a) at (-3.4,-1.6) {$x_2<1/2$};
  \node[split] (b) at ( 3.4,-1.6) {$x_2<1/2$};

  \node[split] (a1) at (-5.1,-3.2) {$x_3\le c_1$};
  \node[split] (a2) at (-1.7,-3.2) {$x_3\le c_2$};
  \node[split] (b1) at ( 1.7,-3.2) {$x_3\le c_3$};
  \node[split] (b2) at ( 5.1,-3.2) {$x_3\le c_4$};

  \node[leaf] (l11) at (-6.0,-4.8) {$1$};
  \node[leaf] (l12) at (-4.2,-4.8) {$0$};
  \node[leaf] (l21) at (-2.6,-4.8) {$1$};
  \node[leaf] (l22) at (-0.8,-4.8) {$0$};
  \node[leaf] (l31) at ( 0.8,-4.8) {$1$};
  \node[leaf] (l32) at ( 2.6,-4.8) {$0$};
  \node[leaf] (l41) at ( 4.2,-4.8) {$1$};
  \node[leaf] (l42) at ( 6.0,-4.8) {$0$};

  \draw (r) -- (a) node[midway, above left, edge label] {yes};
  \draw (r) -- (b) node[midway, above right, edge label] {no};

  \draw (a) -- (a1) node[midway, above left, edge label] {yes};
  \draw (a) -- (a2) node[midway, above right, edge label] {no};
  \draw (b) -- (b1) node[midway, above left, edge label] {yes};
  \draw (b) -- (b2) node[midway, above right, edge label] {no};

  \draw (a1) -- (l11);
  \draw (a1) -- (l12);
  \draw (a2) -- (l21);
  \draw (a2) -- (l22);
  \draw (b1) -- (l31);
  \draw (b1) -- (l32);
  \draw (b2) -- (l41);
  \draw (b2) -- (l42);

  \node[font=\small] at (-5.1,-2.55) {$Q_1$};
  \node[font=\small] at (-1.7,-2.55) {$Q_2$};
  \node[font=\small] at ( 1.7,-2.55) {$Q_3$};
  \node[font=\small] at ( 5.1,-2.55) {$Q_4$};
\end{tikzpicture}
\caption{In this illustration of tree construction where $d=3$ and $J_{\ell}=1$, the first two levels divide the $(x_1,x_2)$-space into four dyadic cubes $Q_1,\ldots,Q_4$. The final division on $x_3$ compares it to the cube-specific constant $c_i=f^*(\bs u_{Q_i})=f_{\ell}(\bs x_{-3})$ for $\bs x_{-3}\in Q_i$.}
\label{fig:dt-tree-construction}
\end{figure}
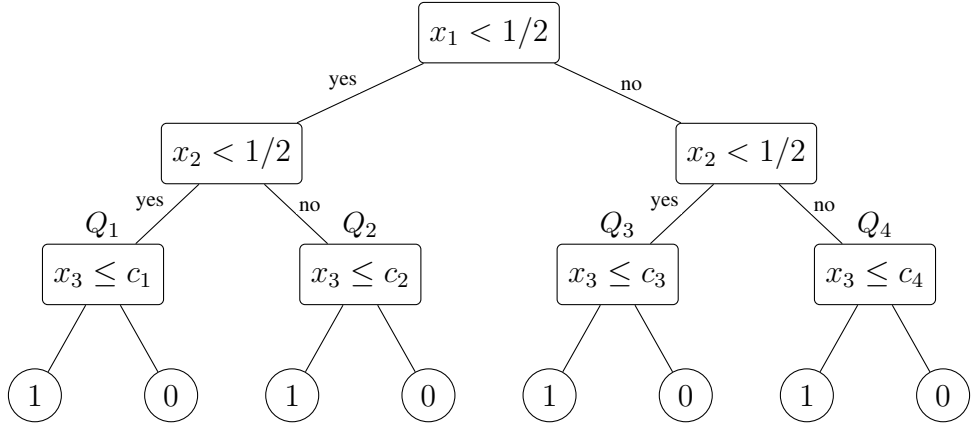

Start from the root node, which corresponds to the entire domain $[0,1]^{d-1}$
for the first $d-1$ coordinates. We first build the partition $\mathcal{Q}_{\ell}$
using only these initial coordinates. The \emph{refinement } process is defined as follows. Suppose the current cell in the first $d-1$
coordinates is of the form $\prod^{d-1}_{r=1}[a_{r},a_{r}+\delta)$
for some side length $\delta>0$. We split this cell once in the first coordinate
at the midpoint $a_{1}+\delta/2$. Then, we split each resulting
child cell once in the second coordinate at its midpoint. This process continues sequentially through coordinates 3 to $d-1$. After these $d-1$ successive
splits, the original cell is divided into $2^{d-1}$ smaller
cubes, each with a side length of $\delta/2$.

Apply this refinement procedure repeatedly. At the beginning, there is a
unique cell, which is $[0,1]^{d-1}$, so its side length is $1$. After one
refinement round, every resulting cell has a side length of $2^{-1}$.
After two rounds of refinement, each resulting cell has a side length of
$2^{-2}$. Continuing in this pattern, after exactly $J_{\ell}$ refinement
rounds, every resulting cell will have a side length of $2^{-J_{\ell}}$ and
will therefore be one of the cubes in $\mathcal{Q}_{\ell}$. These final
cells are precisely the leaves produced by the partitioning stage of
the construction. When we refer to a \emph{terminal cube}, we mean
one of these final cubes in $\mathcal{Q}_{\ell}$, or equivalently,
one leaf cell obtained after the $J_{\ell}$ refinement rounds on
the first $d-1$ coordinates.

Along any root-to-leaf path in this partitioning stage, each refinement
round uses exactly $d-1$ binary splits, one for each of the first
$d-1$ coordinates. Therefore, the total number of splits needed to
determine which terminal cube in $\mathcal{Q}_{\ell}$ contains $\bs x_{-d}$
is exactly $J_{\ell}(d-1)$.

We then add one final split at each terminal cube $Q\in\mathcal{Q}_{\ell}$
using the last coordinate $x_{d}$. Since $f_{\ell}$ is constant
on $Q$, the value $f_{\ell}(\bs u)$ remains the same for all $\bs u\in Q$
and equals $f^{*}(\bs u_{Q})$ by construction. Thus, once the path
has identified that $\bs x_{-d}\in Q$, we can perform a split on the last coordinate
at the threshold $f^{*}(\bs u_{Q})$. We send the observation to the left
child if $x_{d}\le f^{*}(\bs u_{Q})$ and to the right child otherwise.
We assign label $1$ to the left leaf and label $0$ to the right leaf.
Because $f_{\ell}(\bs x_{-d})=f^{*}(\bs u_{Q})$ whenever $\bs x_{-d}\in Q$,
this final split effectively implements exactly the rule $1\{x_{d}\le f_{\ell}(\bs x_{-d})\}$
for all $(\bs x_{-d},x_{d})\in Q\times[0,1]$.

Since this process is performed for every terminal cube $Q\in\mathcal{Q}_{\ell}$,
the resulting tree implements the policy $\pi_{\ell}(\bs x_{-d},x_{d})=1\{x_{d}\le f_{\ell}(\bs x_{-d})\}$
across the entire space. The depth of the tree is $J_{\ell}(d-1)+1$, because the first
$J_{\ell}(d-1)$ splits determine the cube containing $\bs x_{-d}$
and the final split compares $x_{d}$ with the constant value associated
with that cube. By definition, $J_{\ell}(d-1)+1\le\ell$, so $\pi_{\ell}\in\Pi_{\ell}$.

From Lemma~\ref{lem:sdb-to-policy}, $\inf_{\pi\in\Pi_{\ell}}d_{\Delta}(\pi,\pi^{*})\le d_{\Delta}(\pi_{\ell},\pi^{*})\le A_{X}\|f_{\ell}-f^{*}\|_{\infty}\le A_{X}L_{f}2^{-J_{\ell}}$.
The bound for the welfare approximation error follows from (\ref{eq:approx-error-transfer}).

To prove the VC-dimension bound, we first reduce the problem
to complete trees. A tree in the class $\Pi_{\ell}$ may terminate early at some branches, resulting in leaves appearing at a depth of $r<\ell$. To address this, we can attach
a full binary subtree of depth $\ell-r$ below each leaf that appears at depth $r$. All newly created leaves in this subtree will be assigned the same
label as the original leaf. This adjustment does not alter the policy
implemented by the tree: once the original leaf is reached, the continuation will still return the same label. Thus, for the purpose
of bounding the VC dimension, we can focus on complete binary
trees of depth $\ell$, where every root-to-leaf path has a length of exactly $\ell$. Such a tree contains $2^{\ell}-1$ internal nodes and
$2^{\ell}$ leaves. Here, an internal node refers to a non-terminal splitting
node, while a leaf denotes a terminal node that carries a label from the set $\{0,1\}$.
Figure~\ref{fig:dt-complete-tree} illustrates the padding argument, which
reduces the problem to complete binary trees of depth $\ell$
without changing the induced policy.

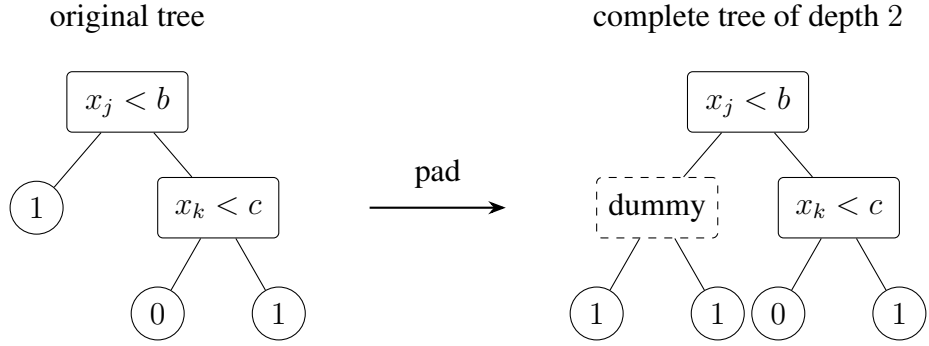
\begin{figure}[t]

\centering
\begin{tikzpicture}[>=Stealth, x=1cm, y=1cm,
  split/.style={draw, rounded corners=2pt, minimum width=16mm, minimum height=8mm, align=center},
  leaf/.style={draw, circle, minimum size=7mm, inner sep=0pt},
  dummy/.style={draw, dashed, rounded corners=2pt, minimum width=16mm, minimum height=8mm, align=center}]

  \node at (-4.1,1.1) {original tree};
  \node[split] (A) at (-4.1,0) {$x_j<b$};
  \node[leaf]  (B) at (-5.3,-1.4) {$1$};
  \node[split] (C) at (-2.9,-1.4) {$x_k<c$};
  \node[leaf]  (D) at (-3.7,-2.8) {$0$};
  \node[leaf]  (E) at (-2.1,-2.8) {$1$};
  \draw (A) -- (B);
  \draw (A) -- (C);
  \draw (C) -- (D);
  \draw (C) -- (E);

  \draw[->, thick] (-0.9,-1.4) -- (0.9,-1.4);
  \node at (0,-0.95) {pad};

  \node at (4.1,1.1) {complete tree of depth $2$};
  \node[split] (A2) at (4.1,0) {$x_j<b$};
  \node[dummy] (B2) at (2.9,-1.4) {dummy};
  \node[split] (C2) at (5.3,-1.4) {$x_k<c$};
  \node[leaf]  (F) at (2.1,-2.8) {$1$};
  \node[leaf]  (G) at (3.7,-2.8) {$1$};
  \node[leaf]  (D2) at (4.5,-2.8) {$0$};
  \node[leaf]  (E2) at (6.1,-2.8) {$1$};
  \draw (A2) -- (B2);
  \draw (A2) -- (C2);
  \draw (B2) -- (F);
  \draw (B2) -- (G);
  \draw (C2) -- (D2);
  \draw (C2) -- (E2);
\end{tikzpicture}
\caption{Expanding an early leaf to achieve a complete tree of depth $\ell$ without altering the induced policy. This reduction is utilized in the VC-dimension argument.}
\label{fig:dt-complete-tree}
\end{figure}

To analyze a sample set $S=\{\bs x_{1},\ldots,\bs x_{m}\}\subset[0,1]^{d}$, we need to determine how many distinct binary labelings of these $m$ sample
points can be generated by trees in $\Pi_{\ell}$.

We begin by considering a single internal node of the tree. The decision to split at that node involves choosing a coordinate $j\in\{1,\ldots,d\}$ and a threshold $b\in\mathbb{R}$. A sample point $\bs x$ will be sent to the left child if $x_{j}<b$
and to the right child otherwise. Focusing on the chosen coordinate $j$, the
threshold $b$ can only create changes in the left/right split of the sample when it passes one of the sample values $x_{1j},\ldots,x_{mj}$. Thus, for a fixed coordinate, there
can be at most $m+1$ distinct ways to divide the sample points into
left and right groups. Since there are $d$ coordinate options available, a single internal node can lead to at most $d(m+1)$ different ways to split sample points.

A complete tree has $2^{\ell}-1$ internal nodes. As a result, the
decisions made at all internal nodes can lead to at most $\{d(m+1)\}^{2^{\ell}-1}$
different ways of sending the sample points down the tree. This means there are at most $\{d(m+1)\}^{2^{\ell}-1}$ different ways to determine which leaf each sample point ultimately reaches.
Once these paths are established, each of the $2^{\ell}$ leaves can be assigned
label $0$ or label $1$ independently. Therefore, the total number of binary
labelings of the set $S$ induced by $\Pi_{\ell}$ is at most $2^{2^{\ell}}\{d(m+1)\}^{2^{\ell}-1}$.

If $S$ were to be completely shattered by $\Pi_{\ell}$, then every one of the $2^{m}$
possible binary labelings of the sample would need to be achievable.
Consequently, we must have $2^{m}\le2^{2^{\ell}}\{d(m+1)\}^{2^{\ell}-1}$. Taking logarithms gives $m\log2\le2^{\ell}\log2+(2^{\ell}-1)\log d+(2^{\ell}-1)\log(m+1)$, and hence $m\log2\le2\cdot2^{\ell}\log2+2^{\ell}\log d+2^{\ell}\log(m+1)$.
Next, we will eliminate the remaining $\log(m+1)$ term from the right-hand
side. First, we note that $\log(m+1)\le\log(2m)$. We can rewrite $\log(2m)=\log(4\cdot2^{\ell})+\log\!\bigl(m/(2\cdot2^{\ell})\bigr)$.
Furthermore, for every $t>0$, we know that $\log t\le t/2$. By applying this inequality with $t=m/(2\cdot2^{\ell})$, we obtain $\log(m+1)\le\log(4\cdot2^{\ell})+m/(4\cdot2^{\ell})$. Substituting this bound back and rearranging gives \(\bigl(\log2-\frac{1}{4}\bigr)m\le2^{\ell}\bigl(2\log2+\log d+\log(2^{\ell+2})\bigr)\).
Since $\log2-\tfrac{1}{4}>0$, we conclude that $m\le C\cdot2^{\ell}\log(2^{\ell}d)$
for some universal constant $C>0$. Thus, no sample larger
than $C\cdot2^{\ell}\log(2^{\ell}d)$ can be shattered by $\Pi_{\ell}$. Therefore, $\VC(\Pi_{\ell})\le C\cdot2^{\ell}\log(2^{\ell}d)\leq C\cdot2^{\ell}(\ell+\log d)$.
This completes the proof.
\end{proof}

\subsubsection{Example~\ref{exa:Neural Networks}: deep neural networks}

Let $\mathcal{F}_{\mathrm{DNN},\ell}$ denote the class of fully connected
feedforward ReLU networks defined on the domain $[0,1]^{d}$, with a specified width $\mathcal{H}_{\ell}$
and depth $\mathcal{D}_{\ell}$ (see \eqref{eq:DNN class}). We define the threshold-policy sieve $\Pi_{\mathrm{DNN},\ell}=\{\pi_{g}(\bs x)=1\{g(\bs x)\ge0\}:g\in\mathcal{F}_{\mathrm{DNN},\ell}\}$.

\begin{prop}[Approximation rate and VC dimension for the neural-network sieve]
\label{prop:nn-approx-vc} Assuming that Assumptions~\ref{ass:high-level}
and \ref{ass:sdb} hold, let $\Pi_{\ell}=\Pi_{\mathrm{DNN},\ell}$. For all large $\mathcal{H}_{\ell}$ and $\mathcal{D}_{\ell}$,
\[
W(\pi^{*})-\max_{\pi\in\Pi_{\ell}}W(\pi)\lesssim\left(\frac{\mathcal{H}_{\ell}}{\log\mathcal{H}_{\ell}}\right)^{-2s/(d-1)}\left(\frac{\mathcal{D}_{\ell}}{\log\mathcal{D}_{\ell}}\right)^{-2s/(d-1)},\quad
\VC(\Pi_{\ell})\lesssim\mathcal{D}^{2}_{\ell}\mathcal{H}^{2}_{\ell}\log\left(\mathcal{D}_{\ell}\mathcal{H}^{2}_{\ell}\right).
\]
\end{prop}
\begin{proof}
We begin by proving the approximation error bound. We apply the results from \citet[Corollary~1.2]{lu2021Deep}
to the decision-boundary function $f^{*}:[0,1]^{d-1}\to[0,1]$ as stated in
Assumption~\ref{ass:sdb}. Since $f^{*}\in C^{s}([0,1]^{d-1})$ and
$\|f^{*}\|_{C^{s}}\le M$, this result indicates that there exists a
ReLU network $\phi_{\ell}$ defined on $[0,1]^{d-1}$ with width at most $\mathcal{H}_{\ell}-1$ and depth $\mathcal{D}_{\ell}$. For large $\mathcal{H}_{\ell}$ and $\mathcal{D}_{\ell}$, we have:
\[
\|\phi_{\ell}-f^{*}\|_{\infty}\le C_{s,d}M\left(\frac{\mathcal{H}_{\ell}}{\log(8\mathcal{H}_{\ell}+8)}\right)^{-2s/(d-1)}\left(\frac{\mathcal{D}_{\ell}}{\log(4\mathcal{D}_{\ell}+4)}\right)^{-2s/(d-1)},
\]
after absorbing fixed depth-convention differences into $C_{s,d}$.

We will next convert $\phi_{\ell}$ into a neural network defined on the domain $[0,1]^{d}$.
We define $g_{\ell}(\bs x_{-d},x_{d}):=\phi_{\ell}(\bs x_{-d})-x_{d}$.
Since $x_{d}\in[0,1]$, we can propagate the identity function $x_{d}\mapsto x_{d}$ through a ReLU network. This can be achieved by dedicating a neuron in each
hidden layer to carry the value of $x_{d}$. In the first hidden layer,
this neuron outputs $\sigma(x_{d})=x_{d}$, and in each subsequent
hidden layer, applying the activation function $\sigma$ leaves the value unchanged.
Thus, the function $g_{\ell}$ can be implemented by a ReLU network on $[0,1]^{d}$
with depth $\mathcal{D}_{\ell}$ and width at most $\mathcal{H}_{\ell}$. Since $\mathcal{H}_{\ell}-1\asymp\mathcal{H}_{\ell}$ for large $\mathcal{H}_{\ell}$, the approximation rate is unchanged. Therefore, we can consider $g_{\ell}$ to be an element of the space
$\mathcal{F}_{\mathrm{DNN},\ell}$.

Define $\pi_{\ell}(\bs x_{-d},x_{d}):=1\{g_{\ell}(\bs x_{-d},x_{d})\ge0\}=1\{x_{d}\le\phi_{\ell}(\bs x_{-d})\}$.
Then $\pi_{\ell}\in\Pi_{\mathrm{DNN},\ell}$. By Lemma~\ref{lem:sdb-to-policy}, $d_{\Delta}(\pi_{\ell},\pi^{*})\le A_{X}\|\phi_{\ell}-f^{*}\|_{\infty}$. Substituting the approximation error bound for $\phi_{\ell}$, with
\[
r_{\ell}:=\left(\frac{\mathcal{H}_{\ell}}{\log(8\mathcal{H}_{\ell}+8)}\right)^{-2s/(d-1)}\left(\frac{\mathcal{D}_{\ell}}{\log(4\mathcal{D}_{\ell}+4)}\right)^{-2s/(d-1)},
\]
gives $d_{\Delta}(\pi_{\ell},\pi^{*})\le C_{s,d}A_{X}Mr_{\ell}$. Since $\pi_{\ell}\in\Pi_{\mathrm{DNN},\ell}$, $\inf_{\pi\in\Pi_{\mathrm{DNN},\ell}}d_{\Delta}(\pi,\pi^{*})\le C_{s,d}A_{X}Mr_{\ell}$.
The welfare approximation error bound follows from (\ref{eq:approx-error-transfer}).

To prove the VC-dimension bound, we observe that $\mathcal{D}_{\ell}$
counts the hidden layers in (\ref{eq:DNN class}), a fully connected network. A network with a width of $\mathcal{H}_{\ell}$ and a depth of $\mathcal{D}_{\ell}$ has $\mathcal{D}_{\ell}$ hidden layers plus one
output layer, resulting in $\mathcal{D}_{\ell}+1$ weight matrices. The total number of trainable parameters in such a network is given by the formula  $
(d+1)\mathcal{H}_{\ell}+(\mathcal{D}_{\ell}-1)(\mathcal{H}_{\ell}+1)\mathcal{H}_{\ell}+(\mathcal{H}_{\ell}+1)$, provided that $\mathcal{D}_{\ell}\ge1$. Therefore, there exists a $d$-dependent $C_d$ such that the total number
of trainable parameters can be bounded by $C_{d}\mathcal{D}_{\ell}\mathcal{H}^{2}_{\ell}$
for large $\mathcal{H}_{\ell}$ and $\mathcal{D}_{\ell}$.

According to a result from \citet{bartlett2019Nearlytight}, the VC dimension of the sign class associated with a piecewise-linear network is bounded by a universal
constant multiplied by the number of parameters, the number of layers, and the logarithm of the number of parameters. In this context, we note that the number of layers is bounded by $\mathcal{D}_{\ell}+1=O(\mathcal{D}_{\ell})$. Furthermore, the number of parameters is at most $C_{d}\mathcal{D}_{\ell}\mathcal{H}^{2}_{\ell}$. Using this information, we can derive the following result:
\[
\VC(\Pi_{\ell})\le C\bigl(C_{d}\mathcal{D}_{\ell}\mathcal{H}^{2}_{\ell}\bigr)(\mathcal{D}_{\ell}+1)\log\!\bigl(C_{d}\mathcal{D}_{\ell}\mathcal{H}^{2}_{\ell}\bigr)\lesssim\mathcal{D}^{2}_{\ell}\mathcal{H}^{2}_{\ell}\log\left(\mathcal{D}_{\ell}\mathcal{H}^{2}_{\ell}\right)
\]
for some generic constant $C>0$. This completes the proof.
\end{proof}

\section{Some useful lemmas}
\begin{lem}
[Nearness of argmins of convex functions]\label{lem:Nearness of the argmins of the convex functions} Consider the functions
$A_{n}(\bs s;\pi):\mathcal{S}\times\Pi\to\mathbb{R}$ and $B_{n}(\bs s;\pi):\mathcal{S}\times\Pi_{\infty}\to\mathbb{R}$,
which represent two sequences of random functions. Here, $\Pi\subset\Pi_{\infty}$
is a policy class and $\mathcal{S}\subset\R^{p}$ is an open convex
set. For every policy $\pi\in\Pi$, the function $A_{n}(\bs s;\pi)$ is convex with respect
to $\bs s$. Let $\bs{\alpha}_{n}(\pi)$ be a measurable minimizer of $A_{n}(\bs s;\pi)$ for any $\pi \in \Pi$, and assume that $B_{n}(\bs s;\pi)$ has a unique minimum at $\bs{\beta}_{n}(\pi)$. Then, for each $\delta\geq0$, we have:
\[
P\left(\sup_{\pi\in\Pi}\left\Vert \bs{\alpha}_{n}(\pi)-\bs{\beta}_{n}(\pi)\right\Vert >\delta\right)\leq P\left(\Delta_{n}(\delta;\pi)\geq h_{n}(\delta;\pi),\ \exists\pi\in\Pi\right),
\]
where $h_{n}(\delta;\pi)=\inf_{\bs s\in\mathcal S:\left\Vert \bs s-\bs{\beta}_{n}(\pi)\right\Vert =\delta}B_{n}(\bs s;\pi)-B_{n}(\bs{\beta}_{n}(\pi);\pi)$, and
\[
\Delta_{n}(\delta;\pi)=\sup_{\bs s\in\mathcal S:\left\Vert \bs s-\bs{\beta}_{n}(\pi)\right\Vert =\delta}\left|A_{n}(\bs s;\pi)-B_{n}(\bs s;\pi)-\left\{ A_{n}(\bs{\beta}_{n}(\pi);\pi)-B_{n}(\bs{\beta}_{n}(\pi);\pi)\right\} \right|.
\]
\end{lem}
\begin{proof}
The proof is implicit in the reference \citet[Lemma 2]{hjort1993Asymptotics}. However, we will provide a detailed proof here for completeness. The case $\delta=0$ is immediate, so take $\delta>0$. To begin, we fix any policy $\pi\in\Pi$.
Let $\bs s\in\left\{ \bs x\in\mathcal S:\left\Vert \bs x-\bs{\beta}_{n}(\pi)\right\Vert \geq\delta\right\} $
represent an arbitrary point that is located either on or outside the ball surrounding $\bs{\beta}_{n}(\pi)$ with radius $\delta$. Specifically, we define $\bs s=\bs{\beta}_{n}(\pi)+l\cdot\bs u$,
where $\bs u$ is a unit vector and $l\geq\delta$. Due to the convexity
of $A_{n}(\bs s;\pi)$, we have
\[
(1-\delta/l)A_{n}(\bs{\beta}_{n}(\pi);\pi)+(\delta/l)A_{n}(\bs s;\pi)\geq A_{n}((1-\delta/l)\bs{\beta}_{n}(\pi)+(\delta/l)\bs s;\pi)=A_{n}(\bs{\beta}_{n}(\pi)+\delta\cdot\bs u;\pi).
\]
By rearranging the inequality above, we obtain:
\begin{align*}
 & (\delta/l)\left\{ A_{n}(\bs s;\pi)-A_{n}(\bs{\beta}_{n}(\pi);\pi)\right\} \geq A_{n}(\bs{\beta}_{n}(\pi)+\delta\cdot\bs u;\pi)-A_{n}(\bs{\beta}_{n}(\pi);\pi)\\
= & B_{n}(\bs{\beta}_{n}(\pi)+\delta\cdot\bs u;\pi)-B_{n}(\bs{\beta}_{n}(\pi);\pi)+A_{n}(\bs{\beta}_{n}(\pi)+\delta\cdot\bs u;\pi)-B_{n}(\bs{\beta}_{n}(\pi)+\delta\cdot\bs u;\pi)\\
 & \qquad-A_{n}(\bs{\beta}_{n}(\pi);\pi)+B_{n}(\bs{\beta}_{n}(\pi);\pi)\geq h_{n}(\delta;\pi)-\Delta_{n}(\delta;\pi).
\end{align*}
If $h_{n}(\delta;\pi)>\Delta_{n}(\delta;\pi)$, then $A_{n}(\bs s;\pi)>A_{n}(\bs{\beta}_{n}(\pi);\pi)$
for all $\bs s$ located either on or outside the $\delta$-ball,
which implies that
\(
\left\{ h_{n}(\delta;\pi)>\Delta_{n}(\delta;\pi)\right\} \subset\left\{ \left\Vert \bs{\alpha}_{n}(\pi)-\bs{\beta}_{n}(\pi)\right\Vert <\delta\right\}
\).
As a result,
\(
\left\{ h_{n}(\delta;\pi)>\Delta_{n}(\delta;\pi),\ \forall\pi\in\Pi\right\} \subset\left\{ \left\Vert \bs{\alpha}_{n}(\pi)-\bs{\beta}_{n}(\pi)\right\Vert <\delta,\ \forall\pi\in\Pi\right\}
\), which leads to
\begin{align*}
\left\{ \sup_{\pi\in\Pi}\left\Vert \bs{\alpha}_{n}(\pi)-\bs{\beta}_{n}(\pi)\right\Vert >\delta\right\}  & \subset\left\{ \left\Vert \bs{\alpha}_{n}(\pi)-\bs{\beta}_{n}(\pi)\right\Vert \geq\delta,\ \exists\pi\in\Pi\right\} \\
 & \subset\left\{ \Delta_{n}(\delta;\pi)\geq h_{n}(\delta;\pi),\ \exists\pi\in\Pi\right\} .
\end{align*}
This completes the proof.
\end{proof}
\begin{lem}
[Uniform entropy for the transformed function class]\label{lem:transformed uniform entropy} Let $\mathcal{F}_{1},\ldots,\mathcal{F}_{M}$ represent classes of measurable
functions $S\to\R$, equipped with measurable envelopes $F_{1},\ldots,F_{M}$,
respectively. Also, let $\phi:\R^{M}\to\R$ be a map satisfying:
\begin{equation}
\left|\phi\circ\bs f(x)-\phi\circ\bs g(x)\right|^{2}\leq\sum_{m=1}^{M}L_{m}^{2}(x)\left|f_{m}(x)-g_{m}(x)\right|^{2},\label{eq:transformed uniform entropy condition}
\end{equation}
for all $\bs f=(f_{1},\ldots,f_{M}),\boldsymbol{g}=(g_{1},\ldots,g_{M})\in\mathcal{F}_{1}\times\cdots\times\mathcal{F}_{M}=:\mathcal{F}$
and every $x\in S$, where $L_{1},\ldots,L_{M}$ are non-negative
measurable functions on $S$. Let $\phi\circ\mathcal{F}:=\left\{ \phi\circ\bs f:\bs f\in\mathcal{F}\right\} $
represent a class of functions. Additionally, define $L\cdot F(x):=\sqrt{\sum_{m=1}^{M}L_{m}^{2}(x)F_{m}^{2}(x)}$
as a measurable function. Then,
\[
\sup_{Q}\log N\left(\epsilon\left\Vert L\cdot F\right\Vert _{Q,2},\phi\circ\mathcal{F},\left\Vert \cdot\right\Vert _{Q,2}\right)\leq\sum_{m=1}^{M}\sup_{R_{m}}\log N\left(\epsilon\left\Vert F_{m}\right\Vert _{R_{m},2},\mathcal{F}_{m},\left\Vert \cdot\right\Vert _{R_{m},2}\right),
\]
for all $0<\epsilon\leq1$, where the suprema are taken over all finite
discrete probability measures on $(S,\mathcal{S})$. Additionally, let $\Phi(x):=\left|\phi\circ\bs f_{0}(x)\right|+2\sqrt{\sum_{m=1}^{M}L_{m}^{2}(x)F_{m}^{2}(x)}$
be a measurable function, where $\bs f_{0}:=(f_{01},\ldots,f_{0M})$
is any function in $\mathcal{F}$. Then, $\Phi$ serves as an envelope of
$\phi\circ\mathcal{F}$ and
\[
\sup_{Q}\log N\left(\epsilon\left\Vert \Phi\right\Vert _{Q,2},\phi\circ\mathcal{F},\left\Vert \cdot\right\Vert _{Q,2}\right)\leq\sum_{m=1}^{M}\sup_{R_{m}}\log N\left(\epsilon\left\Vert F_{m}\right\Vert _{R_{m},2},\mathcal{F}_{m},\left\Vert \cdot\right\Vert _{R_{m},2}\right),
\]
for all $0<\epsilon\leq1$, where the suprema are taken over all finite
discrete probability measures on $(S,\mathcal{S})$.
\end{lem}
\begin{proof}
The proof is implicit in \citet[p.~199]{vaart1996Weak}; see also \citet[Lemma~A.6]{chernozhukov2014Gaussian} and \citet[Lemma~9.13]{kosorok2008Introduction}.
\end{proof}

\section{Proofs for Section \ref{sec:A-General-Learning}}\label{sec:proofs-general-learning}

\begin{proof}[Proof of Theorem \ref{thm:oracle inequality-general}]

Let $n=N/K$ denote the common fold size. Define
\[
R_{CV}(\ell):=K^{-1}\sum_{k=1}^{K}\{ \widehat{W}_{I_{k}}(\widehat{\pi}_{\ell,I_{-k}})-\log\ell/\sqrt{N}\},\quad
\mathcal{E}_{k\ell}:=\widehat{W}_{I_{k}}(\widehat{\pi}_{\ell,I_{-k}})-W(\widehat{\pi}_{\ell,I_{-k}}).
\]
By \eqref{eq:def of l_hat-general}, $\widehat{\ell}=\arg\max_{\ell=1,2,\ldots}R_{CV}(\ell)$.
Then, for every $\ell\geq1$, it holds that
\begin{align}
W(\pi^{*})-W(\widehat{\pi}) & =W(\pi^{*})-R_{CV}(\widehat{\ell})+R_{CV}(\widehat{\ell})-W(\widehat{\pi})\nonumber\\
&\leq\underbrace{\left\{ W(\pi^{*})-R_{CV}(\ell)\right\} }_{\text{term 1}}+\underbrace{\left\{ R_{CV}(\widehat{\ell})-W(\widehat{\pi})\right\} }_{\text{term 2}}.\label{eq:decomposition of the regret}
\end{align}
We will handle term 1 and term 2 sequentially.

\textbf{Term 1.} Note that
\begin{equation}
\frac{1}{K}\sum_{k=1}^{K}W(\widehat{\pi}_{\ell,I_{-k}})-R_{CV}(\ell)=\frac{1}{K}\sum_{k=1}^{K}\left\{ W(\widehat{\pi}_{\ell,I_{-k}})-\widehat{W}_{I_{k}}(\widehat{\pi}_{\ell,I_{-k}})\right\} +\frac{\log\ell}{\sqrt{N}}\label{eq:term 1 decomposition}
\end{equation}
We handle every summand on the right-hand side of the above equality. Since $\widehat{\pi}_{\ell,I_{-k}}$ only depends on $\{\bs Z_{i}:=(\bs X_{i},Y_{i},T_{i})\}_{i\in I_{-k}}$, it remains independent of $\{\bs Z_{i}\}_{i\in I_{k}}$. For every $k=1,\ldots,K$, Assumption \ref{assu:general assu on hatW} gives
$P(|\mathcal{E}_{k\ell}|\geq\epsilon+C_{1}/\sqrt{n})\leq C_{2}\exp\{ -C_{3}n\epsilon^{2}\}$.
A standard integration argument then gives
$\e|\mathcal{E}_{k\ell}|\leq C_{1}/\sqrt{n}+\int_{0}^{\infty}C_{2}\exp\{ -C_{3}n\epsilon^{2}\} d\epsilon\leq C_{5}/\sqrt{N}$,
for every $\ell=1,2,\ldots$, where in the last step we let $C_{5}:=\sqrt{K}\left\{ C_{1}+\int_{0}^{\infty}C_{2}\cdot\exp\left\{ -C_{3}\epsilon^{2}\right\} d\epsilon\right\} <\infty$.
From the decomposition of (\ref{eq:term 1 decomposition}), we obtain the following results:
$\e[\text{term 1}]\leq W(\pi^{*})-K^{-1}\sum_{k=1}^{K}\e[W(\widehat{\pi}_{\ell,I_{-k}})]+\log\ell/\sqrt{N}+C_{5}/\sqrt{N}$
for all $\ell=1,2,\ldots$. Since the observed data are i.i.d., $\frac{1}{K}\sum_{k=1}^{K}\e\left[W(\widehat{\pi}_{\ell,I_{-k}})\right]=\e\left[W(\widehat{\pi}_{\ell,I_{-1}})\right]$. Since the decomposition above holds for every $\ell=1,2,\ldots$, we obtain
\begin{equation}
\e\left[\text{term 1}\right]\leq\inf_{\ell=1,2,\ldots}\left\{ W(\pi^{*})-\e\left[W(\widehat{\pi}_{\ell,I_{-1}})\right]+\frac{\log\ell}{\sqrt{N}}\right\} +\frac{C_{5}}{\sqrt{N}},\label{eq:result-for-term-1}
\end{equation}

\textbf{Term 2.} We now address term 2 in (\ref{eq:decomposition of the regret}).
Based on the definition of $\widehat{\pi}$ in (\ref{eq:CV criterion-general}),
we have:
\begin{align}
\text{term 2}=R_{CV}(\widehat{\ell})-W(\widehat{\pi}) & =\frac{1}{K}\sum_{k=1}^{K}\widehat{W}_{I_{k}}(\widehat{\pi}_{\widehat{\ell},I_{-k}})-W(\widehat{\pi})-\frac{\log\widehat{\ell}}{\sqrt{N}}\nonumber\\
&\leq\widehat{W}_{I_{\widehat{k}}}(\widehat{\pi}_{\widehat{\ell},I_{-\widehat{k}}})-W(\widehat{\pi}_{\widehat{\ell},I_{-\widehat{k}}})-\frac{\log\widehat{\ell}}{\sqrt{N}}.\label{eq:bound for term 2}
\end{align}
Since $\widehat{\pi}_{\ell,I_{-k}}$ only depends on $\{\bs Z_{i}\}_{i\in I_{-k}}$, it is independent of $\{\bs Z_{i}\}_{i\in I_{k}}$. For every $k=1,\ldots,K$ and $\ell=1,2,\ldots$, Assumption \ref{assu:general assu on hatW} gives
\[
P\left(|\mathcal{E}_{k\ell}|\geq\epsilon+\frac{\log\ell}{\sqrt{N}}+\frac{C_{1}}{\sqrt{n}}\right)\leq C_{2}\exp\left(-C_{3}\frac{\log^{2}\ell}{K}\right)\exp\left\{ -C_{3}n\epsilon^{2}\right\}.
\]
When combined with the union bound, this leads to
\begin{align}
 & P\left(\mathcal{E}_{\widehat{k}\widehat{\ell}}-\frac{\log\widehat{\ell}}{\sqrt{N}}\geq\epsilon+\frac{C_{1}}{\sqrt{n}}\right)\leq\sum_{k=1}^{K}\sum_{\ell=1}^{\infty}P\left(\mathcal{E}_{k\ell}\geq\epsilon+\frac{\log\ell}{\sqrt{N}}+\frac{C_{1}}{\sqrt{n}}\right)\nonumber \\
\leq & \sum_{\ell=1}^{\infty}K\cdot C_{2}\exp\left(-C_{3}\frac{\log^{2}\ell}{K}\right)\exp\left\{ -C_{3}n\epsilon^{2}\right\} =:\Delta\cdot\exp\left\{ -C_{3}n\epsilon^{2}\right\} \label{eq:bound-for-term-2-main-term}
\end{align}
for all $\epsilon\geq0$ and $n\geq C_{4}$, where in the last step
we define $\Delta:=\sum_{\ell=1}^{\infty}KC_{2}\exp\left(-C_{3}\frac{\log^{2}\ell}{K}\right)<\infty$.
Now, using a standard integration argument, we have
\[
\e\left[\mathcal{E}_{\widehat{k}\widehat{\ell}}-\frac{\log\widehat{\ell}}{\sqrt{N}}\right]\leq\frac{C_{1}}{\sqrt{n}}+\int_{0}^{\infty}\Delta\cdot\exp\left\{ -C_{3}n\epsilon^{2}\right\} d\epsilon\leq\frac{C_{6}}{\sqrt{N}},
\]
where $C_{6}:=\sqrt{K}\{ C_{1}+\int_{0}^{\infty}\Delta\cdot\exp( -C_{3}\epsilon^{2}) d\epsilon\}<\infty$.
Thus, we conclude that
\begin{equation}
\e\left[\text{term 2}\right]\leq\e\left[\mathcal{E}_{\widehat{k}\widehat{\ell}}-\frac{\log\widehat{\ell}}{\sqrt{N}}\right]\leq\frac{C_{6}}{\sqrt{N}}.\label{eq:result for term 2}
\end{equation}

\textbf{Aggregating the results} by combining (\ref{eq:decomposition of the regret}),
(\ref{eq:result-for-term-1}) and (\ref{eq:result for term 2}) leads to

\[
\begin{aligned}
\e\left[W(\pi^{*})-W(\widehat{\pi})\right]&\leq\inf_{\ell=1,2,\ldots}\left\{ W(\pi^{*})-\max_{\pi\in\Pi_{\ell}}W(\pi)+\max_{\pi\in\Pi_{\ell}}W(\pi)\right.\\
&\qquad\left.-\e\left[W(\widehat{\pi}_{\ell,I_{-1}})\right]+\frac{\log\ell}{\sqrt{N}}\right\} +\sqrt{\frac{C}{N}} .
\end{aligned}
\]
for large $N$, where $C>0$ is a constant.
\end{proof}

\begin{proof}[Proof of Corollary \ref{cor:oracle inequality prob bound}]

By applying the union bound, we obtain:
\begin{align*}
&P\left(\sup_{\ell\geq1,1\leq k\leq K}\left\{ \left|\widehat{W}_{I_{k}}(\widehat{\pi}_{\ell,I_{-k}})-W(\widehat{\pi}_{\ell,I_{-k}})\right|-\frac{\log\ell}{\sqrt{N}}\right\} \geq\epsilon+\frac{C_{1}}{\sqrt{n}}\right)\\
&\leq \sum_{\ell=1}^{\infty}\sum_{k=1}^{K}P\left(\left|\widehat{W}_{I_{k}}(\widehat{\pi}_{\ell,I_{-k}})-W(\widehat{\pi}_{\ell,I_{-k}})\right|-\frac{\log\ell}{\sqrt{N}}\geq\epsilon+\frac{C_{1}}{\sqrt{n}}\right)\\
&\leq \sum_{\ell=1}^{\infty}K\cdot C_{2}\exp\left(-C_{3}\frac{\log^{2}\ell}{K}\right)\exp\left\{ -C_{3}n\epsilon^{2}\right\} =\Delta\cdot\exp\left\{ -C_{3}n\epsilon^{2}\right\} ,
\end{align*}
where $\Delta=\sum_{\ell=1}^{\infty}KC_{2}\exp\left(-C_{3}\frac{\log^{2}\ell}{K}\right)<\infty$. We note that term 1 defined in (\ref{eq:decomposition of the regret}) can be bounded as
\begin{align*}
  & \text{term 1}=W(\pi^{*})-\frac{1}{K}\sum_{k=1}^{K}W(\widehat{\pi}_{\ell,I_{-k}})+\frac{1}{K}\sum_{k=1}^{K}W(\widehat{\pi}_{\ell,I_{-k}})-R_{CV}(\ell)\\
 = & W(\pi^{*})-\frac{1}{K}\sum_{k=1}^{K}W(\widehat{\pi}_{\ell,I_{-k}})+\frac{1}{K}\sum_{k=1}^{K}\left\{ W(\widehat{\pi}_{\ell,I_{-k}})-\widehat{W}_{I_{k}}(\widehat{\pi}_{\ell,I_{-k}})\right\} +\frac{\log\ell}{\sqrt{N}}\\
 \leq & W(\pi^{*})-\frac{1}{K}\sum_{k=1}^{K}W(\widehat{\pi}_{\ell,I_{-k}})+\sup_{1\leq k\leq K}\left|W(\widehat{\pi}_{\ell,I_{-k}})-\widehat{W}_{I_{k}}(\widehat{\pi}_{\ell,I_{-k}})\right|+\frac{\log\ell}{\sqrt{N}}\\
 \leq & W(\pi^{*})-\frac{1}{K}\sum_{k=1}^{K}W(\widehat{\pi}_{\ell,I_{-k}})+2\frac{\log\ell}{\sqrt{N}}+\sup_{\ell\geq1,1\leq k\leq K}\left\{ \left|\widehat{W}_{I_{k}}(\widehat{\pi}_{\ell,I_{-k}})-W(\widehat{\pi}_{\ell,I_{-k}})\right|-\frac{\log\ell}{\sqrt{N}}\right\} .
 \end{align*}
Combining the two inequalities with (\ref{eq:decomposition of the regret})
and (\ref{eq:bound for term 2})--(\ref{eq:bound-for-term-2-main-term})
leads to the conclusion that
\[
W(\pi^{*})-W(\widehat{\pi})\leq\inf_{\ell\geq1}\left\{ W(\pi^{*})-\frac{1}{K}\sum_{k=1}^{K}W(\widehat{\pi}_{\ell,I_{-k}})+2\frac{\log\ell}{\sqrt{N}}\right\} +\sqrt{\frac{C+\delta}{N}}
\]
w.p. at least $1-C_{1}\exp(-C_{2}\delta)$ for all
$\delta>0$ and large $N$, where $C>0$ is a constant.
\end{proof}

\begin{proof}[Proof of Corollary \ref{cor:estimation-error-bound-general}]

First, establish an arbitrary deterministic sequence $\{\ell_{N}\}_{N\geq1}$
such that $\VC(\Pi_{\ell_{N}})/N\to0$. It follows from Theorem~\ref{thm:oracle inequality-general} that:
\[
\begin{aligned}
\e\left[W(\pi^{*})-W(\widehat{\pi})\right]&\leq W(\pi^{*})-\max_{\pi\in\Pi_{\ell_{N}}}W(\pi)+\max_{\pi\in\Pi_{\ell_{N}}}W(\pi)\\
&\quad-\e\left[W(\widehat{\pi}_{\ell_{N},I_{-1}})\right]+\frac{\log\ell_{N}}{\sqrt{N}}+\sqrt{\frac{C}{N}} .
\end{aligned}
\]
Let $\pi_{\ell_{N}}^{*}\in\arg\max_{\pi\in\Pi_{\ell_{N}}}W(\pi)$. According to the definition of $\widehat{\pi}_{\ell_{N},I_{-1}}$, we can derive:
\[
\begin{aligned}
W(\pi_{\ell_{N}}^{*})-W(\widehat{\pi}_{\ell_{N},I_{-1}})
&\leq\left\{ W(\pi_{\ell_{N}}^{*})-\widehat{W}_{I_{-1}}(\pi_{\ell_{N}}^{*})\right\}\!+\!\left\{\widehat{W}_{I_{-1}}(\widehat{\pi}_{\ell_{N},I_{-1}})-W(\widehat{\pi}_{\ell_{N},I_{-1}})\right\}\\
&\leq2\sup_{\pi\in\Pi_{\ell_{N}}}\left|\widehat{W}_{I_{-1}}(\pi)-W(\pi)\right|.
\end{aligned}
\]
Taking expectations on both sides of the inequality leads us to the conclusion that:
\[
\max_{\pi\in\Pi_{\ell_{N}}}W(\pi)-\e\left[W(\widehat{\pi}_{\ell_{N},I_{-1}})\right]
\leq
2\e\left[\sup_{\pi\in\Pi_{\ell_{N}}}\left|\widehat{W}_{I_{-1}}(\pi)-W(\pi)\right|\right].
\]
Since $|I_{-1}|=(K-1)N/K$ and $\VC(\Pi_{\ell_{N}})/|I_{-1}|\to0$,
Assumption~\ref{assu:general ass uniform convergence} implies that
\[
P\left(\sup_{\pi\in\Pi_{\ell_{N}}}\left|\widehat{W}_{I_{-1}}(\pi)-W(\pi)\right|\geq\delta+C_{1}\sqrt{\frac{\VC(\Pi_{\ell_{N}})}{|I_{-1}|}}\right)\leq C_{2}\exp(-C_{3}|I_{-1}|\delta^{2})
\]
for all $\delta>0$ and large $N$. Integrating this tail bound yields
\[
\begin{aligned}
  &\e\left[\sup_{\pi\in\Pi_{\ell_{N}}}\left|\widehat{W}_{I_{-1}}(\pi)-W(\pi)\right|\right]
\leq C_{1}\sqrt{\frac{\VC(\Pi_{\ell_{N}})}{|I_{-1}|}}+\int_{0}^{\infty}C_{2}\exp(-C_{3}|I_{-1}|\delta^{2})d\delta\\
&\leq C_{1}\sqrt{\frac{K}{K-1}}\sqrt{\frac{\VC(\Pi_{\ell_{N}})}{N}}+\sqrt{\frac{K}{K-1}}\sqrt{\frac{1}{N}}\int_{0}^{\infty}C_{2}\exp(-C_{3}\delta^{2})d\delta .
\end{aligned}
\]
Combining the previous inequalities and renaming constants as needed,
we conclude that for every deterministic sequence $\{\ell_{N}\}_{N\geq1}$
satisfying $\VC(\Pi_{\ell_{N}})/N\to0$, the following holds:
\begin{equation}
\begin{aligned}
\e\left[W(\pi^{*})-W(\widehat{\pi})\right]&\leq W(\pi^{*})-\max_{\pi\in\Pi_{\ell_{N}}}W(\pi)+C^{\prime}\sqrt{\frac{K}{K-1}}\sqrt{\frac{\VC(\Pi_{\ell_{N}})}{N}}+\frac{\log\ell_{N}}{\sqrt{N}}+\sqrt{\frac{C}{N}} .
\end{aligned}\label{eq:deterministic-sequence-bound-corollary}
\end{equation}

We still need to select an admissible sequence. Define
\[
\Psi_{N}(\ell):=W(\pi^{*})-\max_{\pi\in\Pi_{\ell}}W(\pi)+C^{\prime}\sqrt{\frac{K}{K-1}}\sqrt{\frac{\VC(\Pi_{\ell})}{N}}+\frac{\log\ell}{\sqrt{N}},
\]
and let $\ell_{N}$ be the smallest minimizer of $\Psi_{N}(\ell)$ over $\ell=1,2,\ldots$. A minimizer exists because all three terms are non-negative and $\log\ell/\sqrt{N}\to\infty$ as $\ell\to\infty$ for each fixed $N$. Next, we demonstrate that $\VC(\Pi_{\ell_{N}})/N\to0$. Fix any $\epsilon>0$. By the standing condition on the approximating sequence, there exists a fixed integer $L$ such that $W(\pi^{*})-\max_{\pi\in\Pi_{L}}W(\pi)<\epsilon/2$. Since $\VC(\Pi_{L})<\infty$, for all large $N$ we also have $C^{\prime}\sqrt{K/(K-1)}\sqrt{\VC(\Pi_{L})/N}+\log L/\sqrt{N}<\epsilon/2$. Thus $\Psi_{N}(\ell_{N})\le\Psi_{N}(L)<\epsilon$ for all large $N$, so $\Psi_{N}(\ell_{N})\to0$. Since $0\leq C^{\prime}\sqrt{K/(K-1)}\sqrt{\VC(\Pi_{\ell_{N}})/N}\leq\Psi_{N}(\ell_{N})$, we have $\VC(\Pi_{\ell_{N}})/N\to0$. It follows that $\ell_{N}$ satisfies the condition for (\ref{eq:deterministic-sequence-bound-corollary}). In addition, $\Psi_{N}(\ell_{N})=\inf_{\ell=1,2,\ldots}\Psi_{N}(\ell)$.
Substituting this identity into
(\ref{eq:deterministic-sequence-bound-corollary}), we can conclude:
\[
\begin{aligned}
\e\left[W(\pi^{*})-W(\widehat{\pi})\right]
&\leq \inf_{\ell}\left\{W(\pi^{*})-\max_{\pi\in\Pi_{\ell}}W(\pi)+C^{\prime}\sqrt{\frac{K}{K-1}}\sqrt{\frac{\VC(\Pi_{\ell})}{N}}+\frac{\log\ell}{\sqrt{N}}\right\}+\sqrt{\frac{C}{N}}.
\end{aligned}
\]
\end{proof}

\begin{proof}[Proof of Remark \ref{rem:finite L}]

The proof closely resembles that of Theorem \ref{thm:oracle inequality-general}.
Let $n=N/K$ denote the common fold size, let $L_N\ge2$, and define $\mathcal{E}_{k\ell}:=\widehat{W}_{I_{k}}(\widehat{\pi}_{\ell,I_{-k}})-W(\widehat{\pi}_{\ell,I_{-k}})$ and $r_N:=C^{*}\sqrt{\log L_N/N}$. For $\ell=1,\ldots,L_{N}$, set $R_{CV}(\ell):=K^{-1}\sum_{k=1}^{K}\widehat{W}_{I_{k}}(\widehat{\pi}_{\ell,I_{-k}})$; then $\widehat{\ell}=\arg\max_{\ell=1,\ldots,L_{N}}R_{CV}(\ell)$.
Note that for every $\ell=1,\ldots,L_{N}$, the following statement holds:
\begin{align}
W(\pi^{*})-W(\widehat{\pi}) & =W(\pi^{*})-R_{CV}(\widehat{\ell})+R_{CV}(\widehat{\ell})-W(\widehat{\pi})\nonumber\\
&\leq\underbrace{\left\{ W(\pi^{*})-R_{CV}(\ell)\right\} }_{\text{term 1}}+\underbrace{\left\{ R_{CV}(\widehat{\ell})-W(\widehat{\pi})\right\} }_{\text{term 2}}.\label{eq:decomposition of the regret-1}
\end{align}
We will discuss term 1 and term 2 one at a time.

\textbf{term 1.} Note that
\begin{equation}
\frac{1}{K}\sum_{k=1}^{K}W(\widehat{\pi}_{\ell,I_{-k}})-R_{CV}(\ell)=\frac{1}{K}\sum_{k=1}^{K}\left\{ W(\widehat{\pi}_{\ell,I_{-k}})-\widehat{W}_{I_{k}}(\widehat{\pi}_{\ell,I_{-k}})\right\} \label{eq:term 1 decomposition-1}
\end{equation}
We address each term on the right-hand side of the above equality. Since $\widehat{\pi}_{\ell,I_{-k}}$ only relies on $\{\bs Z_{i}:=(\bs X_{i},Y_{i},T_{i})\}_{i\in I_{-k}}$, it remains independent of $\{\bs Z_{i}\}_{i\in I_{k}}$. For every $k=1,\ldots,K$, Assumption \ref{assu:general assu on hatW} gives
$P(|\mathcal{E}_{k\ell}|\geq\epsilon+C_{1}/\sqrt{n})\leq C_{2}\exp\{ -C_{3}n\epsilon^{2}\}$.
By a standard integration argument, $\e|\mathcal{E}_{k\ell}|\leq C_{1}/\sqrt{n}+\int_{0}^{\infty}C_{2}\exp\{ -C_{3}n\epsilon^{2}\}d\epsilon\leq C_{5}/\sqrt{N}$,
for every $\ell=1,\ldots,L_{N}$, where $C_{5}:=\sqrt{K}\{ C_{1}+\int_{0}^{\infty}C_{2}\exp( -C_{3}\epsilon^{2})d\epsilon\}<\infty$.
From (\ref{eq:term 1 decomposition-1}), we obtain
$\e[\text{term 1}]\leq W(\pi^{*})-K^{-1}\sum_{k=1}^{K}\e[W(\widehat{\pi}_{\ell,I_{-k}})]+C_{5}/\sqrt{N}$
for all $\ell=1,\ldots,L_{N}$. Since the observed data are i.i.d., we can express the expected value as $\frac{1}{K}\sum_{k=1}^{K}\e\left[W(\widehat{\pi}_{\ell,I_{-k}})\right]=\e\left[W(\widehat{\pi}_{\ell,I_{-1}})\right]$.
Since the decomposition above holds for every $\ell=1,\ldots,L_N$, we obtain
\begin{equation}
\e\left[\text{term 1}\right]\leq\inf_{\ell=1,\ldots,L_{N}}\left\{ W(\pi^{*})-\e\left[W(\widehat{\pi}_{\ell,I_{-1}})\right]\right\} +\frac{C_{5}}{\sqrt{N}},\label{eq:result-for-term-1-1}
\end{equation}

\textbf{term 2.} Now we address term 2 in (\ref{eq:decomposition of the regret-1}).
According to the finite-candidate definition of $\widehat{\pi}$ in Remark~\ref{rem:finite L}, we have
\[
\text{term 2}=R_{CV}(\widehat{\ell})-W(\widehat{\pi})=\frac{1}{K}\sum_{k=1}^{K}\widehat{W}_{I_{k}}(\widehat{\pi}_{\widehat{\ell},I_{-k}})-W(\widehat{\pi})\leq\widehat{W}_{I_{\widehat{k}}}(\widehat{\pi}_{\widehat{\ell},I_{-\widehat{k}}})-W(\widehat{\pi}_{\widehat{\ell},I_{-\widehat{k}}}).
\]
Since $\widehat{\pi}_{\ell,I_{-k}}$ depends solely on $\{\bs Z_{i}\}_{i\in I_{-k}}$, it is independent of $\{\bs Z_{i}\}_{i\in I_{k}}$. For every $k=1,\ldots,K$ and $\ell=1,\ldots,L_{N}$, Assumption \ref{assu:general assu on hatW} gives
$P(|\mathcal{E}_{k\ell}|\geq\epsilon+r_N+C_{1}/\sqrt{n})\leq C_{2}\exp(-(C^{*})^{2}C_{3}\log L_{N}/K)\exp\{ -C_{3}n\epsilon^{2}\}$,
where $C^{*}>\sqrt{K/C_{3}}$ is a constant. The union bound gives
\begin{align*}
 & P\left(\mathcal{E}_{\widehat{k}\widehat{\ell}}\geq\epsilon+\frac{C_{1}}{\sqrt{n}}+r_{N}\right)\leq\sum_{k=1}^{K}\sum_{\ell=1}^{L_{N}}P\left(|\mathcal{E}_{k\ell}|\geq\epsilon+\frac{C_{1}}{\sqrt{n}}+r_{N}\right)\\
\leq & L_{N}K\cdot C_{2}\exp\left(-\left(C^{*}\right)^{2}C_{3}\frac{\log L_{N}}{K}\right)\exp\left\{ -C_{3}n\epsilon^{2}\right\} \leq KC_{2}\exp\left\{ -C_{3}n\epsilon^{2}\right\}
\end{align*}
for all $\epsilon\geq0$ and $n\geq C_{4}$. By applying a standard integration argument, we have
$\e[\mathcal{E}_{\widehat{k}\widehat{\ell}}-r_N]\leq C_{1}/\sqrt{n}+\int_{0}^{\infty}KC_{2}\exp\{ -C_{3}n\epsilon^{2}\}d\epsilon\leq C_{6}/\sqrt{N}$,
where $C_{6}:=\sqrt{K}\{ C_{1}+\int_{0}^{\infty}KC_{2}\exp( -C_{3}\epsilon^{2}) d\epsilon\}<\infty$.
Thus, we obtain that
\begin{equation}
\e\left[\text{term 2}\right]\leq\e\left[\mathcal{E}_{\widehat{k}\widehat{\ell}}\right]\leq C^{*}\sqrt{\frac{\log L_{N}}{N}}+\frac{C_{6}}{\sqrt{N}}.\label{eq:result for term 2-1}
\end{equation}

Combining the results from (\ref{eq:decomposition of the regret-1}),
(\ref{eq:result-for-term-1-1}), and (\ref{eq:result for term 2-1}) leads
to
\[
\begin{aligned}
\e\left[W(\pi^{*})-W(\widehat{\pi})\right]&\leq\inf_{\ell=1,\ldots,L_{N}}\left\{ W(\pi^{*})-\max_{\pi\in\Pi_{\ell}}W(\pi)+\max_{\pi\in\Pi_{\ell}}W(\pi)\right.\\
&\qquad\left.-\e\left[W(\widehat{\pi}_{\ell,I_{-1}})\right]\right\} +C\sqrt{\frac{\log L_{N}}{N}}
\end{aligned}
\]
for large $N$ and $L_N\ge2$, where $C>0$ is a constant.
\end{proof}

\section{Proof of Theorem \ref{thm:oracle holdout unknown propensity}}

\subsection[Asymptotic properties of nuisance estimators]{The asymptotic properties of $\widehat{e}_{I}(\boldsymbol{X})$ and
$\widehat{\mu}_{I,jt}(\boldsymbol{X};\widehat{\boldsymbol{\beta}}_{I}^{\mathrm{init}}(\pi))$}
\begin{lem}
\label{lem:converge of propensity} Suppose that Assumptions
\ref{assu:unconfounded}, \ref{assu:Overlap assumption and smoothness},
\ref{assu:regularity-conditions-on-L},
\ref{assu:regularity-assumptions-on-U},
\ref{assu:smoothness for mu}, \ref{assu:regularity on mu}, and
\ref{assu:regularity-assumptions-on-nuisance} hold. Let $I$ be a fixed index set such that $I\subset\left\{ 1,\ldots,N\right\} $ and $\left|I\right|=m$. If the DNN class $\mathcal{F}_{\mathrm{DNN}}(\mathcal{H}_{e},\mathcal{D}_{e})$
is constructed with $\mathcal{H}_{e}\mathcal{D}_{e}\asymp m^{d/(4s_{e}+2d)}\left(\log m\right)^{2}$,
then the DNN estimator defined in (\ref{eq:logistic reg for propensity})
satisfies
\(P\left(\left\Vert \widehat{e}_{I}(\bs X)-e^{*}(\bs X)\right\Vert _{P,2}>\rho_{e,m}+t\right)\leq c_{1}\exp\left(-c_{2}mt^{2}\right)\)
for all $t>0$, where $c_{1},c_{2}>0$ are finite constants independent
of $t$ and $I$, and $\rho_{e,m}$ is a non-negative sequence satisfying $\rho_{e,m}=o(m^{-1/4})$.
\end{lem}
\begin{proof}
Without loss of generality, we let $I=\{1,\ldots,m\}$. In what follows, we write $\mathcal{F}_{e,m}:=\mathcal{F}_{\mathrm{DNN}}(\mathcal{H}_{e},\mathcal{D}_{e})\cap\{f:\|f\|_{\infty}\le M\}$.
Denote the logistic function $\frac{1}{1+\exp(-x)}$ by $\iota(x)$. Let
\[
\begin{aligned}
L(t,u)&:=\frac{t}{\iota(u)^{2}}-\frac{2}{\iota(u)}+\frac{1-t}{\left(1-\iota(u)\right)^{2}}-\frac{2}{1-\iota(u)},\\
\mathcal{L}(f)&:=\e\left[L(T,f(\bs X))\right],\qquad \widehat{\mathcal{L}}_{m}(f):=\frac{1}{m}\sum_{i\in I}L(T_{i},f(\bs X_{i})).
\end{aligned}
\]
Let $\widehat{\lambda}_{I}(\bs X):=\log\{\widehat{e}_{I}(\bs X)/(1-\widehat{e}_{I}(\bs X))\}$. By construction and Assumption~\ref{assu:regularity-assumptions-on-nuisance}(i),
$\widehat{\lambda}_{I}\in\mathcal{F}_{e,m}$ almost surely. Since $\widehat{e}_{I}$ minimizes the empirical risk
in (\ref{eq:logistic reg for propensity}) over
$\mathrm{logistic}\circ\mathcal{F}_{\mathrm{DNN}}(\mathcal{H}_{e},\mathcal{D}_{e})$, and the logistic map is one-to-one, we may regard $\widehat{\lambda}_{I}$
as satisfying $\widehat{\lambda}_{I}=\arg\min_{f\in\mathcal{F}_{e,m}}\widehat{\mathcal{L}}_{m}(f)$. Let $\lambda^{*}:=\log\{e^{*}/(1-e^{*})\}$ and $\overline{\lambda}:=\arg\min_{f\in\mathcal{F}_{e,m}}\mathcal{L}(f)$.
Then, $\widehat{e}_{I}=\iota\circ\widehat{\lambda}_{I}$ and $e^{*}=\iota\circ\lambda^{*}$.
We will first establish the convergence results for $\widehat{\lambda}_{I}(\bs X)$,
and then extend these findings to derive the convergence properties
of $\widehat{e}_{I}(\bs X)$. According to Assumption~\ref{assu:unconfounded}(ii)
and Assumption~\ref{assu:Overlap assumption and smoothness}, $e^{*}$
is bounded away from zero and one and
$\left\Vert \lambda^{*}\right\Vert _{\infty}\leq M$. It is important to note that
\begin{align*}
\mathcal{L}(f)-\mathcal{L}(\lambda^{*})
&=\e\left[e^{*}(\bs X)\left\{ \frac{1}{\iota(f(\bs X))}-\frac{1}{e^{*}(\bs X)}\right\} ^{2}\right]\\
&\quad+\e\left[\left\{ 1-e^{*}(\bs X)\right\} \left\{ \frac{1}{1-\iota(f(\bs X))}-\frac{1}{1-e^{*}(\bs X)}\right\} ^{2}\right].
\end{align*}
For $f\in\mathcal{F}_{e,m}$, since both $f$ and $\lambda^{*}$ are bounded by $M$, both $\iota(f(\bs X))$ and $e^{*}(\bs X)=\iota(\lambda^{*}(\bs X))$ are bounded away from zero and one. Additionally,
\begin{align*}
e^{*}(\bs X)\left\{ \frac{1}{\iota(f(\bs X))}-\frac{1}{e^{*}(\bs X)}\right\} ^{2}
&=\frac{\left\{ \iota(f(\bs X))-e^{*}(\bs X)\right\} ^{2}}{\iota(f(\bs X))^{2}e^{*}(\bs X)},\\
\left\{ 1-e^{*}(\bs X)\right\} \left\{ \frac{1}{1-\iota(f(\bs X))}-\frac{1}{1-e^{*}(\bs X)}\right\} ^{2}
&=\frac{\left\{ \iota(f(\bs X))-e^{*}(\bs X)\right\} ^{2}}{\left(1-\iota(f(\bs X))\right)^{2}\left(1-e^{*}(\bs X)\right)}.
\end{align*}
By substituting these two identities into the expression for
$\mathcal{L}(f)-\mathcal{L}(\lambda^{*})$, we obtain:
\[
\mathcal{L}(f)-\mathcal{L}(\lambda^{*})
=\e\left[\left\{ \iota(f(\bs X))-e^{*}(\bs X)\right\} ^{2}\left\{ \frac{1}{\iota(f(\bs X))^{2}e^{*}(\bs X)}+\frac{1}{\left(1-\iota(f(\bs X))\right)^{2}\left(1-e^{*}(\bs X)\right)}\right\} \right].
\]
The multiplier of $\left\{ \iota(f(\bs X))-e^{*}(\bs X)\right\} ^{2}$
is bounded above and below by positive constants that depend solely on
$M$. Additionally, by applying the mean value theorem and considering that $\iota'$ is bounded above and below on the interval $[-M,M]$, the norms
$\left\Vert \iota\circ f-\iota\circ\lambda^{*}\right\Vert _{P,2}$ and
$\left\Vert f-\lambda^{*}\right\Vert _{P,2}$ are equivalent. Therefore,
\begin{equation}
\begin{aligned}
c^{*}\left\Vert f-\lambda^{*}\right\Vert _{P,2}^{2}&\leq c_{1}^{*}\left\Vert \iota\circ f-\iota\circ\lambda^{*}\right\Vert _{P,2}^{2}\leq\mathcal{L}(f)-\mathcal{L}(\lambda^{*})\\
&\le C_{1}^{*}\left\Vert \iota\circ f-\iota\circ\lambda^{*}\right\Vert _{P,2}^{2}\leq C^{*}\left\Vert f-\lambda^{*}\right\Vert _{P,2}^{2}.
\end{aligned}\label{eq:curvature condition}
\end{equation}
for all $f\in\mathcal{F}_{e,m}$, where
$c_{1}^{*},C_{1}^{*},c^{*},C^{*}>0$ are constants depending only on the assumptions.
For any $t\in\{0,1\}$ and $-M\leq v\leq u\leq M$, we have
\begin{align}
\left|L(t,u)-L(t,v)\right| & \leq C_{lip}^{\prime}\left|\iota(u)-\iota(v)\right|\leq C_{lip}\left|u-v\right|,\label{eq:lip of loss in propensity}
\end{align}
where $C_{lip},C_{lip}^{\prime}>0$ are constants.

We first document the approximation floor for the population projection.
Passing from $\mathcal{F}_{\mathrm{DNN}}(\mathcal{H}_{e},\mathcal{D}_{e})$ to $\mathcal{F}_{e,m}$ only imposes $\left\Vert f\right\Vert _{\infty}\leq M$, which does not change the approximation rate because $\left\Vert \lambda^{*}\right\Vert _{\infty}\leq M$. By Assumption~\ref{assu:Overlap assumption and smoothness} and
\citet[Corollary~1.2]{lu2021Deep}, there exists $\widetilde{\lambda}\in\mathcal{F}_{e,m}$ such that
\[
\left\Vert \widetilde{\lambda}-\lambda^{*}\right\Vert _{\infty}
\leq C\left(\mathcal{H}_{e}/\log\mathcal{H}_{e}\right)^{-2s_{e}/d}
\left(\mathcal{D}_{e}/\log\mathcal{D}_{e}\right)^{-2s_{e}/d}.
\]
Since $\overline{\lambda}$ minimizes $\mathcal{L}$ over $\mathcal{F}_{e,m}$, it can be concluded from (\ref{eq:curvature condition}) that:
\begin{equation}
0\leq\mathcal{L}(\overline{\lambda})-\mathcal{L}(\lambda^{*})
\leq\mathcal{L}(\widetilde{\lambda})-\mathcal{L}(\lambda^{*})
\leq C\left\Vert \widetilde{\lambda}-\lambda^{*}\right\Vert _{\infty}^{2}
\leq C m^{-2s_{e}/(2s_{e}+d)}(\log m)^{6},\label{eq:approximation floor propensity}
\end{equation}
where the last inequality utilizes the construction and growth conditions
on $\mathcal{H}_{e}$ and $\mathcal{D}_{e}$ in
Assumption~\ref{assu:regularity-assumptions-on-nuisance}(ii), and
incorporates the resulting logarithmic factors into $(\log m)^{6}$. Define $a_{e,m}:=C m^{-2s_{e}/(2s_{e}+d)}(\log m)^{6}$.
By (\ref{eq:approximation floor propensity}) and the definition of $a_{e,m}$, it follows that $\mathcal{L}(\overline{\lambda})-\mathcal{L}(\lambda^{*})\leq a_{e,m}$.
This error bound will be utilized in the localization argument
below. In particular, by (\ref{eq:curvature condition}), for every
$f\in\mathcal{F}_{e,m}$, and by substituting $f=\overline{\lambda}$ into the resulting inequality, we obtain
\begin{align}
\mathcal{L}(f)-\mathcal{L}(\overline{\lambda})
&=\left\{\mathcal{L}(f)-\mathcal{L}(\lambda^{*})\right\}-\left\{\mathcal{L}(\overline{\lambda})-\mathcal{L}(\lambda^{*})\right\}\geq c^{*}\left\Vert f-\lambda^{*}\right\Vert _{P,2}^{2}-a_{e,m},\label{eq:localized lower propensity}\\
\left\Vert \overline{\lambda}-\lambda^{*}\right\Vert _{P,2}^{2}
&\leq C a_{e,m}.\label{eq:bar lambda approximation propensity}
\end{align}
We also apply the crude bound $0\leq\mathcal{L}(f)-\mathcal{L}(\overline{\lambda})
\leq\mathcal{L}(f)-\mathcal{L}(\lambda^{*})
\leq C\| f-\lambda^{*}\| _{P,2}^{2}\leq4CM^{2}$, derived from the definition
of $\overline{\lambda}$ and (\ref{eq:curvature condition}).

Next, we bound the empirical process. Let $v_{e,m}:=\mathcal{H}_{e}^{2}\mathcal{D}_{e}^{2}\log(\mathcal{H}_{e}\mathcal{D}_{e})$.
By \citet[Theorem~7]{bartlett2019Nearlytight}, the class $\mathcal{F}_{e,m}$ is a VC-subgraph class with a VC index that is bounded by
$Cv_{e,m}$. Hence, \citet[Theorem~2.6.7]{vaart1996Weak} implies that,
for constants $a>1$ and $C>0$,
\[
\sup_{Q}\log N\left(M\epsilon,\mathcal{F}_{e,m},\left\Vert \cdot\right\Vert _{Q,2}\right)
\leq Cv_{e,m}\log(a/\epsilon),\qquad0<\epsilon<1,
\]
By (\ref{eq:lip of loss in propensity}) and Lemma~\ref{lem:transformed uniform entropy}, the loss-difference class
\(\{(T,\bs X)\mapsto L(T,f(\bs X))-L(T,h(\bs X)): f,h\in\mathcal{F}_{e,m}\}\)
is VC-type with index \(v_{e,m}\), up to constants.

For $s>0$, define the localized class
\(\mathcal{G}_{e}(s):=\{(T,\bs X)\mapsto L(T,f(\bs X))-L(T,\overline{\lambda}(\bs X)):f\in\mathcal{F}_{e,m},\ \mathcal{L}(f)-\mathcal{L}(\overline{\lambda})\leq s\}\).
Every function in $\mathcal{G}_{e}(s)$ is bounded by a constant that depends solely on $M$. Furthermore, if $s\geq2a_{e,m}$ and
$g_{f}\in\mathcal{G}_{e}(s)$, it then follows from (\ref{eq:lip of loss in propensity})
that $Pg_{f}^{2}\leq C\left\Vert f-\overline{\lambda}\right\Vert _{P,2}^{2}$.
Since $g_{f}\in\mathcal{G}_{e}(s)$, we have:
$\mathcal{L}(f)-\mathcal{L}(\overline{\lambda})\leq s$. Thus, the implication from
(\ref{eq:localized lower propensity}) leads to:
$\left\Vert f-\lambda^{*}\right\Vert _{P,2}^{2}\leq C(s+a_{e,m})$.
Combining these bounds with \eqref{eq:bar lambda approximation propensity} and using $s\geq2a_{e,m}$ gives
\[
Pg_{f}^{2}\leq C\left\Vert f-\overline{\lambda}\right\Vert _{P,2}^{2}\leq C\left\{\left\Vert f-\lambda^{*}\right\Vert _{P,2}^{2}+\left\Vert \overline{\lambda}-\lambda^{*}\right\Vert _{P,2}^{2}\right\}\leq Cs.
\]
For every $s\geq s_{0}:=\max\{2a_{e,m},m^{-2}\}$, the work by
\citet[Lemma~6.2]{chernozhukov2018Double} provides the following result:
\[
\e\left[\sup_{g\in\mathcal{G}_{e}(s)}
\left|\left(P_{m}-P\right)g\right|\right]
\leq C\left(\sqrt{\frac{sv_{e,m}\log m}{m}}+\frac{v_{e,m}\log m}{m}\right).
\]
Applying Bousquet's version of Talagrand's
inequality \citep[Theorem~2.3]{bousquet2002Bennett} results in the following for
all $u\geq0$:
\[
P\left(\sup_{g\in\mathcal{G}_{e}(s)}
\left|\left(P_{m}-P\right)g\right|>
C\left[\sqrt{\frac{s(v_{e,m}\log m+u)}{m}}+\frac{v_{e,m}\log m+u}{m}\right]\right)
\leq e^{-u}.
\]

Let $s_{k}:=2^{k}s_{0}$, $k=0,1,\ldots,K_{m}$, where
$K_{m}:=\max\{0,\lceil\log_{2}(4CM^{2}/s_{0})\rceil\}$. Then $K_{m}=O(\log m)$ and the dyadic intervals
$(s_{k-1},s_{k}]$, with $s_{-1}:=0$, cover all possible values of
$\mathcal{L}(f)-\mathcal{L}(\overline{\lambda})$. By applying the previous concentration bound to $\mathcal{G}_{e}(s_{k})$ with a tail parameter of $u+k\log2$ and taking a union bound, there exists an event
$\mathcal{E}_{e,m}(u)$ such that $P\{\mathcal{E}_{e,m}(u)\}\geq1-c_{1}\exp(-c_{2}u)$ and, on $\mathcal{E}_{e,m}(u)$, for every $k=0,\ldots,K_{m}$,
\[
\sup_{g\in\mathcal{G}_{e}(s_{k})}\left|\left(P_{m}-P\right)g\right|
\leq C\left[\sqrt{\frac{s_{k}(v_{e,m}\log m+u)}{m}}
+\frac{v_{e,m}\log m+u}{m}\right],
\]
where $k\leq K_{m}=O(\log m)$ and $v_{e,m}\geq1$ for large $m$.

Set $R_{m}:=\mathcal{L}(\widehat{\lambda}_{I})-\mathcal{L}(\overline{\lambda})$.
Since $\widehat{\lambda}_{I}$ minimizes the empirical risk, we have that
$P_{m}\{L(T,\widehat{\lambda}_{I}(\bs X))-L(T,\overline{\lambda}(\bs X))\}\leq0$. Therefore,
\[
R_{m}\leq(P-P_{m})\left\{L(T,\widehat{\lambda}_{I}(\bs X))-L(T,\overline{\lambda}(\bs X))\right\}
\leq\left|\left(P_{m}-P\right)\left\{L(T,\widehat{\lambda}_{I}(\bs X))-L(T,\overline{\lambda}(\bs X))\right\}\right|.
\]
On $\mathcal{E}_{e,m}(u)$, if $R_{m}\leq s_{0}$, then since $s_{0}=\max\{2a_{e,m},m^{-2}\}$ and
$m^{-2}\leq(v_{e,m}\log m+u)/m$ for large $m$, we have
$R_{m}\leq C\{a_{e,m}+(v_{e,m}\log m+u)/m\}$. If $R_{m}>s_{0}$,
choose $k$ such that $s_{k}<R_{m}\leq s_{k+1}$. Then,
$\widehat{\lambda}_{I}$ belongs to the localized class $\mathcal{G}_{e}(s_{k+1})$, and hence \(R_{m}\leq C[\sqrt{s_{k+1}(v_{e,m}\log m+u)/m}+(v_{e,m}\log m+u)/m]\).
Since $s_{k+1}=2s_{k}\leq2R_{m}$, the elementary inequality
$r\leq A\sqrt{r}+B\Rightarrow r\leq2A^{2}+2B$ leads to
\begin{equation}
R_{m}\leq C\left\{a_{e,m}+\frac{v_{e,m}\log m+u}{m}\right\}
\label{eq:excess risk propensity}
\end{equation}
on $\mathcal{E}_{e,m}(u)$. By combining (\ref{eq:excess risk propensity})
with (\ref{eq:bar lambda approximation propensity}) and
(\ref{eq:curvature condition}), we derive a result that holds w.p. at least $1-c_{1}\exp(-c_{2}u)$:
\[
\left\Vert \widehat{\lambda}_{I}-\lambda^{*}\right\Vert _{P,2}^{2}\leq C R_{m}+C\left\Vert \overline{\lambda}-\lambda^{*}\right\Vert _{P,2}^{2}\leq C a_{e,m}+C\frac{v_{e,m}\log m}{m}+C\frac{u}{m}.
\]

With $u=mt^{2}$ and using $\sqrt{x+y}\leq\sqrt{x}+\sqrt{y}$, we find that, w.p. at least
$1-c_{1}\exp(-c_{2}mt^{2})$, \(\left\Vert \widehat{\lambda}_{I}-\lambda^{*}\right\Vert _{P,2}\leq C(a_{e,m}+v_{e,m}\log m/m)^{1/2}+Ct\).
Since \(\left(a_{e,m}+v_{e,m}\log m/m\right)^{1/2}
\leq C m^{-s_{e}/(2s_{e}+d)}(\log m)^{3}=o(m^{-1/4})\),
where the last equality follows from $s_{e}>d/2$, the desired deterministic
rate follows for the logit estimator. Finally, the mean value theorem
and the boundedness of the logits imply
\(\left|\widehat{e}_{I}(\bs X)-e^{*}(\bs X)\right|
\leq C\left|\widehat{\lambda}_{I}(\bs X)-\lambda^{*}(\bs X)\right|\).
For every $t>0$, the following expression holds w.p. at least $1-c_{1}\exp(-c_{2}mt^{2})$:
\begin{align}
\left\Vert \widehat{e}_{I}(\bs X)-e^{*}(\bs X)\right\Vert _{P,2}
& \leq C m^{-s_{e}/(2s_{e}+d)}(\log m)^{3}+Ct \leq o(m^{-1/4})+Ct.\label{eq:convergence rate of hat e}
\end{align}
After enlarging $\rho_{e,m}$, and renaming constants if necessary, we obtain the result for all $t>0$. This completes the proof.
\end{proof}
\begin{lem}
\label{lem:convergence of initial beta} Suppose Assumptions
\ref{assu:unconfounded}, \ref{assu:Overlap assumption and smoothness},
\ref{assu:regularity-conditions-on-L},
\ref{assu:regularity-assumptions-on-U},
\ref{assu:smoothness for mu}, \ref{assu:regularity on mu}, and
\ref{assu:regularity-assumptions-on-nuisance} hold. Let $I\subset\left\{ 1,\ldots,N\right\} $
be an index set with $\left|I\right|=m$. For any policy class $\Pi\subset\Pi_{\infty}$
with a VC dimension $\mathrm{VC}(\Pi)$, if $\mathrm{VC}(\Pi)/m\to0$
 as $m\to\infty$, then the quantity $\widehat{\bs{\beta}}_{I}^{\mathrm{init}}(\pi)$
defined by (\ref{eq:beta-init-pi-I}) satisfies
  \[
P\left(\sup_{\pi\in\Pi}\left\Vert \widehat{\bs{\beta}}_{I}^{\mathrm{init}}(\pi)-\bs{\beta}^{*}(\pi)\right\Vert \geq C\sqrt{\frac{\mathrm{VC}(\Pi)}{m}}+C\rho_{e,m}+\delta\right)\leq c_{1}\exp\left(-c_{2}m\delta^{2}\right)
\]
 for all $0<\delta<c_{3}$ and large $m$, where $c_{1},c_{2},c_{3}>0$
are finite constants independent of $\delta$, $m$, $I$
and $\Pi$, and $C>0$ is finite.
\end{lem}
\begin{proof}
Without loss of generality, we assume $I=\{1,\ldots,m\}$. For notational
simplicity, we present the argument for one coordinate. Applying the same
argument to $j=1,\ldots,p$ and taking a finite union bound yields the
displayed vector-norm result, since $p$ is fixed; the constants below
absorb $p$. In this proof, we write $\mathcal{F}_{\mathrm{DNN},e}$
for $\text{logistic}\circ\left\{ \mathcal{F}_{\mathrm{DNN}}(\mathcal{H}_{e},\mathcal{D}_{e})\bigcap\{f:\left\Vert f\right\Vert _{\infty}\le M\}\right\} $.
Let $v_{e,m}:=\mathcal{H}_{e}^{2}\mathcal{D}_{e}^{2}\log(\mathcal{H}_{e}\mathcal{D}_{e})$. By Assumption~\ref{assu:regularity-assumptions-on-nuisance}(ii), we have
$\sqrt{v_{e,m}/m}=o(m^{-1/4})$. By \citet[Theorem 7]{bartlett2019Nearlytight}
and \citet[Theorem 2.6.7]{vaart1996Weak}, the truncated DNN class inside $\mathcal{F}_{\mathrm{DNN},e}$ is VC-type with index \(v_{e,m}\). Since the logistic map is bounded Lipschitz on $[-M,M]$, Lemma~\ref{lem:transformed uniform entropy} yields
\begin{equation}
\sup_{Q}\log N\left(\epsilon,\mathcal{F}_{\mathrm{DNN},e},\left\Vert \cdot\right\Vert _{Q,2}\right)\leq Cv_{e,m}\log\left(a/\epsilon\right)\text{ for all }0<\epsilon<1,\label{eq:entropy for DNNe}
\end{equation}
  where $a,C>0$ do not depend on $m$. Enlarging $\rho_{e,m}$ if
necessary, we may assume $\sqrt{v_{e,m}/m}\leq\rho_{e,m}$ for large $m$.
 Let
\begin{align*}
Q(\beta;\pi,e)&=\e\left[\left\{ \frac{\pi(\bs X)T}{e(\bs X)}+\frac{\left(1-\pi(\bs X)\right)\left(1-T\right)}{1-e(\bs X)}\right\} \mathcal{L}_{1}(Y-\beta)\right],\\
Q_{m}(\beta;\pi,e)&=\frac{1}{m}\sum_{i\in I}\Bigl\{\frac{\pi(\bs X_{i})T_{i}}{e(\bs X_{i})}+\frac{(1-\pi(\bs X_{i}))(1-T_{i})}{1-e(\bs X_{i})}\Bigr\}\mathcal{L}_{1}(Y_{i}-\beta).
\end{align*}
Then $\widehat{\beta}_{I}^{\mathrm{init}}(\pi)=\underset{\beta\in\R}{\arg\min}\ Q_{m}(\beta;\pi,\widehat{e})$
and $\beta^{*}(\pi)=\underset{\beta\in\R}{\arg\min}\ Q(\beta;\pi,e^{*})$, where $\widehat{e}:=\widehat{e}_{I}$.
Since $Q_{m}(\beta;\pi,\widehat{e})$ is convex in $\beta$, we apply
Lemma \ref{lem:Nearness of the argmins of the convex functions} to
bound $\widehat{\beta}_{I}^{\mathrm{init}}(\pi)-\beta^{*}(\pi)$.
We decompose the proof into three steps and write $a_{m}(\Pi):=\sqrt{\mathrm{VC}(\Pi)/m}+\rho_{e,m}$.

\textbf{Step 1. (Developing a lower bound for $h(\delta;\pi)$).}
For all $\delta\in(0,\min\{(3\underline{Q}^{\prime\prime})/(4Q_{lip}^{\prime\prime}),c_{0}\})$, let $h(\delta;\pi):=\inf_{|\beta-\beta^{*}(\pi)|=\delta}Q(\beta;\pi,e^{*})-Q(\beta^{*}(\pi);\pi,e^{*})$.
By the first-order condition for $\beta^{*}(\pi)$, the mean value theorem, and Assumption~\ref{assu:regularity-conditions-on-L}(ii), we have
\begin{equation}
\inf_{\pi\in\Pi_{\infty}}h(\delta;\pi)\geq\frac{\underline{Q}^{\prime\prime}}{4}\delta^{2}\ \text{ for all }0<\delta\leq\min\left\{\frac{3\underline{Q}^{\prime\prime}}{4Q_{lip}^{\prime\prime}},c_{0}\right\}.\label{eq:init-beta-lower-bound}
\end{equation}

\textbf{Step 2. (Fluctuation bound).}
For any $\delta$ in the interval specified in Step 1,
we let
\[
\Delta(\delta;\pi):=\sup_{\left|\beta-\beta^{*}(\pi)\right|=\delta}\left|Q_{m}(\beta;\pi,\widehat{e})-Q_{m}(\beta^{*}(\pi);\pi,\widehat{e})-\left\{ Q(\beta;\pi,e^{*})-Q(\beta^{*}(\pi);\pi,e^{*})\right\} \right|.
\]
Let
\begin{align*}
Q_{m}^{\prime}(\beta;\pi,e)&:=\frac{1}{m}\sum_{i\in I}\Bigl\{\frac{\pi(\bs X_{i})T_{i}}{e(\bs X_{i})}+\frac{(1-\pi(\bs X_{i}))(1-T_{i})}{1-e(\bs X_{i})}\Bigr\}\mathcal{L}_{1}^{\prime}(Y_{i}-\beta),\\
Q^{\prime}(\beta;\pi,e)&:=\e\left[\left\{ \frac{\pi(\bs X)T}{e(\bs X)}+\frac{\left(1-\pi(\bs X)\right)\left(1-T\right)}{1-e(\bs X)}\right\} \mathcal{L}_{1}^{\prime}(Y-\beta)\right].
\end{align*}
By Assumption \ref{assu:regularity-conditions-on-L},
we have $-\int_{\beta^{*}(\pi)}^{\beta}Q_{m}^{\prime}(\widetilde{\beta};\pi,\widehat{e})d\widetilde{\beta}=Q_{m}(\beta;\pi,\widehat{e})-Q_{m}(\beta^{*}(\pi);\pi,\widehat{e})$, and analogously for $Q^{\prime}$,
and thus
\begin{align*}
\Delta(\delta;\pi) & =\sup_{\left|\beta-\beta^{*}(\pi)\right|=\delta}\left|-\int_{\beta^{*}(\pi)}^{\beta}\left\{ Q_{m}^{\prime}(\widetilde{\beta};\pi,\widehat{e})-Q^{\prime}(\widetilde{\beta};\pi,e^{*})\right\} d\widetilde{\beta}\right|\nonumber \\
 & \leq\delta\sup_{\left|\beta-\beta^{*}(\pi)\right|\leq\delta}\left|Q_{m}^{\prime}(\beta;\pi,\widehat{e})-Q^{\prime}(\beta;\pi,e^{*})\right|=\delta\widetilde{\Delta}(\delta;\pi),
\end{align*}
where we have let $\widetilde{\Delta}(\delta;\pi):=\sup_{\left|\beta-\beta^{*}(\pi)\right|\leq\delta}\left|Q_{m}^{\prime}(\beta;\pi,\widehat{e})-Q^{\prime}(\beta;\pi,e^{*})\right|$.
Now, it suffices to bound $\widetilde{\Delta}(\delta;\pi)$. Note
that $\widetilde{\Delta}(\delta;\pi)$ can be decomposed into
\begin{align}
\sup_{\pi\in\Pi}\widetilde{\Delta}(\delta;\pi)
&=\sup_{\pi\in\Pi}\sup_{\left|\beta-\beta^{*}(\pi)\right|\leq\delta}\left|Q_{m}^{\prime}(\beta;\pi,\widehat{e})-Q^{\prime}(\beta;\pi,e^{*})\right|\nonumber\\
&\leq\sup_{\pi\in\Pi}\sup_{\left|\beta-\beta^{*}(\pi)\right|\leq\delta}\left|Q_{m}^{\prime}(\beta;\pi,\widehat{e})-Q^{\prime}(\beta;\pi,\widehat{e})\right|\nonumber\\
&\quad+\sup_{\pi\in\Pi}\sup_{\left|\beta-\beta^{*}(\pi)\right|\leq\delta}\left|Q^{\prime}(\beta;\pi,\widehat{e})-Q^{\prime}(\beta;\pi,e^{*})\right|.\label{eq:decomp-of-tilde-delta}
\end{align}
We bound these two terms one by one.

\textbf{Step 2.1: Empirical-process term.}
By Lemma~\ref{lem:converge of propensity}, together with the construction
of $\widehat{e}$ and Assumption~\ref{assu:regularity-assumptions-on-nuisance}(i),
we have
\[
\sup_{\pi\in\Pi}\sup_{\left|\beta-\beta^{*}(\pi)\right|\leq\delta}\left|Q_{m}^{\prime}(\beta;\pi,\widehat{e})-Q^{\prime}(\beta;\pi,\widehat{e})\right|\leq\sup_{f\in\mathcal{M}}\left|\left\{ P_{m}-P\right\} f(Z)\right|
\]
w.p. at least $1-c_{1}\exp\left(-c_{2}mt^{2}\right)$
for all $0<t<c_{3}$ and large $m$, where
\begin{align*}
\mathcal{M} & :=\Bigl\{ Z\mapsto\left\{ \frac{\pi(\bs X)T}{e(\bs X)}+\frac{(1-\pi(\bs X))(1-T)}{1-e(\bs X)}\right\} \mathcal{L}_{1}^{\prime}(Y-\beta):\left|\beta-\beta^{*}(\pi)\right|\leq\delta,\\
 & \qquad\qquad e\in\mathcal{F}_{\mathrm{DNN},e},\left\Vert e-e^{*}\right\Vert _{P,2}\leq\rho_{e,m}+t,\pi\in\Pi\Bigr\}
\end{align*}
is a function class with envelope $C>0$. Since $\mathcal{L}_{1}^{\prime}$ is non-decreasing, \citet[Lemma 2.6.16]{vaart1996Weak} implies that $\{\bs Z\mapsto\mathcal{L}_{1}^{\prime}(Y-\beta):|\beta-\beta^{*}(\pi)|\leq\delta,\pi\in\Pi\}$ is VC-type with a fixed index. The map $(\pi,e,\ell)\mapsto\{\pi T/e+(1-\pi)(1-T)/(1-e)\}\ell$ is bounded Lipschitz on the overlap range, so Lemma~\ref{lem:transformed uniform entropy}, the VC-subgraph bound for $\Pi$, and (\ref{eq:entropy for DNNe}) give
\[
\sup_{Q}\log N\left(C\epsilon,\mathcal{M},\left\Vert \cdot\right\Vert _{Q,2}\right)\leq C\left\{ \mathrm{VC}(\Pi)+v_{e,m}\right\}\log\left(a/\epsilon\right)
\]
for all $0<\epsilon<1$, where $a>1$ is a constant. Applying
\citet[Lemma 6.2]{chernozhukov2018Double} gives that
  \[
\e\left[\sup_{f\in\mathcal{M}}\left|\left\{ P_{m}-P\right\} f(Z)\right|\right]\leq C\sqrt{\frac{\mathrm{VC}(\Pi)}{m}}+C\sqrt{\frac{v_{e,m}}{m}}.
\]
By the bounded difference inequality, we have
\[
P\left(\sup_{f\in\mathcal{M}}\left|\left\{ P_{m}-P\right\} f(Z)\right|-\e\left[\sup_{f\in\mathcal{M}}\left|\left\{ P_{m}-P\right\} f(Z)\right|\right]\geq t\right)\leq c_{1}\exp\left(-c_{2}mt^{2}\right)
\]
for any $t\geq0$. Thus, we obtain that
  \[
\sup_{\pi\in\Pi}\sup_{\left|\beta-\beta^{*}(\pi)\right|\leq\delta}\left|Q_{m}^{\prime}(\beta;\pi,\widehat{e})-Q^{\prime}(\beta;\pi,\widehat{e})\right|\leq C\sqrt{\frac{\mathrm{VC}(\Pi)}{m}}+C\sqrt{\frac{v_{e,m}}{m}}+t
\]
w.p. at least $1-2c_{1}\exp\left(-c_{2}mt^{2}\right)$
for all $0<t<c_{3}$ and large $m$.

\textbf{Step 2.2: Bound $\sup_{\pi\in\Pi}\sup_{\left|\beta-\beta^{*}(\pi)\right|\leq\delta}\left|Q^{\prime}(\beta;\pi,\widehat{e})-Q^{\prime}(\beta;\pi,e^{*})\right|$.}
By Lemma~\ref{lem:converge of propensity}, together with the construction
of $\widehat{e}$ and Assumption~\ref{assu:regularity-assumptions-on-nuisance}(i),
we have
\begin{align*}
 & \sup_{\pi\in\Pi}\sup_{\left|\beta-\beta^{*}(\pi)\right|\leq\delta}\left|Q^{\prime}(\beta;\pi,\widehat{e})-Q^{\prime}(\beta;\pi,e^{*})\right|\\
\leq & \sup_{\pi\in\Pi}\sup_{e\in\mathcal{F}_{\mathrm{DNN},e}:\left\Vert e-e^{*}\right\Vert _{P,2}\leq\rho_{e,m}+t}\sup_{\left|\beta-\beta^{*}(\pi)\right|\leq\delta}\left|Q^{\prime}(\beta;\pi,e)-Q^{\prime}(\beta;\pi,e^{*})\right|
\end{align*}
w.p. at least $1-c_{1}\exp\left(-c_{2}mt^{2}\right)$
for all $0<t<c_{3}$ and large $m$. By Assumption~\ref{assu:regularity-conditions-on-L}
and the boundedness of the functions in $\mathcal{F}_{\mathrm{DNN},e}$,
we have
\begin{align*}
&\sup_{\pi\in\Pi}\sup_{e\in\mathcal{F}_{\mathrm{DNN},e}:\| e-e^{*}\| _{P,2}\leq\rho_{e,m}+t}
\sup_{\left|\beta-\beta^{*}(\pi)\right|\leq\delta}
\left|Q^{\prime}(\beta;\pi,e)-Q^{\prime}(\beta;\pi,e^{*})\right|\\
&\qquad\lesssim\sup_{\left\Vert e-e^{*}\right\Vert _{P,2}\leq\rho_{e,m}+t}\left\Vert e-e^{*}\right\Vert _{P,2}\leq\rho_{e,m}+t.
\end{align*}
Therefore,
  \[
\sup_{\pi\in\Pi}\sup_{\left|\beta-\beta^{*}(\pi)\right|\leq\delta}\left|Q^{\prime}(\beta;\pi,\widehat{e})-Q^{\prime}(\beta;\pi,e^{*})\right|\leq C\rho_{e,m}+Ct
\]
w.p. at least $1-c_{1}\exp\left(-c_{2}mt^{2}\right)$
for all $0<t<c_{3}$ and large $m$.

\textbf{Step 3: Aggregating the results.} Recalling (\ref{eq:decomp-of-tilde-delta}),
we have
  \begin{equation}
\sup_{\pi\in\Pi}\widetilde{\Delta}(\delta;\pi)\leq C\sqrt{\frac{\mathrm{VC}(\Pi)}{m}}+C\rho_{e,m}+Ct\label{eq:bound for tilde Delta}
\end{equation}
w.p. at least $1-3c_{1}\exp\left(-c_{2}mt^{2}\right)$
for all $0<t<c_{3}$ and large $m$. Since $\sqrt{v_{e,m}/m}\leq\rho_{e,m}$ for large $m$, the bound in (\ref{eq:bound for tilde Delta}) holds as stated. For any $\delta<\min\{(3\underline{Q}^{\prime\prime})/(8Q_{lip}^{\prime\prime}),c_{0}/2\}$
and large $m$, let \(\delta^{\prime}=\delta+\frac{8C}{\underline{Q}^{\prime\prime}}\sqrt{\frac{\mathrm{VC}(\Pi)}{m}}+\frac{8C}{\underline{Q}^{\prime\prime}}\rho_{e,m}<\min\{(3\underline{Q}^{\prime\prime})/(4Q_{lip}^{\prime\prime}),c_{0}\}\).
 Applying Lemma~\ref{lem:Nearness of the argmins of the convex functions},
we have
\[
\begin{aligned}
& P\left(\sup_{\pi\in\Pi}\left|\widehat{\beta}_{I}^{\mathrm{init}}(\pi)-\beta^{*}(\pi)\right|>\delta^{\prime}\right)\leq P\left(\Delta(\delta^{\prime};\pi)\geq h(\delta^{\prime};\pi),\ \exists\pi\in\Pi\right)\leq P\left(\sup_{\pi\in\Pi}\widetilde{\Delta}(\delta^{\prime};\pi)\geq\frac{\underline{Q}^{\prime\prime}}{4}\delta^{\prime}\right)\\
\leq & P\left(\sup_{\pi\in\Pi}\widetilde{\Delta}(\delta^{\prime};\pi)\geq C\sqrt{\frac{\mathrm{VC}(\Pi)}{m}}+C\rho_{e,m}+\frac{\underline{Q}^{\prime\prime}}{8}\delta\right)\leq3c_{1}\exp\left(-c_{2}m\frac{\left(\underline{Q}^{\prime\prime}\right)^{2}}{64}\delta^{2}\right)
\end{aligned}
\]
where the second inequality follows from (\ref{eq:init-beta-lower-bound})
and the last one follows from (\ref{eq:bound for tilde Delta}). Since
$\delta^{\prime}\leq C\delta+C\sqrt{\mathrm{VC}(\Pi)/m}+C\rho_{e,m}$,
the stated result follows for $0<\delta<c_{3}$ and large $m$. This completes the proof.
\end{proof}
\begin{lem}
\label{lem:convergence of reg at any beta} Suppose that Assumptions
\ref{assu:unconfounded}, \ref{assu:Overlap assumption and smoothness},
\ref{assu:regularity-conditions-on-L},
\ref{assu:regularity-assumptions-on-U},
\ref{assu:smoothness for mu}, \ref{assu:regularity on mu}, and
\ref{assu:regularity-assumptions-on-nuisance} hold and $I\subset\left\{ 1,\ldots,N\right\} $
is an index set with $\left|I\right|=m$. Let $\Pi\subset\Pi_{\infty}$
be any policy class. If the DNN class $\mathcal{F}_{\mathrm{DNN}}(\mathcal{H}_{\mu},\mathcal{D}_{\mu})$
is constructed with $\mathcal{H}_{\mu}\mathcal{D}_{\mu}\asymp m^{d/(4s_{\mu}+2d)}\left(\log m\right)^{2}$,
then the DNN estimators defined in (\ref{eq:DNN reg for L}) and
(\ref{eq:DNN reg for U}) satisfy
  \[
P\left(\sup_{\pi\in\Pi}\sup_{\left\Vert \bs{\beta}-\bs{\beta}^{*}(\pi)\right\Vert \leq c_{0}}\left\Vert \widehat{\mu}_{I,jt}(\bs X;\bs{\beta})-\mu_{jt}^{*}(\bs X;\bs{\beta})\right\Vert _{P,2}>\rho_{\mu,m}+\delta\right)\leq c_{1}\exp\left(-c_{2}m\delta^{2}\right)
\]
for $t=0,1$, $j=0,\ldots,p$, $\delta\geq0$, and large $m$,
where $c_{1},c_{2}>0$ are finite constants independent
of $I$, $m$, $\Pi$, $j$, $t$, and $\delta$, and $\rho_{\mu,m}$ is
a non-negative deterministic sequence satisfying $\rho_{\mu,m}=o(m^{-1/4})$.
In particular, this lemma does not require any restriction on $\mathrm{VC}(\Pi)$.
\end{lem}
\begin{proof}
Without loss of generality, we assume $I=\{1,\ldots,m\}$. We abbreviate
$\widehat{\mu}_{I,jt}(\cdot;\bs{\beta})$ by $\widehat{\mu}(\cdot;\bs{\beta})$
and $\mu_{jt}^{*}(\cdot;\bs{\beta})$ by $\mu^{*}(\cdot;\bs{\beta})$
when no confusion arises. Fix $j\in\{0,1,\ldots,p\}$ and $t\in\{0,1\}$.
The proof below is written for this fixed pair; the constants are
uniform over the finitely many choices of $(j,t)$. Let
\(R_{j}(Y,\bs X;\bs{\beta}):=U(Y,\bs X,\bs{\beta})\) for \(j=0\), and
\(R_{j}(Y,\bs X;\bs{\beta}):=\mathcal{L}_{j}^{\prime}(Y-\beta_{j})\)
for \(j=1,\ldots,p\).
Define $\mathcal{B}_{0}:=\{\bs{\beta}\in\mathbb{R}^{p}:\|\bs{\beta}-\bs{\beta}^{*}(\pi)\|\leq c_{0}\text{ for some }\pi\in\Pi_{\infty}\}$. This set is deterministic, so the argument below does not involve
the complexity of $\Pi$.
We write $\mathcal{F}_{\mu,m}:=\mathcal{F}_{\mathrm{DNN}}(\mathcal{H}_{\mu},\mathcal{D}_{\mu})\cap\{f:\|f\|_{\infty}\leq M\}$. This bounded class will be used throughout the proof; the truncation to $\left\Vert f\right\Vert _{\infty}\leq M$ does not change the approximation rate used below because Assumption~\ref{assu:smoothness for mu} implies $\left\Vert \mu_{\ell r}^{*}(\cdot;\bs{\beta})\right\Vert _{\infty}\leq M$ uniformly over $\ell=0,\ldots,p$, $r\in\{0,1\}$ and $\bs{\beta}\in\mathcal{B}_{0}$.
By construction and
Assumption~\ref{assu:regularity-assumptions-on-nuisance}(i),
for every $\bs{\beta}\in\mathcal{B}_{0}$ the estimator defined in
(\ref{eq:DNN reg for L})--(\ref{eq:DNN reg for U}) belongs to
$\mathcal{F}_{\mu,m}$ almost surely. Since this estimator minimizes
the empirical risk over the larger class
$\mathcal{F}_{\mathrm{DNN}}(\mathcal{H}_{\mu},\mathcal{D}_{\mu})$,
we may also regard $\widehat{\mu}(\cdot;\bs{\beta})$ as an empirical
risk minimizer over $\mathcal{F}_{\mu,m}$ throughout the proof. Also
define the approximation and complexity
quantities
\[
\epsilon_{\mu,m}:=\sup_{0\leq\ell\leq p}\sup_{r\in\{0,1\}}\sup_{\bs{\beta}\in\mathcal{B}_{0}}\inf_{f\in\mathcal{F}_{\mu,m}}\left\Vert f-\mu_{\ell r}^{*}(\cdot;\bs{\beta})\right\Vert _{\infty}
\]
and $v_{\mu,m}:=\mathcal{H}_{\mu}^{2}\mathcal{D}_{\mu}^{2}\log(\mathcal{H}_{\mu}\mathcal{D}_{\mu})$.

First, the response is uniformly bounded. For $j=0$ this follows
from Assumption~\ref{assu:regularity-assumptions-on-U}(i),
and for $j\geq1$ from Assumption~\ref{assu:regularity-conditions-on-L}(iii):
\begin{equation}
\sup_{\bs{\beta}\in\mathcal{B}_{0}}\left|R_{j}(Y,\bs X;\bs{\beta})\right|\leq M/4
\qquad\text{almost surely.}\label{eq:uniform response bound reg beta}
\end{equation}
Assumption~\ref{assu:smoothness for mu} gives the corresponding
uniform smoothness of $\mu_{jt}^{*}(\cdot;\bs{\beta})$ over
$\bs{\beta}\in\mathcal{B}_{0}$. Also, Assumption~\ref{assu:Overlap assumption and smoothness}
implies that there exists $\underline{e}\in(0,1/2)$ such that $\underline{e}\leq e^{*}(\bs X)\leq1-\underline{e}$ almost surely.
Finally, by \citet[Theorem 7]{bartlett2019Nearlytight} and
\citet[Theorem 2.6.7]{vaart1996Weak}, there exist constants
$a>1$ and $C>0$ such that
\begin{equation}
\sup_{Q}\log N\left(M\epsilon,\mathcal{F}_{\mu,m},\left\Vert \cdot\right\Vert _{Q,2}\right)\leq Cv_{\mu,m}\log(a/\epsilon),\qquad0<\epsilon<1,\label{eq:entropy dnn reg beta}
\end{equation}
The response class $\{(Y,\bs X)\mapsto R_{j}(Y,\bs X;\bs{\beta}):\bs{\beta}\in\mathcal{B}_{0}\}$ is VC-type with a fixed index by Assumption~\ref{assu:regularity-assumptions-on-U}(ii) for $j=0$ and by \citet[Lemma 2.6.16]{vaart1996Weak} and monotonicity of $\mathcal{L}_{j}^{\prime}$ for $j\geq1$.

For $\bs z=(\tau,y,\bs x^{\trans})^{\trans}$ and $v\in\mathbb{R}$, define
\[
\begin{aligned}
L(\bs z,v;\bs{\beta})&:=1(\tau=t)\{R_{j}(y,\bs x;\bs{\beta})-v\}^{2},\\
\mathcal{L}(f;\bs{\beta})&:=\e[L(\bs Z,f(\bs X);\bs{\beta})],\qquad \widehat{\mathcal{L}}_{m}(f;\bs{\beta}):=m^{-1}\sum_{i=1}^{m}L(\bs Z_{i},f(\bs X_{i});\bs{\beta}).
\end{aligned}
\]
Then $\widehat{\mu}(\cdot;\bs{\beta})=\arg\min_{f\in\mathcal{F}_{\mu,m}}\widehat{\mathcal{L}}_{m}(f;\bs{\beta})$ and $\mu^{*}(\cdot;\bs{\beta})=\arg\min_{f}\mathcal{L}(f;\bs{\beta})$. Let \(\overline{\mu}(\cdot;\bs{\beta}):=\arg\min_{f\in\mathcal{F}_{\mu,m}}\mathcal{L}(f;\bs{\beta})\), \(\mathscr{E}(f;\bs{\beta}):=\mathcal{L}(f;\bs{\beta})-\mathcal{L}(\overline{\mu}(\cdot;\bs{\beta});\bs{\beta})\), and \(a_{\mu,m}:=\sup_{\bs{\beta}\in\mathcal{B}_{0}}\{\mathcal{L}(\overline{\mu}(\cdot;\bs{\beta});\bs{\beta})-\mathcal{L}(\mu^{*}(\cdot;\bs{\beta});\bs{\beta})\}\).

\textbf{Step 1.} We first record the curvature identity and the
approximation floor. Since $\mu^{*}(\cdot;\bs{\beta})=\e[R_{j}(Y,\bs X;\bs{\beta})\mid\bs X,T=t]$,
for every measurable $f$, by overlap,
\begin{align}
\mathcal{L}(f;\bs{\beta})-\mathcal{L}(\mu^{*}(\cdot;\bs{\beta});\bs{\beta})
&=\e\left[1(T=t)\{f(\bs X)-\mu^{*}(\bs X;\bs{\beta})\}^{2}\right],\label{eq:basic identity reg beta}\\
\underline{e}\left\Vert f(\bs X)-\mu^{*}(\bs X;\bs{\beta})\right\Vert _{P,2}^{2}
&\leq\mathcal{L}(f;\bs{\beta})-\mathcal{L}(\mu^{*}(\cdot;\bs{\beta});\bs{\beta})
\leq\left\Vert f(\bs X)-\mu^{*}(\bs X;\bs{\beta})\right\Vert _{P,2}^{2}.\label{eq:curvature reg beta}
\end{align}
By the definition of $\epsilon_{\mu,m}$, for each $\bs{\beta}\in\mathcal{B}_{0}$
there exists $f_{\bs{\beta}}\in\mathcal{F}_{\mu,m}$ with
$\left\Vert f_{\bs{\beta}}-\mu^{*}(\cdot;\bs{\beta})\right\Vert _{\infty}\leq\epsilon_{\mu,m}$.
Since $\overline{\mu}(\cdot;\bs{\beta})$ minimizes $\mathcal{L}(\cdot;\bs{\beta})$
over $\mathcal{F}_{\mu,m}$, (\ref{eq:basic identity reg beta}) yields
\begin{equation}
a_{\mu,m}\leq\epsilon_{\mu,m}^{2}.\label{eq:approx floor reg beta}
\end{equation}

\textbf{Step 2.} For $s>0$, define the localized class
\(\mathscr{F}(s;\bs{\beta}):=\{f\in\mathcal{F}_{\mu,m}:\mathscr{E}(f;\bs{\beta})\leq s\}\) and
\[
\mathcal{G}(s):=\left\{\bs z\mapsto L(\bs z,f(\bs x);\bs{\beta})-L(\bs z,\overline{\mu}(\bs x;\bs{\beta});\bs{\beta}):\bs{\beta}\in\mathcal{B}_{0},\ f\in\mathscr{F}(s;\bs{\beta})\right\}.
\]
Let $s_{0}:=\max\{2a_{\mu,m},m^{-2}\}$. If $s\geq s_{0}$ and
$f\in\mathscr{F}(s;\bs{\beta})$, then
\(\mathcal{L}(f;\bs{\beta})-\mathcal{L}(\mu^{*}(\cdot;\bs{\beta});\bs{\beta})
\leq s+a_{\mu,m}\leq\frac{3}{2}s\). Combining this bound with
(\ref{eq:curvature reg beta}) and applying
the same argument to $\overline{\mu}(\cdot;\bs{\beta})$ gives
$\left\Vert f(\bs X)-\overline{\mu}(\bs X;\bs{\beta})\right\Vert _{P,2}\leq C\sqrt{s}$.
For any $g\in\mathcal{G}(s)$, write $r=R_{j}(y,\bs x;\bs{\beta})$,
$u=f(\bs x)$ and $v=\overline{\mu}(\bs x;\bs{\beta})$. By
(\ref{eq:uniform response bound reg beta}) and the definition of
$\mathcal{F}_{\mu,m}$, $\left|(r-u)^{2}-(r-v)^{2}\right|\leq3M|u-v|$.
Thus $\mathcal{G}(s)$ has a bounded envelope depending only on $M$, and
\begin{equation}
\sup_{g\in\mathcal{G}(s)}\left\Vert g\right\Vert _{P,2}^{2}\leq Cs,\qquad s\geq s_{0}.\label{eq:variance reg beta}
\end{equation}

\textbf{Step 3.} We next bound the entropy of $\mathcal{G}(s)$. Since
$\overline{\mu}(\cdot;\bs{\beta})\in\mathcal{F}_{\mu,m}$, the class
$\mathcal{G}(s)$ is contained in
\[
\left\{\bs z\mapsto L(\bs z,f(\bs x);\bs{\beta})-L(\bs z,g(\bs x);\bs{\beta}):f,g\in\mathcal{F}_{\mu,m},\ \bs{\beta}\in\mathcal{B}_{0}\right\}.
\]
This containment, together with the bounded Lipschitz map $(r,u,v)\mapsto1(\tau=t)\{(r-u)^{2}-(r-v)^{2}\}$, Lemma~\ref{lem:transformed uniform entropy}, (\ref{eq:entropy dnn reg beta}), and the fixed response-class entropy, implies that, for constants $A>1$ and $C>0$,
\begin{equation}
\sup_{Q}\log N\left(C\epsilon,\mathcal{G}(s),\left\Vert \cdot\right\Vert _{Q,2}\right)\leq Cv_{\mu,m}\log(A/\epsilon),\qquad0<\epsilon<1,\label{eq:entropy loss class reg beta}
\end{equation}
uniformly over $s>0$.

\textbf{Step 4.} By (\ref{eq:variance reg beta}) and
(\ref{eq:entropy loss class reg beta}), \citet[Lemma 6.2]{chernozhukov2018Double}
gives, for every $s\geq s_{0}$ and all large $m$,
\[
\e\left[\sup_{g\in\mathcal{G}(s)}\left|\left(P_{m}-P\right)g\right|\right]
\leq C\left(\sqrt{\frac{sv_{\mu,m}\log m}{m}}+\frac{v_{\mu,m}\log m}{m}\right).
\]
Applying Bousquet's version of Talagrand's inequality
\citep[Theorem 2.3]{bousquet2002Bennett} together with
(\ref{eq:variance reg beta}) yields, for all $u\geq0$,
\begin{equation}
P\left(\sup_{g\in\mathcal{G}(s)}\left|\left(P_{m}-P\right)g\right|>
C\left[\sqrt{\frac{s(v_{\mu,m}\log m+u)}{m}}+\frac{v_{\mu,m}\log m+u}{m}\right]\right)\leq e^{-u}.\label{eq:localized concentration reg beta}
\end{equation}
Because $\left|R_{j}(Y,\bs X;\bs{\beta})\right|\leq M/4$ uniformly
over $\bs{\beta}\in\mathcal{B}_{0}$ and all functions in
$\mathcal{F}_{\mu,m}$ are bounded by $M$, the loss satisfies
$0\leq L(\bs Z,f(\bs X);\bs{\beta})\leq25M^{2}/16$ for every
$f\in\mathcal{F}_{\mu,m}$ and $\bs{\beta}\in\mathcal{B}_{0}$. Hence
$0\leq\mathscr{E}(f;\bs{\beta})\leq25M^{2}/16$. Let
$s_{k}:=2^{k}s_{0}$ for $k=0,1,\ldots,K_{m}$, where
$K_{m}:=\max\{0,\lceil\log_{2}(25M^{2}/(16s_{0}))\rceil\}$.
Then $s_{K_{m}}\geq25M^{2}/16$. Since $s_{0}\geq m^{-2}$, we also
have $K_{m}\leq\lceil\log_{2}(25M^{2}m^{2}/16)\rceil=O(\log m)$.
Thus the dyadic intervals $\{(s_{k-1},s_{k}]\}_{k=0}^{K_{m}}$, with
the convention $s_{-1}:=0$, cover all possible values of
$\mathscr{E}(f;\bs{\beta})$. Applying
(\ref{eq:localized concentration reg beta}) to each $\mathcal{G}(s_{k})$
with tail parameter $u+k\log2$ gives, for each $k=0,\ldots,K_{m}$,
\begin{align*}
&P\left(\sup_{g\in\mathcal{G}(s_{k})}\left|\left(P_{m}-P\right)g\right|>C\left[\sqrt{\frac{s_{k}(v_{\mu,m}\log m+u+k\log2)}{m}}\right.\right.\\
&\left.\left.\qquad+\frac{v_{\mu,m}\log m+u+k\log2}{m}\right]\right)\leq2^{-k}e^{-u}.
\end{align*}
By the union bound, the preceding display holds simultaneously for
all $k=0,\ldots,K_{m}$ w.p. at least $1-2e^{-u}$. Moreover,
since $k\leq K_{m}=O(\log m)$ and $v_{\mu,m}\geq1$ for
large $m$, the term $k\log2/m$ can be absorbed into
$v_{\mu,m}\log m/m$ after enlarging the
constant. Therefore, for constants $c_{1},c_{2},C>0$, there
exists an event $\mathcal{E}_{m}(u)$ with $P\{\mathcal{E}_{m}(u)\}\geq1-c_{1}\exp(-c_{2}u)$ such that, on $\mathcal{E}_{m}(u)$, for every $k=0,\ldots,K_{m}$,
\begin{equation}
\sup_{g\in\mathcal{G}(s_{k})}\left|\left(P_{m}-P\right)g\right|\leq C\left[\sqrt{\frac{s_{k}(v_{\mu,m}\log m+u)}{m}}+\frac{v_{\mu,m}\log m+u}{m}\right].\label{eq:peeling event reg beta}
\end{equation}

\textbf{Step 5.} Define $R_{m}:=\sup_{\bs{\beta}\in\mathcal{B}_{0}}\mathscr{E}(\widehat{\mu}(\cdot;\bs{\beta});\bs{\beta})$.
Since $\widehat{\mu}(\cdot;\bs{\beta})$ minimizes the empirical risk,
for every $\bs{\beta}\in\mathcal{B}_{0}$,
\begin{equation}
\mathscr{E}(\widehat{\mu}(\cdot;\bs{\beta});\bs{\beta})
\leq\left|\left(P_{m}-P\right)\left[L(\bs Z,\widehat{\mu}(\bs X;\bs{\beta});\bs{\beta})-L(\bs Z,\overline{\mu}(\bs X;\bs{\beta});\bs{\beta})\right]\right|.\label{eq:erm excess reg beta}
\end{equation}
On $\mathcal{E}_{m}(u)$, if $R_{m}\leq s_{0}$, then
$R_{m}\leq C(a_{\mu,m}+(v_{\mu,m}\log m+u)/m)$
for all large $m$, since $s_{0}=\max\{2a_{\mu,m},m^{-2}\}$
and $m^{-2}\leq(v_{\mu,m}\log m+u)/m$. If $R_{m}>s_{0}$, choose
$k\in\{0,\ldots,K_{m}-1\}$ such that $s_{k}<R_{m}\leq s_{k+1}$.
For any $\eta>0$, there exists
$\bs{\beta}_{\eta}\in\mathcal{B}_{0}$ such that
$R_{m}-\eta<\mathscr{E}(\widehat{\mu}(\cdot;\bs{\beta}_{\eta});\bs{\beta}_{\eta})\leq R_{m}$.
Then $\widehat{\mu}(\cdot;\bs{\beta}_{\eta})\in\mathscr{F}(s_{k+1};\bs{\beta}_{\eta})$,
so (\ref{eq:erm excess reg beta}) and (\ref{eq:peeling event reg beta})
imply $R_{m}-\eta\leq C[\sqrt{s_{k+1}(v_{\mu,m}\log m+u)/m}+(v_{\mu,m}\log m+u)/m]$.
Letting $\eta\to0$ and using $s_{k+1}\leq2R_{m}$, we get
$R_{m}\leq C[\sqrt{R_{m}(v_{\mu,m}\log m+u)/m}+(v_{\mu,m}\log m+u)/m]$.
The elementary inequality $r\leq A\sqrt{r}+B\Rightarrow r\leq2A^{2}+2B$
therefore gives
\begin{equation}
R_{m}\leq C\left(a_{\mu,m}+\frac{v_{\mu,m}\log m+u}{m}\right)\label{eq:restricted excess reg beta}
\end{equation}
on $\mathcal{E}_{m}(u)$.

\textbf{Step 6.} Combining (\ref{eq:curvature reg beta}) and
(\ref{eq:restricted excess reg beta}), on $\mathcal{E}_{m}(u)$,
\begin{align*}
&\sup_{\bs{\beta}\in\mathcal{B}_{0}}\left\Vert \widehat{\mu}(\bs X;\bs{\beta})-\mu^{*}(\bs X;\bs{\beta})\right\Vert _{P,2}\leq C\left(\sqrt{a_{\mu,m}}+\sqrt{\frac{v_{\mu,m}\log m}{m}}+\sqrt{\frac{u}{m}}\right)\\
\leq & C\left(\epsilon_{\mu,m}+\sqrt{\frac{v_{\mu,m}\log m}{m}}+\sqrt{\frac{u}{m}}\right),
\end{align*}
where the last inequality uses (\ref{eq:approx floor reg beta}).
Since $\Pi\subset\Pi_{\infty}$, the supremum over
$\{\bs{\beta}:\left\Vert \bs{\beta}-\bs{\beta}^{*}(\pi)\right\Vert \leq c_{0},\pi\in\Pi\}$
is bounded by the supremum over $\mathcal{B}_{0}$. Taking
$u$ to be a large enough multiple of $m\delta^{2}$
and absorbing constants into $C_{1},C_{2},c_{2}$ gives
\begin{align*}
&P\left(\sup_{\pi\in\Pi}\sup_{\left\Vert \bs{\beta}-\bs{\beta}^{*}(\pi)\right\Vert \leq c_{0}}
\left\Vert \widehat{\mu}_{I,jt}(\bs X;\bs{\beta})-\mu_{jt}^{*}(\bs X;\bs{\beta})\right\Vert _{P,2}\right.\\
&\qquad\left.>C_{1}\epsilon_{\mu,m}+C_{2}\sqrt{\frac{v_{\mu,m}\log m}{m}}+\delta\right)\leq c_{1}\exp\left(-c_{2}m\delta^{2}\right).
\end{align*}

It remains to verify that the two deterministic terms have the claimed
rates under the existing assumptions. By Assumption~\ref{assu:smoothness for mu},
\citet[Corollary 1.2]{lu2021Deep}, and
Assumption~\ref{assu:regularity-assumptions-on-nuisance}(ii), uniformly over
$j=0,\ldots,p$, $t\in\{0,1\}$ and $\bs{\beta}\in\mathcal{B}_{0}$,
\[
\begin{aligned}
  &\epsilon_{\mu,m}
\leq C\left(\mathcal{H}_{\mu}/\log\mathcal{H}_{\mu}\right)^{-2s_{\mu}/d}\left(\mathcal{D}_{\mu}/\log\mathcal{D}_{\mu}\right)^{-2s_{\mu}/d}=O\left(m^{-s_{\mu}/(2s_{\mu}+d)}(\log m)^{3}\right)=o(m^{-1/4}),\\
&\sqrt{\frac{v_{\mu,m}\log m}{m}}
=\sqrt{\frac{\mathcal{H}_{\mu}^{2}\mathcal{D}_{\mu}^{2}\log(\mathcal{H}_{\mu}\mathcal{D}_{\mu})\log m}{m}}=O\left(m^{-s_{\mu}/(2s_{\mu}+d)}(\log m)^{3}\right)=o(m^{-1/4}),
\end{aligned}
\]
where the first line also uses $s_{\mu}>d/2$.
Set $\rho_{\mu,m}:=C_{1}\epsilon_{\mu,m}+C_{2}\sqrt{v_{\mu,m}\log m/m}$.
The preceding displays imply $\rho_{\mu,m}=o(m^{-1/4})$, and the
claimed statement follows.
This completes the proof.
\end{proof}
\begin{lem}
\label{lem:converge of regression} Suppose that Assumptions
\ref{assu:unconfounded}, \ref{assu:Overlap assumption and smoothness},
\ref{assu:regularity-conditions-on-L},
\ref{assu:regularity-assumptions-on-U},
\ref{assu:smoothness for mu}, \ref{assu:regularity on mu}, and
\ref{assu:regularity-assumptions-on-nuisance} hold and $I\subset\left\{ 1,\ldots,N\right\} $
is an index set with $\left|I\right|=m$. Let $\Pi\subset\Pi_{\infty}$
be a policy class with VC dimension $\mathrm{VC}(\Pi)$ and $\mathrm{VC}(\Pi)/m\to0$
 as $m\to\infty$. If the DNN class $\mathcal{F}_{\mathrm{DNN}}(\mathcal{H}_{\mu},\mathcal{D}_{\mu})$
is constructed with $\mathcal{H}_{\mu}\mathcal{D}_{\mu}\asymp m^{d/(4s_{\mu}+2d)}\left(\log m\right)^{2}$,
then the DNN estimators $\widehat{\mu}_{I,jt}(\bs X;\widehat{\bs{\beta}}_{I}^{\mathrm{init}}(\pi))$
($t=0,1$ and $j=0,1,\ldots,p$) satisfy
\[
\begin{aligned}
&P\begin{Bmatrix}\sup_{\pi\in\Pi}\left\Vert \widehat{\mu}_{I,jt}(\bs X;\widehat{\bs{\beta}}_{I}^{\mathrm{init}}(\pi))-\mu_{jt}^{*}(\bs X;\bs{\beta}^{*}(\pi))\right\Vert _{P,2}\leq C\rho_{\mu,m}+C\sqrt{\frac{\mathrm{VC}(\Pi)}{m}}+C\rho_{e,m}+\delta\\
\widehat{\mu}_{I,jt}(\bs X;\widehat{\bs{\beta}}_{I}^{\mathrm{init}}(\pi))\in\mathcal{F}_{\mathrm{DNN}}(\mathcal{H}_{\mu},\mathcal{D}_{\mu})\bigcap\{f:\left\Vert f\right\Vert _{\infty}\le M\},\ \forall\pi\in\Pi
\end{Bmatrix}\\
\geq& \ 1-c_{1}\exp\left(-c_{2}m\delta^{2}\right)
\end{aligned}
\]
for all $0<\delta<c_{3}$, where $c_{1},c_{2},c_{3}>0$ are finite
constants independent of $\delta$, $I$, $m$, and $\Pi$,
and $\rho_{\mu,m}$ is a non-negative sequence satisfying
$\rho_{\mu,m}=o(m^{-1/4})$.
\end{lem}
\begin{proof}
We focus on establishing the result for $\widehat{\mu}_{I,01}(\bs X;\widehat{\bs{\beta}}_{I}^{\mathrm{init}}(\pi))$,
as the proofs for the remaining cases follow analogously. According
to Lemma \ref{lem:convergence of initial beta}, we have
\begin{equation}
\sup_{\pi\in\Pi}\left\Vert \widehat{\bs{\beta}}_{I}^{\mathrm{init}}(\pi)-\bs{\beta}^{*}(\pi)\right\Vert \leq c_{0}\label{eq:beta-init-near-beta-pi}
\end{equation}
w.p. at least $1-c_{1}\exp\left(-c_{2}m\delta^{2}\right)$
for all $0<\delta<c_{4}$ and large $m$, where $c_{4}>0$
is finite. Applying Lemma \ref{lem:convergence of reg at any beta}
we have
  \begin{align*}
 & P\left(\sup_{\pi\in\Pi}\left\Vert \widehat{\mu}_{I,01}(\bs X;\widehat{\bs{\beta}}_{I}^{\mathrm{init}}(\pi))-\mu_{01}^{*}(\bs X;\widehat{\bs{\beta}}_{I}^{\mathrm{init}}(\pi))\right\Vert _{P,2}>\rho_{\mu,m}+\delta\right)\\
\leq & P\left(\sup_{\pi\in\Pi}\sup_{\left\Vert \bs{\beta}-\bs{\beta}^{*}(\pi)\right\Vert \leq c_{0}}\left\Vert \widehat{\mu}_{I,01}(\bs X;\bs{\beta})-\mu_{01}^{*}(\bs X;\bs{\beta})\right\Vert _{P,2}>\rho_{\mu,m}+\delta\right)\\
&\quad+P\left(\sup_{\pi\in\Pi}\left\Vert \widehat{\bs{\beta}}_{I}^{\mathrm{init}}(\pi)-\bs{\beta}^{*}(\pi)\right\Vert>c_{0}\right)\\
\leq & c_{1}\exp\left(-c_{2}m\delta^{2}\right)+c_{1}\exp\left(-c_{2}m\delta^{2}\right)\leq2c_{1}\exp\left(-c_{2}m\delta^{2}\right)
\end{align*}
for all $0<\delta<c_{4}$ and large $m$. In addition,
by Assumption~\ref{assu:regularity on mu} and Lemma~\ref{lem:convergence of initial beta}
we have
\begin{align*}
\sup_{\pi\in\Pi}\left\Vert \mu_{01}^{*}(\bs X;\widehat{\bs{\beta}}_{I}^{\mathrm{init}}(\pi))-\mu_{01}^{*}(\bs X;\bs{\beta}^{*}(\pi))\right\Vert _{P,2}
&\leq L_{\mu}\sup_{\pi\in\Pi}\left\Vert \widehat{\bs{\beta}}_{I}^{\mathrm{init}}(\pi)-\bs{\beta}^{*}(\pi)\right\Vert\\
&\leq C\sqrt{\frac{\mathrm{VC}(\Pi)}{m}}+C\rho_{e,m}+L_{\mu}\delta
\end{align*}
w.p. at least $1-c_{1}\exp\left(-c_{2}m\delta^{2}\right)$
for all $0<\delta<c_{3}$ and large $m$. Combining the
last two displays, we have
  \begin{align*}
 & \sup_{\pi\in\Pi}\left\Vert \widehat{\mu}_{I,01}(\bs X;\widehat{\bs{\beta}}_{I}^{\mathrm{init}}(\pi))-\mu_{01}^{*}(\bs X;\bs{\beta}^{*}(\pi))\right\Vert _{P,2}\leq\sup_{\pi\in\Pi}\left\Vert \widehat{\mu}_{I,01}(\bs X;\widehat{\bs{\beta}}_{I}^{\mathrm{init}}(\pi))-\mu_{01}^{*}(\bs X;\widehat{\bs{\beta}}_{I}^{\mathrm{init}}(\pi))\right\Vert _{P,2}\\
 & \quad+\sup_{\pi\in\Pi}\left\Vert \mu_{01}^{*}(\bs X;\widehat{\bs{\beta}}_{I}^{\mathrm{init}}(\pi))-\mu_{01}^{*}(\bs X;\bs{\beta}^{*}(\pi))\right\Vert _{P,2}\leq\rho_{\mu,m}+C\sqrt{\frac{\mathrm{VC}(\Pi)}{m}}+C\rho_{e,m}+(L_{\mu}+1)\delta
\end{align*}
w.p. at least $1-3c_{1}\exp\left(-c_{2}m\delta^{2}\right)$
for all $0<\delta<\min\{c_{3},c_{4}\}$ and large $m$.

Besides, under (\ref{eq:beta-init-near-beta-pi}), by construction
and Assumption~\ref{assu:regularity-assumptions-on-nuisance}(i) we also have
\(\widehat{\mu}_{I,01}(\bs X;\widehat{\bs{\beta}}_{I}^{\mathrm{init}}(\pi))\in\mathcal{F}_{\mathrm{DNN}}(\mathcal{H}_{\mu},\mathcal{D}_{\mu})\cap\{f:\left\Vert f\right\Vert _{\infty}\le M\}\) for all \(\pi\in\Pi\).
Thus, for all $0<\delta<\min\{c_{3},c_{4}\}$ and large $m$, w.p. at least $1-4c_{1}\exp\left(-c_{2}m\delta^{2}\right)$,
\[
\sup_{\pi\in\Pi}\left\Vert \widehat{\mu}_{I,01}(\bs X;\widehat{\bs{\beta}}_{I}^{\mathrm{init}}(\pi))-\mu_{01}^{*}(\bs X;\bs{\beta}^{*}(\pi))\right\Vert _{P,2}
\leq\rho_{\mu,m}+C\sqrt{\frac{\mathrm{VC}(\Pi)}{m}}+C\rho_{e,m}+(L_{\mu}+1)\delta,
\]
and the membership condition above holds for all \(\pi\in\Pi\).
This completes the proof.
\end{proof}

\subsection[Asymptotic properties of calibrated weights]{The asymptotic properties of the weights $\widehat{w}_{I,i}(\pi)$}

To begin with, we derive the asymptotic properties of the weights
$\widehat{w}_{I,i}(\pi)$. Without loss of generality, we assume
$I=\{1,\ldots,m\}$.

We let $\widehat{\bs{\xi}}_{0}(\bs X):=(\widehat{\mu}_{I,j0}(\bs X;\widehat{\bs{\beta}}_{I}^{\mathrm{init}}(\pi)):j=0,1,\ldots,p)^{\trans}$ and $\widehat{\bs{\xi}}_{1}(\bs X)=(\widehat{\mu}_{I,j1}(\bs X;\widehat{\bs{\beta}}_{I}^{\mathrm{init}}(\pi)):j=0,1,\ldots,p)^{\trans}$.
We also write $\widehat{e}:=\widehat{e}_{I}$ and $\widehat{\bs{\xi}}(\bs X):=(\widehat{\bs{\xi}}_{0}(\bs X)^{\trans},\widehat{\bs{\xi}}_{1}(\bs X)^{\trans})^{\trans}$, and let
\[
\bs A(T,\bs X;e,\bs{\xi},\pi):=\frac{(1-T)(1-\pi(\bs X))}{1-e(\bs X)}\bs{\xi}_{0}(\bs X)+\frac{T\pi(\bs X)}{e(\bs X)}\bs{\xi}_{1}(\bs X).
\]
Write \(\bs A_{i}:=\bs A(T_{i},\bs X_{i};\widehat{e},\widehat{\bs{\xi}},\pi)\) for \(i\in I\). Besides, let $\bs w=(w_{1},\ldots,w_{m})^{\trans}$ and $\bs B=m^{-1}\sum_{i\in I}\{ (1-\pi(\bs X_{i}))\widehat{\bs{\xi}}_{0}(\bs X_{i})+\pi(\bs X_{i})\widehat{\bs{\xi}}_{1}(\bs X_{i})\}$.
For any $v\in\R$, denote the derivative of $D(v)$ by $D^{\prime}(v)$.
Then $D^{\prime}(1)=0$ and $D^{\prime\prime}(1)>0$. Let $D(\bs w)=\sum_{i\in I}D(w_{i})$.
Then the calibration problem can be rewritten as
\begin{equation}
\left(\widehat{w}_{I,i}(\pi)\right)_{i\in I}=\arg\min_{w_{i}>0:i\in I}D(\bs w)\ \text{ s.t. }\ \left(\bs A_{1},\ldots,\bs A_{m}\right)\bs w=m\bs B.\label{eq:EL step rewritten}
\end{equation}
Note that the conjugate function of $D(\bs w)$ is
\[
D^{*}(\bs z)=\sum_{i\in I}\left\{ z_{i}\cdot(D^{\prime})^{-1}(z_{i})-D\left\{ (D^{\prime})^{-1}(z_{i})\right\} \right\} =\sum_{i\in I}-\rho(-z_{i}),\quad\forall\bs z=(z_{1},\ldots,z_{m})^{\trans},
\]
where $\rho(v):=D\left\{ (D^{\prime})^{-1}(-v)\right\} +v\cdot(D^{\prime})^{-1}(-v)$
for any $v\in\R$. The dual problem of (\ref{eq:EL step rewritten})
is
\begin{align*}
&\max_{\bs{\lambda}\in\R^{p+1}}\left\{ \bs{\lambda}^{\trans}m\bs B-D^{*}((\bs A_{1},\ldots,\bs A_{m})^{\trans}\bs{\lambda})\right\}=\max_{\bs{\lambda}\in\R^{p+1}}\sum_{i\in I}\left\{ \bs{\lambda}^{\trans}\bs B+\rho(-\bs{\lambda}^{\trans}\bs A_{i})\right\} \\
&\quad=\max_{\bs{\lambda}\in\R^{p+1}}\frac{1}{m}\sum_{i\in I}\left\{ \rho(\bs{\lambda}^{\trans}\bs A_{i})-\bs{\lambda}^{\trans}\bs B\right\} .
\end{align*}
Let $\hat{G}(\bs{\lambda};\pi):=\frac{1}{m}\sum_{i\in I}\left\{ \rho(\bs{\lambda}^{\trans}\bs A_{i})-\bs{\lambda}^{\trans}\bs B\right\} $
and $\hat{\bs{\lambda}}_{I}(\pi):=\arg\max_{\bs{\lambda}\in\R^{p+1}}\hat{G}(\bs{\lambda};\pi)$.
On any event on which the affine constraint in (\ref{eq:EL step rewritten})
admits a vector $\bs w=(w_i)_{i\in I}$ with $w_i>0$ for every $i\in I$, Slater's condition for this
equality-constrained entropy program holds, and strong duality follows from
Section~5 of \citet{boyd2004Convex}. The bounded-dual argument in
Lemma~\ref{lem:convergence rate for lambda} below verifies this strict
feasibility on the high-probability event used in the asymptotic analysis.
On this event, the first-order condition of the dual problem gives
\begin{equation}
\widehat{w}_{I,i}(\pi)=\rho^{\prime}(\hat{\bs{\lambda}}_{I}(\pi)^{\trans}\bs A_{i}),\text{ for all }i\in I,\label{eq:expression for wi}
\end{equation}
where $\rho^{\prime}(\cdot)$ is the derivative of $\rho(\cdot)$.
We now give the asymptotic behavior of $\widehat{w}_{I,i}(\pi)$,
which is critical in the subsequent analysis.
\begin{lem}
\label{lem:convergence rate for lambda} Suppose that Assumptions
\ref{assu:unconfounded}, \ref{assu:Overlap assumption and smoothness},
\ref{assu:regularity-conditions-on-L},
\ref{assu:regularity-assumptions-on-U},
\ref{assu:smoothness for mu}, \ref{assu:regularity on mu}, and
\ref{assu:regularity-assumptions-on-nuisance} hold and $I\subset\left\{ 1,\ldots,N\right\} $
is an index set with $\left|I\right|=m$. Let $\Pi\subset\Pi_{\infty}$
be a policy class with VC dimension $\mathrm{VC}(\Pi)$ and $\mathrm{VC}(\Pi)/m\to0$
 as $m\to\infty$. Then
  \[
P\left(\sup_{\pi\in\Pi}\left\Vert \hat{\bs{\lambda}}_{I}(\pi)\right\Vert \geq C\sqrt{\frac{\mathrm{VC}(\Pi)}{m}}+C\rho_{e,m}+C\rho_{\mu,m}+\delta\right)\leq c_{1}\exp\left(-c_{2}m\delta^{2}\right)
\]
for all $0<\delta<c_{3}$ and large $m$, where $c_{1},c_{2},c_{3}>0$
are finite constants independent of $\delta$, $I$, $m$
and $\Pi$, and $C>0$ is finite.
\end{lem}
\begin{proof}
Note that $\rho^{\prime}(v)=(-v)\frac{d\{(D^{\prime})^{-1}(-v)\}}{dv}+v\frac{d\{(D^{\prime})^{-1}(-v)\}}{dv}+(D^{\prime})^{-1}(-v)=(D^{\prime})^{-1}(-v)$ and $\rho^{\prime\prime}(v)=-1/D^{\prime\prime}((D^{\prime})^{-1}(-v))$.
Since $D^{\prime}(1)=0$, $D^{\prime\prime}(1)>0$ and $D^{\prime\prime}(v),(D^{\prime})^{-1}(v)$
are continuous, we have $\rho^{\prime}(0)=1$, $\rho^{\prime\prime}(0)<0$
and $\rho^{\prime\prime}(v)$ is continuous. Since $D(v)$ is strictly
convex, $\rho(v)$ is also strictly concave.

In this proof we write $\mathcal{F}_{\mathrm{DNN},e}$ and $\mathcal{F}_{\mathrm{DNN},\mu}$
for $\text{logistic}\circ\left\{ \mathcal{F}_{\mathrm{DNN}}(\mathcal{H}_{e},\mathcal{D}_{e})\bigcap\{f:\left\Vert f\right\Vert _{\infty}\le M\}\right\} $
and $\mathcal{F}_{\mathrm{DNN}}(\mathcal{H}_{\mu},\mathcal{D}_{\mu})\bigcap\{f:\left\Vert f\right\Vert _{\infty}\le M\}$,
respectively. We let $v_{e,m}:=\mathcal{H}_{e}^{2}\mathcal{D}_{e}^{2}\log(\mathcal{H}_{e}\mathcal{D}_{e})$, $v_{\mu,m}:=\mathcal{H}_{\mu}^{2}\mathcal{D}_{\mu}^{2}\log(\mathcal{H}_{\mu}\mathcal{D}_{\mu})$, and $a_{m}(\Pi):=\sqrt{\mathrm{VC}(\Pi)/m}+\rho_{e,m}+\rho_{\mu,m}$.
By Assumption~\ref{assu:Overlap assumption and smoothness} and Assumption~\ref{assu:regularity-assumptions-on-nuisance}(ii), we have $\sqrt{v_{e,m}/m}=o(m^{-1/4})$ and $\sqrt{v_{\mu,m}/m}=o(m^{-1/4})$. Enlarging $\rho_{e,m}$ and $\rho_{\mu,m}$ if necessary, we may assume $\sqrt{v_{e,m}/m}\leq\rho_{e,m}$ and $\sqrt{v_{\mu,m}/m}\leq\rho_{\mu,m}$ for all large $m$.

Note that
\[
\frac{\partial\hat{G}(\bs{\lambda};\pi)}{\partial\bs{\lambda}}=\frac{1}{m}\sum_{i\in I}\left\{ \rho^{\prime}(\bs{\lambda}^{\trans}\bs A_{i})\bs A_{i}-\bs B\right\} \quad\text{and }\frac{\partial^{2}\hat{G}(\bs{\lambda};\pi)}{\partial\bs{\lambda}\partial\bs{\lambda}^{\trans}}=\frac{1}{m}\sum_{i\in I}\left\{ \rho^{\prime\prime}(\bs{\lambda}^{\trans}\bs A_{i})\bs A_{i}\bs A_{i}^{\trans}\right\} .
\]
Then we have
\begin{equation}
\frac{\partial\hat{G}(\bs 0;\pi)}{\partial\bs{\lambda}}=\frac{1}{m}\sum_{i\in I}\bs A_{i}-\bs B=\left\{ P_{m}-P\right\} \bs{\phi}(T,\bs X;\widehat{e},\widehat{\bs{\xi}},\pi)+P\bs{\phi}(T,\bs X;\widehat{e},\widehat{\bs{\xi}},\pi)\label{eq:G-prime-zero}
\end{equation}
where $\bs{\phi}(T,\bs X;e,\bs{\xi},\pi):=\left\{ \frac{1-T}{1-e(\bs X)}-1\right\} (1-\pi(\bs X))\bs{\xi}_{0}(\bs X)+\left\{ \frac{T}{e(\bs X)}-1\right\} \pi(\bs X)\bs{\xi}_{1}(\bs X)$.
We bound $P\bs{\phi}(T,\bs X;\widehat{e},\widehat{\bs{\xi}},\pi)$
and $\left\{ P_{m}-P\right\} \bs{\phi}(T,\bs X;\widehat{e},\widehat{\bs{\xi}},\pi)$,
respectively. By Lemma~\ref{lem:converge of propensity},
Lemma~\ref{lem:converge of regression}, and the membership of the nuisance
estimators in the corresponding bounded DNN classes,
we have
\[
\sup_{\pi\in\Pi}\left\Vert P\bs{\phi}(T,\bs X;\widehat{e},\widehat{\bs{\xi}},\pi)\right\Vert \leq\sup_{\pi\in\Pi}\sup_{e\in\mathcal{F}_{\mathrm{DNN},e}:\left\Vert e-e^{*}\right\Vert _{P,2}\leq\rho_{e,m}+t}\sup_{\bs{\xi}\in\mathcal{F}_{\mathrm{DNN},\mu}^{2p+2}}\left\Vert P\bs{\phi}(T,\bs X;e,\bs{\xi},\pi)\right\Vert
\]
w.p. at least $1-c_{1}\exp(-c_{2}mt^{2})$ for all $0<t<c_{3}$
and large $m$. Note that by the definitions of $\mathcal{F}_{\mathrm{DNN},e}$
and $\mathcal{F}_{\mathrm{DNN},\mu}$, we have
\begin{align*}
\left\Vert P\bs{\phi}(T,\bs X;e,\bs{\xi},\pi)\right\Vert
&=\left\Vert \e\left[\left\{ \frac{e(\bs X)-e^{*}(\bs X)}{1-e(\bs X)}\right\} (1-\pi(\bs X))\bs{\xi}_{0}(\bs X)\right.\right.\\
&\qquad\left.\left.+\left\{ \frac{e^{*}(\bs X)-e(\bs X)}{e(\bs X)}\right\} \pi(\bs X)\bs{\xi}_{1}(\bs X)\right]\right\Vert\lesssim\left\Vert e(\bs X)-e^{*}(\bs X)\right\Vert _{P,2}.
\end{align*}
Thus,
\begin{equation}
\sup_{\pi\in\Pi}\left\Vert P\bs{\phi}(T,\bs X;\widehat{e},\widehat{\bs{\xi}},\pi)\right\Vert \leq C\rho_{e,m}+Ct\label{eq:partial G first term}
\end{equation}
w.p. at least $1-c_{1}\exp(-c_{2}mt^{2})$ for all $0<t<c_{3}$
and large $m$. Furthermore, by Lemmas~\ref{lem:converge of propensity}
and \ref{lem:converge of regression}, and the membership of the nuisance
estimators in the corresponding bounded DNN classes, we have
\begin{equation}
\sup_{\pi\in\Pi}\left\Vert \left\{ P_{m}-P\right\} \bs{\phi}(T,\bs X;\widehat{e},\widehat{\bs{\xi}},\pi)\right\Vert \leq\sup_{f\in\mathcal{M}_{1}}\left|\left\{ P_{m}-P\right\} f(T,\bs X)\right|,\label{eq:partial G second term}
\end{equation}
w.p. at least $1-c_{1}\exp(-c_{2}mt^{2})$ for all $0<t<c_{3}$
and large $m$, where
\begin{align*}
\mathcal{M}_{1} & :=\{(T,\bs X)\mapsto\bs{\alpha}^{\trans}\bs{\phi}(T,\bs X;e,\bs{\xi},\pi):\bs{\xi}\in\mathcal{F}_{\mathrm{DNN},\mu}^{2p+2},\bs{\alpha}\in\mathbb{S}^{p},e\in\mathcal{F}_{\mathrm{DNN},e},\pi\in\Pi\}
\end{align*}
is a class of functions with a measurable envelope $C$. Here $\mathbb{S}^{p}\assign\{\bs{\alpha}\in\mathbb{R}^{p+1}:\left\Vert \bs{\alpha}\right\Vert =1\}$ denotes the unit sphere. By the VC-subgraph property of $\Pi$, (\ref{eq:entropy for DNNe}), (\ref{eq:entropy dnn reg beta}), overlap, the fixed dimension of $\mathbb{S}^{p}$, and Lemma~\ref{lem:transformed uniform entropy},
\begin{equation}
\sup_{Q}\log N\left(C\epsilon,\mathcal{M}_{1},\left\Vert \cdot\right\Vert _{Q,2}\right)\leq C\left\{ \mathrm{VC}(\Pi)+v_{e,m}+v_{\mu,m}\right\}\log\left(a/\epsilon\right).
\label{eq:entropy for M1}
\end{equation}
for all $0<\epsilon<1$, where $a>1$ is a constant. By Lemma
6.2 in \citet{chernozhukov2018Double} we have
  \[
\e\left[\sup_{f\in\mathcal{M}_{1}}\left|\left\{ P_{m}-P\right\} f(T,\bs X)\right|\right]\leq C\sqrt{\frac{\mathrm{VC}(\Pi)}{m}}+C\sqrt{\frac{v_{e,m}}{m}}+C\sqrt{\frac{v_{\mu,m}}{m}}.
\]
Now, applying the bounded differences inequality (Corollary 2.21 in
\citet{wainwright2019HighDimensional}) gives that
\[
P\left(\sup_{f\in\mathcal{M}_{1}}\left|\left\{ P_{m}-P\right\} f(T,\bs X)\right|\geq\e\left[\sup_{f\in\mathcal{M}_{1}}\left|\left\{ P_{m}-P\right\} f(T,\bs X)\right|\right]+t\right)\leq c_{1}\exp\left(-c_{2}mt^{2}\right)
\]
for all $t\geq0$. Combining the above two displays yields that
  \[
P\left(\sup_{f\in\mathcal{M}_{1}}\left|\left\{ P_{m}-P\right\} f(T,\bs X)\right|\geq C\sqrt{\frac{\mathrm{VC}(\Pi)}{m}}+C\sqrt{\frac{v_{e,m}}{m}}+C\sqrt{\frac{v_{\mu,m}}{m}}+t\right)\leq c_{1}\exp\left(-c_{2}mt^{2}\right)
\]
for all $t\geq0$. Recalling (\ref{eq:G-prime-zero}), (\ref{eq:partial G first term}),
and (\ref{eq:partial G second term}), together with $\sqrt{v_{e,m}/m}\leq\rho_{e,m}$ and $\sqrt{v_{\mu,m}/m}\leq\rho_{\mu,m}$,
we
can obtain that
  \begin{equation}
P\left(\sup_{\pi\in\Pi}\left\Vert \frac{\partial\hat{G}(\bs 0;\pi)}{\partial\bs{\lambda}}\right\Vert \geq Ca_{m}(\Pi)+t\right)\leq c_{1}\exp\left(-c_{2}mt^{2}\right)\label{eq:EL step bound for derivative}
\end{equation}
for all $0<t<c_{3}$ and large $m$.

Now, we analyze $m^{-1}\sum_{i\in I}\bs A_{i}\bs A_{i}^{\trans}$. Using the definition of $\bs A(T,\bs X;e,\bs{\xi},\pi)$ above, let
	\[
	\begin{aligned}
	\mathcal{M}_{2}:=\{(T,\bs X)\mapsto{}&\bs{\alpha}^{\trans}\bs A(T,\bs X;e,\bs{\xi},\pi)\bs A(T,\bs X;e,\bs{\xi},\pi)^{\trans}\bs{\alpha}:\\
	&\bs{\xi}\in\mathcal{F}_{\mathrm{DNN},\mu}^{2p+2},\bs{\alpha}\in\mathbb{S}^{p},e\in\mathcal{F}_{\mathrm{DNN},e},\pi\in\Pi\}
	\end{aligned}
	\]
	be a function class with a measurable envelope $C$. By the same argument as for $\mathcal{M}_{1}$, the components of $\bs A$ have bounded envelopes and the same entropy bounds. The fixed-dimensional sphere $\mathbb{S}^{p}$ contributes only a constant-order entropy term, and the map $(\bs a,\bs{\alpha})\mapsto\bs{\alpha}^{\trans}\bs a\bs a^{\trans}\bs{\alpha}$ is Lipschitz on bounded sets. Lemma~\ref{lem:transformed uniform entropy} gives
  \[
\sup_{Q}\log N\left(C\epsilon,\mathcal{M}_{2},\left\Vert \cdot\right\Vert _{Q,2}\right)\leq C\left\{ \mathrm{VC}(\Pi)+v_{e,m}+v_{\mu,m}\right\} \log\left(a/\epsilon\right)
\]
for all $0<\epsilon<1$, where $a>1$ is a constant. By
\citet[Corollary 2.21]{wainwright2019HighDimensional}, we have
\[
P\left(\sup_{f\in\mathcal{M}_{2}}\left|\left\{ P_{m}-P\right\} f(T,\bs X)\right|\geq\e\left[\sup_{f\in\mathcal{M}_{2}}\left|\left\{ P_{m}-P\right\} f(T,\bs X)\right|\right]+t\right)\leq c_{1}\exp\left(-c_{2}mt^{2}\right)
\]
for all $t\geq0$. Besides, by \citet[Lemma 6.2]{chernozhukov2018Double}
we have
  \[
\e\left[\sup_{f\in\mathcal{M}_{2}}\left|\left\{ P_{m}-P\right\} f(T,\bs X)\right|\right]\leq C\sqrt{\frac{\mathrm{VC}(\Pi)}{m}}+C\sqrt{\frac{v_{e,m}}{m}}+C\sqrt{\frac{v_{\mu,m}}{m}}.
\]
Combining the above two displays we have
  \begin{align*}
	&\sup_{\pi\in\Pi,\bs{\xi}\in\mathcal{F}_{\mathrm{DNN},\mu}^{2p+2},e\in\mathcal{F}_{\mathrm{DNN},e}}\left\Vert \left\{ P_{m}-P\right\} \bs A(T,\bs X;e,\bs{\xi},\pi)\bs A(T,\bs X;e,\bs{\xi},\pi)^{\trans}\right\Vert \\
	&\leq \sup_{f\in\mathcal{M}_{2}}\left|\left\{ P_{m}-P\right\} f(T,\bs X)\right|\leq C\sqrt{\frac{\mathrm{VC}(\Pi)}{m}}+C\sqrt{\frac{v_{e,m}}{m}}+C\sqrt{\frac{v_{\mu,m}}{m}}+t
	\end{align*}
	w.p. at least $1-c_{1}\exp\left(-c_{2}mt^{2}\right)$
		for all $t\geq0$. By Weyl's inequality, we have
		  \begin{align}
	&\sup_{\pi\in\Pi,\bs{\xi}\in\mathcal{F}_{\mathrm{DNN},\mu}^{2p+2},e\in\mathcal{F}_{\mathrm{DNN},e}}\bigl|\lambda_{\min}\left\{ P_{m}\bs A(T,\bs X;e,\bs{\xi},\pi)\bs A(T,\bs X;e,\bs{\xi},\pi)^{\trans}\right\}\nonumber\\
	&\qquad\qquad-\lambda_{\min}\left\{ P\bs A(T,\bs X;e,\bs{\xi},\pi)\bs A(T,\bs X;e,\bs{\xi},\pi)^{\trans}\right\} \bigr|\nonumber\\
	&\leq \sup_{\pi\in\Pi,\bs{\xi}\in\mathcal{F}_{\mathrm{DNN},\mu}^{2p+2},e\in\mathcal{F}_{\mathrm{DNN},e}}\left\Vert \left\{ P_{m}-P\right\} \bs A(T,\bs X;e,\bs{\xi},\pi)\bs A(T,\bs X;e,\bs{\xi},\pi)^{\trans}\right\Vert \nonumber\\
	&\leq C\sqrt{\frac{\mathrm{VC}(\Pi)}{m}}+C\sqrt{\frac{v_{e,m}}{m}}+C\sqrt{\frac{v_{\mu,m}}{m}}+t\leq Ca_{m}(\Pi)+t.\label{eq:EL step empirical singular value bound}
	\end{align}
	w.p. at least $1-c_{1}\exp\left(-c_{2}mt^{2}\right)$
	for all $t\geq0$. Let $\bs{\xi}^{*}(\bs X;\pi):=(\bs{\xi}_{0}^{*}(\bs X;\pi)^{\trans},\allowbreak\bs{\xi}_{1}^{*}(\bs X;\pi)^{\trans})^{\trans}$. By Lemmas~\ref{lem:converge of propensity}
	and \ref{lem:converge of regression}, and a finite union bound over
	$j=0,\ldots,p$ and $t\in\{0,1\}$,
	\[
	\left\Vert \widehat{e}(\bs X)-e^{*}(\bs X)\right\Vert _{P,2}+\sup_{\pi\in\Pi}\max_{t=0,1}\left\Vert \widehat{\bs{\xi}}_{t}(\bs X)-\bs{\xi}_{t}^{*}(\bs X;\pi)\right\Vert _{P,2}\leq Ca_{m}(\Pi)+Ct,
	\]
	w.p. at least $1-c_{1}\exp(-c_{2}mt^{2})$ for all $0<t<c_{3}$ and large $m$. Hence, by Assumptions~\ref{assu:Overlap assumption and smoothness}, \ref{assu:smoothness for mu}, and \ref{assu:regularity-assumptions-on-nuisance}(i),
		\begin{equation}
		\begin{aligned}
		&\sup_{\pi\in\Pi}\bigl\Vert P\bs A(T,\bs X;\widehat{e},\widehat{\bs{\xi}},\pi)\bs A(T,\bs X;\widehat{e},\widehat{\bs{\xi}},\pi)^{\trans}\\
		&\qquad-P\bs A(T,\bs X;e^{*},\bs{\xi}^{*}(\pi),\pi)\bs A(T,\bs X;e^{*},\bs{\xi}^{*}(\pi),\pi)^{\trans}\bigr\Vert\leq Ca_{m}(\Pi)+Ct
		\end{aligned}
		\label{eq:EL step population perturbation bound}
		\end{equation}
	w.p. at least $1-c_{1}\exp(-c_{2}mt^{2})$ for all $0<t<c_{3}$ and large $m$. Moreover,
	\[
	\begin{aligned}
	&P\bs A(T,\bs X;e^{*},\bs{\xi}^{*}(\pi),\pi)\bs A(T,\bs X;e^{*},\bs{\xi}^{*}(\pi),\pi)^{\trans}=\e\left[\frac{1-\pi(\bs X)}{1-e^{*}(\bs X)}\bs{\xi}_{0}^{*}(\bs X;\pi)\bs{\xi}_{0}^{*}(\bs X;\pi)^{\trans}\right.\\
	&\qquad\left.+\frac{\pi(\bs X)}{e^{*}(\bs X)}\bs{\xi}_{1}^{*}(\bs X;\pi)\bs{\xi}_{1}^{*}(\bs X;\pi)^{\trans}\right],
	\end{aligned}
	\]
	so, because $\Pi\subset\Pi_{\infty}$, Assumption~\ref{assu:regularity on mu}(ii) and overlap imply
		\begin{equation}
	\inf_{\pi\in\Pi}\lambda_{\min}\left\{ P\bs A(T,\bs X;e^{*},\bs{\xi}^{*}(\pi),\pi)\bs A(T,\bs X;e^{*},\bs{\xi}^{*}(\pi),\pi)^{\trans}\right\}\geq c_{\xi}.
		\label{eq:EL step true weighted minimum singular value}
		\end{equation}
		Because $a_{m}(\Pi)\to0$, choose $c_{4}>0$ small enough and then $m$ large enough so that $Ca_{m}(\Pi)+Ct\leq c_{\xi}/2$ for all $0<t<c_{4}$.
		Combining \eqref{eq:EL step empirical singular value bound}, \eqref{eq:EL step population perturbation bound}, and \eqref{eq:EL step true weighted minimum singular value}, w.p. at least $1-c_{1}\exp\left(-c_{2}mt^{2}\right)$, for all $0<t<c_{4}$,
		  \begin{align}
	 & \inf_{\pi\in\Pi}\lambda_{\min}\left\{ \frac{1}{m}\sum_{i\in I}\bs A_{i}\bs A_{i}^{\trans}\right\} \nonumber
	=  \inf_{\pi\in\Pi}\lambda_{\min}\left\{ \frac{1}{m}\sum_{i\in I}\bs A(T_{i},\bs X_{i};\widehat{e},\widehat{\bs{\xi}},\pi)\bs A(T_{i},\bs X_{i};\widehat{e},\widehat{\bs{\xi}},\pi)^{\trans}\right\} \nonumber \\
	\geq & \inf_{\pi\in\Pi}\lambda_{\min}\left\{ P\bs A(T,\bs X;\widehat{e},\widehat{\bs{\xi}},\pi)\bs A(T,\bs X;\widehat{e},\widehat{\bs{\xi}},\pi)^{\trans}\right\} -Ca_{m}(\Pi)-t\nonumber\\
	\geq & \inf_{\pi\in\Pi}\lambda_{\min}\left\{ P\bs A(T,\bs X;e^{*},\bs{\xi}^{*}(\pi),\pi)\bs A(T,\bs X;e^{*},\bs{\xi}^{*}(\pi),\pi)^{\trans}\right\} -Ca_{m}(\Pi)-Ct\geq c_{\xi}/2\label{eq:EL step minimum singular value}
	\end{align}
		where the equality only expands $\bs A_{i}=\bs A(T_{i},\bs X_{i};\widehat{e},\widehat{\bs{\xi}},\pi)$, which depends on $\pi$.
		The first inequality follows from \eqref{eq:EL step empirical singular value bound} evaluated at $e=\widehat e$ and $\bs{\xi}=\widehat{\bs{\xi}}$.
		The second inequality follows from Weyl's inequality and the perturbation bound \eqref{eq:EL step population perturbation bound}.
		The last inequality follows from \eqref{eq:EL step true weighted minimum singular value} and the preceding choice of $c_{4}$ and $m$.

Let $\Upsilon(t):=\{\bs{\lambda}\in\R^{p+1}:\left\Vert \bs{\lambda}\right\Vert \leq C_{\Upsilon}(a_{m}(\Pi)+t)\}$
and $\partial\Upsilon(t):=\{\bs{\lambda}\in\R^{p+1}:\left\Vert \bs{\lambda}\right\Vert =C_{\Upsilon}(a_{m}(\Pi)+t)\}$,
where the constant $C_{\Upsilon}$ will be determined later. For any
$\bs{\lambda}\in\Upsilon(t)$, by Assumption~\ref{assu:regularity-assumptions-on-nuisance}
and Lemma~\ref{lem:converge of regression} we have
\[
\begin{aligned}
&\sup_{\pi\in\Pi}\sup_{i\in I}\sup_{\bs{\lambda}\in\Upsilon(t)}\left|\bs{\lambda}^{\trans}\bs A_{i}\right|\leq\sup_{\bs{\lambda}\in\Upsilon(t)}\left\Vert \bs{\lambda}\right\Vert \sup_{i\in I}\sup_{\pi\in\Pi}\left\Vert \bs A_{i}\right\Vert\leq CC_{\Upsilon}(a_{m}(\Pi)+t)
\end{aligned}
\]
w.p. at least $1-c_{1}\exp\left(-c_{2}mt^{2}\right)$
for all $0<t<c_{3}$ and large $m$. Then by $a_{m}(\Pi)\to0$, \begin{equation}
\rho^{\prime\prime}\left(\bs{\lambda}^{\trans}\bs A_{i}\right)\leq\rho^{\prime\prime}(0)/2<0,\ \forall\bs{\lambda}\in\Upsilon(t),\ \forall i\in I,\ \forall\pi\in\Pi\label{eq:second derivative is bounded-1}
\end{equation}
w.p. at least $1-c_{1}\exp\left(-c_{2}mt^{2}\right)$
for all $0<t<c_{5}$ and large $m$.

By Taylor's expansion, for any $\bs{\lambda}\in\partial\Upsilon(t)$,
we have, for some $\widetilde{\bs{\lambda}}$ on the line segment between $\bs 0$ and $\bs{\lambda}$,
\begin{align*}
  & \hat{G}(\bs{\lambda};\pi)-\hat{G}(\bs0;\pi)=\frac{\partial\hat{G}(\bs0;\pi)}{\partial\bs{\lambda}^{\trans}}\bs{\lambda}+\frac{1}{2}\bs{\lambda}^{\trans}\frac{\partial^{2}\hat{G}(\widetilde{\bs{\lambda}};\pi)}{\partial\bs{\lambda}\partial\bs{\lambda}^{\trans}}\bs{\lambda}\leq\left\Vert \frac{\partial\hat{G}(\bs0;\pi)}{\partial\bs{\lambda}}\right\Vert \left\Vert \bs{\lambda}\right\Vert \\
  & \quad +\frac{1}{2}\frac{1}{m}\sum_{i\in I}\left\{ \rho^{\prime\prime}(\widetilde{\bs{\lambda}}^{\trans}\bs A_{i})\bs{\lambda}^{\trans}\bs A_{i}\bs A_{i}^{\trans}\bs{\lambda}\right\} \leq C_{\Upsilon}^{-1}\left\Vert \bs{\lambda}\right\Vert ^{2}+\frac{\rho^{\prime\prime}(0)}{4}\frac{1}{m}\sum_{i\in I}\left\{ \bs{\lambda}^{\trans}\bs A_{i}\bs A_{i}^{\trans}\bs{\lambda}\right\} \\
  & \leq C_{\Upsilon}^{-1}\left\Vert \bs{\lambda}\right\Vert ^{2}+\frac{\rho^{\prime\prime}(0)}{4}\left\Vert \bs{\lambda}\right\Vert ^{2}\inf_{\pi\in\Pi}\lambda_{\min}\left\{ \frac{1}{m}\sum_{i\in I}\bs A_{i}\bs A_{i}^{\trans}\right\} \leq\left\Vert \bs{\lambda}\right\Vert ^{2}\left\{ C_{\Upsilon}^{-1}+\frac{\rho^{\prime\prime}(0)c_{\xi}}{8}\right\}
 \end{align*}
for all $\pi\in\Pi$ w.p. at least $1-4c_{1}\exp\left(-c_{2}mt^{2}\right)$
for all $0<t<\min\{c_{3},c_{4},c_{5}\}$, where the second inequality
follows from (\ref{eq:EL step bound for derivative}) and (\ref{eq:second derivative is bounded-1})
and the last one follows from (\ref{eq:EL step minimum singular value}).
If we take $C_{\Upsilon}>16/(c_{\xi}|\rho^{\prime\prime}(0)|)$, then
w.p. at least $1-4c_{1}\exp\left(-c_{2}mt^{2}\right)$:
$\hat{G}(\bs{\lambda};\pi)<\hat{G}(\bs 0;\pi)$ for all
$\bs{\lambda}\in\partial\Upsilon(t)$ and $\pi\in\Pi$.
For any $\pi\in\Pi$, since $\hat{G}(\bs{\lambda};\pi)$ is continuous
with respect to $\bs{\lambda}$, there exists a local maximum of $\hat{G}(\bs{\lambda};\pi)$
in the interior of $\Upsilon(t)$. Since $\hat{G}(\bs{\lambda};\pi)$
is also strictly concave and has a unique global maximum point $\hat{\bs{\lambda}}_{I}(\pi)$,
we conclude that $\hat{\bs{\lambda}}_{I}(\pi)\in\Upsilon(t)\backslash\partial\Upsilon(t)$,
which leads to $\left\Vert \hat{\bs{\lambda}}_{I}(\pi)\right\Vert \leq C_{\Upsilon}(a_{m}(\Pi)+t)$.
Therefore,
$P(\sup_{\pi\in\Pi}\Vert \hat{\bs{\lambda}}_{I}(\pi)\Vert \leq C_{\Upsilon}(a_{m}(\Pi)+t))\geq1-4c_{1}\exp(-c_{2}mt^{2})$
for all $0<t<\min\{c_{3},c_{4},c_{5}\}$ and large $m$. Setting $t=\delta$ completes the proof.
\end{proof}

\subsection[Asymptotic properties of beta and welfare estimators]{Asymptotic properties of $\widehat{\protect\bs{\beta}}_{I}(\pi)$
and $\widehat{W}_{I}(\pi)$}
\begin{lem}
[Convergence rate of $\widehat{\bs{\beta}}_{I}(\pi)$ with unknown propensity score]\label{lem:Convergence rate of beta_H-unknown propensity} Suppose
that Assumptions \ref{assu:unconfounded}, \ref{assu:Overlap assumption and smoothness},
\ref{assu:regularity-conditions-on-L},
\ref{assu:regularity-assumptions-on-U},
\ref{assu:smoothness for mu}, \ref{assu:regularity on mu}, and
\ref{assu:regularity-assumptions-on-nuisance}
hold and $I\subset\left\{ 1,\ldots,N\right\} $ is an index set with
$\left|I\right|=m$. For any policy class $\Pi\subset\Pi_{\infty}$
with VC dimension $\mathrm{VC}(\Pi)\geq1$, if $\mathrm{VC}(\Pi)/m\to0$
 as $m\to\infty$, then the $\widehat{\bs{\beta}}_{I}(\pi)$ defined
by (\ref{eq:def-of-beta-hat-pi-I-unknown-propensity})
satisfies
  \[
P\left(\sup_{\pi\in\Pi}\left\Vert \widehat{\bs{\beta}}_{I}(\pi)-\bs{\beta}^{*}(\pi)\right\Vert \geq C\sqrt{\frac{\mathrm{VC}(\Pi)}{m}}+\delta\right)\leq c_{1}\exp\left(-c_{2}m\delta^{2}\right)
\]
 for all $0<\delta<c_{3}$ and large $m$, where $c_{1},c_{2},c_{3}>0$
are finite constants independent of $\delta$, $m$, $I$
and $\Pi$, and $C>0$ is finite.
\end{lem}
\begin{proof}
Without loss of generality, we assume $I=\{1,\ldots,m\}$. In this
proof, we write $\mathcal{F}_{\mathrm{DNN},e}$ and $\mathcal{F}_{\mathrm{DNN},\mu}$
for $\text{logistic}\circ\left\{ \mathcal{F}_{\mathrm{DNN}}(\mathcal{H}_{e},\mathcal{D}_{e})\bigcap\{f:\left\Vert f\right\Vert _{\infty}\le M\}\right\} $
and $\mathcal{F}_{\mathrm{DNN}}(\mathcal{H}_{\mu},\mathcal{D}_{\mu})\bigcap\{f:\left\Vert f\right\Vert _{\infty}\le M\}$,
respectively. For notational simplicity, we present the argument for one
coordinate. Applying the same argument to $j=1,\ldots,p$ and taking a
finite union bound yields the displayed vector-norm result, since $p$ is
fixed; the constants below absorb $p$. Let
\begin{align*}
Q(\beta;\pi)&=\e\left[\left\{ \frac{\pi(\bs X)T}{e^{*}(\bs X)}+\frac{\left(1-\pi(\bs X)\right)\left(1-T\right)}{1-e^{*}(\bs X)}\right\} \mathcal{L}_{1}(Y-\beta)\right],\\
Q_{m}(\beta;\pi)&=\frac{1}{m}\sum_{i\in I}\widehat{w}_{I,i}(\pi)\Bigl\{\frac{\pi(\bs X_{i})T_{i}}{\widehat{e}_{I}(\bs X_{i})}+\frac{(1-\pi(\bs X_{i}))(1-T_{i})}{1-\widehat{e}_{I}(\bs X_{i})}\Bigr\}\mathcal{L}_{1}(Y_{i}-\beta).
\end{align*}
Then $\widehat{\beta}_{I}(\pi)=\underset{\beta\in\R}{\arg\min}\ Q_{m}(\beta;\pi)$
and $\beta^{*}(\pi)=\underset{\beta\in\R}{\arg\min}\ Q(\beta;\pi)$.
Since $Q_{m}(\beta;\pi)$ is convex in $\beta$, we apply Lemma \ref{lem:Nearness of the argmins of the convex functions}
to bound $\widehat{\beta}_{I}(\pi)-\beta^{*}(\pi)$. We decompose
the proof into two steps and write $v_{e,m}:=\mathcal{H}_{e}^{2}\mathcal{D}_{e}^{2}\log(\mathcal{H}_{e}\mathcal{D}_{e})$, $v_{\mu,m}:=\mathcal{H}_{\mu}^{2}\mathcal{D}_{\mu}^{2}\log(\mathcal{H}_{\mu}\mathcal{D}_{\mu})$, and $a_{m}(\Pi):=\sqrt{\mathrm{VC}(\Pi)/m}+\rho_{e,m}+\rho_{\mu,m}$.
By Assumption~\ref{assu:Overlap assumption and smoothness} and Assumption~\ref{assu:regularity-assumptions-on-nuisance}(ii), we have $\sqrt{v_{e,m}/m}=o(m^{-1/4})$ and $\sqrt{v_{\mu,m}/m}=o(m^{-1/4})$. Enlarging $\rho_{e,m}$ and $\rho_{\mu,m}$ if necessary, we may assume $\sqrt{v_{e,m}/m}\leq\rho_{e,m}$ and $\sqrt{v_{\mu,m}/m}\leq\rho_{\mu,m}$ for all large $m$.

\textbf{Step 1. (Developing a lower bound for $h(\delta;\pi)$).}
For all $\delta\in(0,\min\{(3\underline{Q}^{\prime\prime})/(4Q_{lip}^{\prime\prime}),c_{0}\})$, let $h(\delta;\pi)=\inf_{|\beta-\beta^{*}(\pi)|=\delta}Q(\beta;\pi)-Q(\beta^{*}(\pi);\pi)$.
By the first-order condition for $\beta^{*}(\pi)$, the mean value theorem, and Assumption~\ref{assu:regularity-conditions-on-L}(ii), we have
\begin{equation}
\inf_{\pi\in\Pi_{\infty}}h(\delta;\pi)\geq\frac{\underline{Q}^{\prime\prime}}{4}\delta^{2}\ \text{ for all }0<\delta\leq\min\left\{\frac{3\underline{Q}^{\prime\prime}}{4Q_{lip}^{\prime\prime}},c_{0}\right\}.\label{eq:debiased-beta-lower-bound}
\end{equation}

\textbf{Step 2. (Developing an upper bound for $\Delta(\delta;\pi)$).}
For any $\delta\in(0,\min\{(3\underline{Q}^{\prime\prime})/(4Q_{lip}^{\prime\prime}),c_{0}\})$,
we let
\[
\Delta(\delta;\pi):=\sup_{\left|\beta-\beta^{*}(\pi)\right|=\delta}\left|Q_{m}(\beta;\pi)-Q_{m}(\beta^{*}(\pi);\pi)-\left\{ Q(\beta;\pi)-Q(\beta^{*}(\pi);\pi)\right\} \right|.
\]
Let
\begin{align*}
Q_{m}^{\prime}(\beta;\pi)&:=\frac{1}{m}\sum_{i\in I}\widehat{w}_{I,i}(\pi)\Bigl\{\frac{\pi(\bs X_{i})T_{i}}{\widehat{e}_{I}(\bs X_{i})}+\frac{(1-\pi(\bs X_{i}))(1-T_{i})}{1-\widehat{e}_{I}(\bs X_{i})}\Bigr\}\mathcal{L}_{1}^{\prime}(Y_{i}-\beta),\\
Q^{\prime}(\beta;\pi)&:=\e\left[\left\{ \frac{\pi(\bs X)T}{e^{*}(\bs X)}+\frac{\left(1-\pi(\bs X)\right)\left(1-T\right)}{1-e^{*}(\bs X)}\right\} \mathcal{L}_{1}^{\prime}(Y-\beta)\right].
\end{align*}
By Assumption \ref{assu:regularity-conditions-on-L},
we have $-\int_{\beta^{*}(\pi)}^{\beta}Q_{m}^{\prime}(\widetilde{\beta};\pi)d\widetilde{\beta}=Q_{m}(\beta;\pi)-Q_{m}(\beta^{*}(\pi);\pi)$, and analogously for $Q^{\prime}$,
and thus
\begin{align}
\Delta(\delta;\pi)&=\sup_{\left|\beta-\beta^{*}(\pi)\right|=\delta}\left|-\int_{\beta^{*}(\pi)}^{\beta}\left\{ Q_{m}^{\prime}(\widetilde{\beta};\pi)-Q^{\prime}(\widetilde{\beta};\pi)\right\} d\widetilde{\beta}\right|\nonumber\\
&\leq\delta\sup_{\left|\beta-\beta^{*}(\pi)\right|\leq\delta}\left|Q_{m}^{\prime}(\beta;\pi)-Q^{\prime}(\beta;\pi)\right|=\delta\widetilde{\Delta}(\delta;\pi),\label{eq:proof-of-beta-bound-for-Delta-1}
\end{align}
where we have let $\widetilde{\Delta}(\delta;\pi):=\sup_{\left|\beta-\beta^{*}(\pi)\right|\leq\delta}\left|Q_{m}^{\prime}(\beta;\pi)-Q^{\prime}(\beta;\pi)\right|$.
Now, it suffices to bound $\widetilde{\Delta}(\delta;\pi)$. We let
\begin{align*}
 & \widetilde{Q}_{m}^{\prime}(\beta;e,\mu,\pi):=\frac{1}{m}\sum_{i\in I}\left\{ \frac{\pi(\bs X_{i})T_{i}}{e(\bs X_{i})}\left\{ \mathcal{L}_{1}^{\prime}(Y_{i}-\beta)-\mu_{11}(\bs X_{i})\right\}+\pi(\bs X_{i})\mu_{11}(\bs X_{i})\right\}\\
 & \quad+\frac{1}{m}\sum_{i\in I}\Biggl\{ \frac{(1-\pi(\bs X_{i}))(1-T_{i})}{1-e(\bs X_{i})}\left\{ \mathcal{L}_{1}^{\prime}(Y_{i}-\beta)-\mu_{10}(\bs X_{i})\right\} +(1-\pi(\bs X_{i}))\mu_{10}(\bs X_{i})\Biggr\},\\
 & Q^{\prime}(\beta;e,\mu,\pi)=\e\left[\frac{\pi(\bs X)T}{e(\bs X)}\left\{ \mathcal{L}_{1}^{\prime}(Y-\beta)-\mu_{11}(\bs X)\right\}\right]+\e\left[\pi(\bs X)\mu_{11}(\bs X)\right]\\
 & \quad+\e\left[\frac{(1-\pi(\bs X))(1-T)}{1-e(\bs X)}\left\{ \mathcal{L}_{1}^{\prime}(Y-\beta)-\mu_{10}(\bs X)\right\}\right]+\e\left[(1-\pi(\bs X))\mu_{10}(\bs X)\right].
\end{align*}
Write $\widehat{e}\assign\widehat{e}_{I}$. In the following decomposition, $\widehat{\mu}(\bs X)$ denotes the stacked vector formed from $\widehat{\mu}_{I,10}(\bs X;\widehat{\beta}_{I}^{\mathrm{init}}(\pi))$ and $\widehat{\mu}_{I,11}(\bs X;\widehat{\beta}_{I}^{\mathrm{init}}(\pi))$, and $\mu^{*}(\bs X)$ is defined analogously with $\mu_{10}^{*}(\bs X;\beta^{*}(\pi))$ and $\mu_{11}^{*}(\bs X;\beta^{*}(\pi))$. Then $\widetilde{\Delta}(\delta;\pi)$ can be decomposed into
\begin{align*}
& \sup_{\pi\in\Pi}\widetilde{\Delta}(\delta;\pi)\leq\sup_{\pi\in\Pi}\sup_{\left|\beta-\beta^{*}(\pi)\right|\leq\delta}\left|Q_{m}^{\prime}(\beta;\pi)-\widetilde{Q}_{m}^{\prime}(\beta;\widehat{e},\widehat{\mu},\pi)\right|\\
&\quad+\sup_{\pi\in\Pi}\sup_{\left|\beta-\beta^{*}(\pi)\right|\leq\delta}\left|\widetilde{Q}_{m}^{\prime}(\beta;\widehat{e},\widehat{\mu},\pi)-Q^{\prime}(\beta;\widehat{e},\widehat{\mu},\pi)\right|\\
&\quad+\sup_{\pi\in\Pi}\sup_{\left|\beta-\beta^{*}(\pi)\right|\leq\delta}\left|Q^{\prime}(\beta;\widehat{e},\widehat{\mu},\pi)-Q^{\prime}(\beta;e^{*},\mu^{*},\pi)\right|,
\end{align*}
where we have used the fact that $Q^{\prime}(\beta;e^{*},\mu^{*},\pi)=Q^{\prime}(\beta;\pi)$.
We bound these three terms one by one.

\textbf{Step 2.1: Bound $\sup_{\pi\in\Pi}\sup_{\left|\beta-\beta^{*}(\pi)\right|\leq\delta}\left|Q_{m}^{\prime}(\beta;\pi)-\widetilde{Q}_{m}^{\prime}(\beta;\widehat{e},\widehat{\mu},\pi)\right|$.}
Note that by the definition of $\widehat{w}_{I,i}(\pi)$ we have
\begin{align*}
 & Q_{m}^{\prime}(\beta;\pi)-\widetilde{Q}_{m}^{\prime}(\beta;\widehat{e},\widehat{\mu},\pi)\\
= & \frac{1}{m}\sum_{i\in I}\widehat{w}_{I,i}(\pi)\Bigl\{\frac{\pi(\bs X_{i})T_{i}}{\widehat{e}(\bs X_{i})}+\frac{(1-\pi(\bs X_{i}))(1-T_{i})}{1-\widehat{e}(\bs X_{i})}\Bigr\}\mathcal{L}_{1}^{\prime}(Y_{i}-\beta)\\
 & \qquad-\frac{1}{m}\sum_{i\in I}\left\{ \frac{\pi(\bs X_{i})T_{i}}{\widehat{e}(\bs X_{i})}\left\{ \mathcal{L}_{1}^{\prime}(Y_{i}-\beta)-\widehat{\mu}_{11}(\bs X_{i})\right\} +\pi(\bs X_{i})\widehat{\mu}_{11}(\bs X_{i})\right\} \\
 & \qquad-\frac{1}{m}\sum_{i\in I}\left\{ \frac{(1-\pi(\bs X_{i}))(1-T_{i})}{1-\widehat{e}(\bs X_{i})}\left\{ \mathcal{L}_{1}^{\prime}(Y_{i}-\beta)-\widehat{\mu}_{10}(\bs X_{i})\right\} +(1-\pi(\bs X_{i}))\widehat{\mu}_{10}(\bs X_{i})\right\} \\
= & \underbrace{\frac{1}{m}\sum_{i\in I}\left\{ \widehat{w}_{I,i}(\pi)-1\right\} \left\{ \frac{\pi(\bs X_{i})T_{i}}{\widehat{e}(\bs X_{i})}\left\{ \mathcal{L}_{1}^{\prime}(Y_{i}-\beta)-\widehat{\mu}_{11}(\bs X_{i})\right\} \right\} }_{\mathcal{Q}_{1}}\\
 & \qquad+\underbrace{\frac{1}{m}\sum_{i\in I}\left\{ \widehat{w}_{I,i}(\pi)-1\right\} \left\{ \frac{(1-\pi(\bs X_{i}))(1-T_{i})}{1-\widehat{e}(\bs X_{i})}\left\{ \mathcal{L}_{1}^{\prime}(Y_{i}-\beta)-\widehat{\mu}_{10}(\bs X_{i})\right\} \right\} }_{\mathcal{Q}_{2}}.
\end{align*}
By (\ref{eq:expression for wi}) and Taylor's expansion, we have
\begin{align}
\left\Vert \mathcal{Q}_{1}\right\Vert  & =\left\Vert \frac{1}{m}\sum_{i\in I}\left\{ \rho^{\prime}(\hat{\bs{\lambda}}_{I}(\pi)^{\trans}\bs A_{i})-1\right\} \left\{ \frac{\pi(\bs X_{i})T_{i}}{\widehat{e}(\bs X_{i})}\left\{ \mathcal{L}_{1}^{\prime}(Y_{i}-\beta)-\widehat{\mu}_{11}(\bs X_{i})\right\} \right\} \right\Vert \nonumber \\
 & =\left\Vert \hat{\bs{\lambda}}_{I}(\pi)^{\trans}\frac{1}{m}\sum_{i\in I}\rho^{\prime\prime}(\widetilde{\bs{\lambda}}_{i}(\pi)^{\trans}\bs A_{i})\bs A_{i}\left\{ \frac{\pi(\bs X_{i})T_{i}}{\widehat{e}(\bs X_{i})}\left\{ \mathcal{L}_{1}^{\prime}(Y_{i}-\beta)-\widehat{\mu}_{11}(\bs X_{i})\right\} \right\} \right\Vert \nonumber\\
 & \leq\sup_{\pi\in\Pi}\left\Vert \hat{\bs{\lambda}}_{I}(\pi)\right\Vert \sup_{f\in\mathcal{M}_{3}}\left|P_{m}f(\bs Z)\right|,\label{eq:Q1 decomposition}
\end{align}
w.p. at least $1-2c_{1}\exp(-c_{2}mt^{2})$ for all $0<t<c_{3}$
and large $m$, where the last step follows from
Lemma~\ref{lem:convergence rate for lambda}, Lemma~\ref{lem:converge of regression},
and the membership of the nuisance estimators in the corresponding bounded
DNN classes, and $\widetilde{\bs{\lambda}}_{i}(\pi)$ lies on the line segment between $\bs 0$ and $\hat{\bs{\lambda}}_{I}(\pi)$ for each $i$, where
\begin{align*}
  & \mathcal{M}_{3}:=\biggl\{\bs Z\mapsto\rho^{\prime\prime}(\bs{\lambda}^{\trans}\bs A(T,\bs X;e,\bs{\xi},\pi))\bs{\alpha}^{\trans}\bs A(T,\bs X;e,\bs{\xi},\pi)\\
  & \left.\ \times\left\{ \frac{\pi(\bs X)T}{e(\bs X)}\left\{ \mathcal{L}_{1}^{\prime}(Y-\beta)-\mu_{11}(\bs X)\right\} \right\} :\left\Vert \bs{\lambda}\right\Vert \leq Ca_{m}(\Pi)+t,\bs{\alpha}\in\mathbb{S}^{p},\right.\\
  & \left.\ e\in\mathcal{F}_{\mathrm{DNN},e},\bs{\xi}\in\mathcal{F}_{\mathrm{DNN},\mu}^{2p+2},\mu\in\mathcal{F}_{\mathrm{DNN},\mu}^{2},\left|\beta-\beta^{*}(\pi)\right|\leq\delta,\left\Vert \mu-\mu^{*}\right\Vert _{P,2}\leq C_{\mu}\rho_{\mu,m}+t,\pi\in\Pi\biggr\}\right.
 \end{align*}
is a function class with envelope $C$, and $C_{\mu}>0$ is a constant.
By the bounded difference inequality, we have
\begin{equation}
P\left(\sup_{f\in\mathcal{M}_{3}}\left|\left\{ P_{m}-P\right\} f(\bs Z)\right|-\e\left[\sup_{f\in\mathcal{M}_{3}}\left|\left\{ P_{m}-P\right\} f(\bs Z)\right|\right]>t\right)\leq c_{1}\exp(-c_{2}mt^{2})\label{eq:bounded diff ineq for M3}
\end{equation}
for all $t\geq0$. For the entropy bound, restrict to $0<t<c_{3}$ and large $m$. In this range, for the entropy objective $D(w)=w\log w-w$, $\rho^{\prime\prime}(v)=-\exp(-v)$ is Lipschitz on the bounded range of $\bs{\lambda}^{\trans}\bs A$ considered here. Lemma~\ref{lem:transformed uniform entropy}, the fixed dimensions of $\bs{\lambda}$ and $\bs{\alpha}$, the monotonicity of $\mathcal{L}_{1}^{\prime}$, and the entropy bounds for $\Pi$, $\mathcal{F}_{\mathrm{DNN},e}$, and $\mathcal{F}_{\mathrm{DNN},\mu}$ yield
\begin{equation}
\sup_{Q}\log N\left(C\epsilon,\mathcal{M}_{3},\left\Vert \cdot\right\Vert _{Q,2}\right)\leq C\left\{ \mathrm{VC}(\Pi)+v_{e,m}+v_{\mu,m}\right\}\log\left(a/\epsilon\right)\label{eq:entropy for M3}
\end{equation}
for all $0<\epsilon<1$, where $a>1$ is a constant. Then it
follows from \citet[Lemma 6.2]{chernozhukov2018Double} that
  \[
\e\left[\sup_{f\in\mathcal{M}_{3}}\left|\left\{ P_{m}-P\right\} f(\bs Z)\right|\right]\leq C\sqrt{\frac{\mathrm{VC}(\Pi)}{m}}+C\sqrt{\frac{v_{e,m}}{m}}+C\sqrt{\frac{v_{\mu,m}}{m}}.
\]
This, combined with (\ref{eq:bounded diff ineq for M3}) gives that
  \[
P\left(\sup_{f\in\mathcal{M}_{3}}\left|\left\{ P_{m}-P\right\} f(\bs Z)\right|\geq C\sqrt{\frac{\mathrm{VC}(\Pi)}{m}}+C\sqrt{\frac{v_{e,m}}{m}}+C\sqrt{\frac{v_{\mu,m}}{m}}+t\right)\leq c_{1}\exp(-c_{2}mt^{2})
\]
for all $t\geq0$. Furthermore, by the boundedness of the functions
in $\mathcal{F}_{\mathrm{DNN},\mu}$ and $\mathcal{F}_{\mathrm{DNN},e}$
and (\ref{eq:second derivative is bounded-1}) we have
\begin{align*}
  &\sup_{f\in\mathcal{M}_{3}}\left|Pf(\bs Z)\right|  \lesssim\sup_{\pi\in\Pi}\sup_{\left|\beta-\beta^{*}(\pi)\right|\leq\delta}\sup_{\left\Vert \mu-\mu^{*}\right\Vert _{P,2}\leq C_{\mu}\rho_{\mu,m}+t}\left|\e\left[\mu_{11}^{*}(\bs X;\beta)-\mu_{11}(\bs X)\right]\right|\\
 & \leq\sup_{\pi\in\Pi}\sup_{\left|\beta-\beta^{*}(\pi)\right|\leq\delta}\sup_{\left\Vert \mu-\mu^{*}\right\Vert _{P,2}\leq C_{\mu}\rho_{\mu,m}+t}\left\Vert \mu_{11}^{*}(\bs X;\beta)-\mu_{11}(\bs X)\right\Vert _{P,2}\\
 & \lesssim\sup_{\pi\in\Pi}\sup_{\left|\beta-\beta^{*}(\pi)\right|\leq\delta}\left\Vert \mu_{11}^{*}(\bs X;\beta)-\mu_{11}^{*}(\bs X;\beta^{*}(\pi))\right\Vert _{P,2}+C_{\mu}\rho_{\mu,m}+t\leq L_{\mu}\delta+C_{\mu}\rho_{\mu,m}+t,
\end{align*}
where the last step follows from Assumption~\ref{assu:regularity on mu}.
Therefore, combining the above two displays gives that
\(\sup_{f\in\mathcal{M}_{3}}\left|P_{m}f(\bs Z)\right|\leq C\sqrt{\mathrm{VC}(\Pi)/m}+C\rho_{\mu,m}+C\delta+Ct\)
w.p. at least $1-c_{1}\exp(-c_{2}mt^{2})$ for all $0<t<c_{3}$
and large $m$. Recalling (\ref{eq:Q1 decomposition}),
it follows from Lemma \ref{lem:convergence rate for lambda} and the same argument that
  \begin{align*}
\sup_{\pi\in\Pi}\sup_{\left|\beta-\beta^{*}(\pi)\right|\leq\delta}\left\Vert \mathcal{Q}_{1}\right\Vert &\leq C\left\{ a_{m}(\Pi)+t\right\} \left\{ \delta+a_{m}(\Pi)+t\right\},\\
\sup_{\pi\in\Pi}\sup_{\left|\beta-\beta^{*}(\pi)\right|\leq\delta}\left\Vert \mathcal{Q}_{2}\right\Vert &\leq C\left\{ a_{m}(\Pi)+t\right\} \left\{ \delta+a_{m}(\Pi)+t\right\},
\end{align*}
each w.p. at least $1-3c_{1}\exp(-c_{2}mt^{2})$ for all $0<t<c_{3}$
and large $m$. As a result,
  \begin{align*}
\sup_{\pi\in\Pi}\sup_{\left|\beta-\beta^{*}(\pi)\right|\leq\delta}\left|Q_{m}^{\prime}(\beta;\pi)-\widetilde{Q}_{m}^{\prime}(\beta;\widehat{e},\widehat{\mu},\pi)\right| & \leq\sup_{\pi\in\Pi}\sup_{\left|\beta-\beta^{*}(\pi)\right|\leq\delta}\left\Vert \mathcal{Q}_{1}\right\Vert +\sup_{\pi\in\Pi}\sup_{\left|\beta-\beta^{*}(\pi)\right|\leq\delta}\left\Vert \mathcal{Q}_{2}\right\Vert \\
 & \leq C\left\{ a_{m}(\Pi)+t\right\} \left\{ \delta+a_{m}(\Pi)+t\right\}
\end{align*}
w.p. at least $1-6c_{1}\exp(-c_{2}mt^{2})$ for all $0<t<c_{3}$
and large $m$.

\textbf{Step 2.2: Bound $\sup_{\pi\in\Pi}\sup_{\left|\beta-\beta^{*}(\pi)\right|\leq\delta}\left|\widetilde{Q}_{m}^{\prime}(\beta;\widehat{e},\widehat{\mu},\pi)-Q^{\prime}(\beta;\widehat{e},\widehat{\mu},\pi)\right|$.}
We have the following decomposition:
\begin{align*}
 & \sup_{\pi\in\Pi}\sup_{\left|\beta-\beta^{*}(\pi)\right|\leq\delta}\left|\widetilde{Q}_{m}^{\prime}(\beta;\widehat{e},\widehat{\mu},\pi)-Q^{\prime}(\beta;\widehat{e},\widehat{\mu},\pi)\right|\\
\leq & \underbrace{\sup_{\pi\in\Pi}\sup_{\left|\beta-\beta^{*}(\pi)\right|\leq\delta}\left|\widetilde{Q}_{m}^{\prime}(\beta;\widehat{e},\widehat{\mu},\pi)-Q^{\prime}(\beta;\widehat{e},\widehat{\mu},\pi)-\left\{ \widetilde{Q}_{m}^{\prime}(\beta;e^{*},\mu^{*},\pi)-Q^{\prime}(\beta;e^{*},\mu^{*},\pi)\right\} \right|}_{\mathcal{Q}_{3}}\\
 & \qquad+\underbrace{\sup_{\pi\in\Pi}\sup_{\left|\beta-\beta^{*}(\pi)\right|\leq\delta}\left|\widetilde{Q}_{m}^{\prime}(\beta;e^{*},\mu^{*},\pi)-Q^{\prime}(\beta;e^{*},\mu^{*},\pi)\right|}_{\mathcal{Q}_{4}}.
\end{align*}
We analyze $\mathcal{Q}_{3}$ and $\mathcal{Q}_{4}$ one by one. By
Lemma~\ref{lem:converge of propensity}, Lemma~\ref{lem:converge of regression},
and the membership of the nuisance estimators in the corresponding bounded
DNN classes,
we have
\begin{equation}
\mathcal{Q}_{3}\leq\sup_{f\in\mathcal{M}_{4}(t)}\left|\left\{ P_{m}-P\right\} f(\bs Z)\right|,\label{eq:Q3-M4-bound}
\end{equation}
w.p. at least $1-2c_{1}\exp(-c_{2}mt^{2})$ for all $0<t<c_{3}$
and large $m$, where
\begin{align*}
\mathcal{M}_{4}(t) & :=\biggl\{\bs Z\mapsto\psi(\bs Z;\beta,e,\mu,\pi)-\psi(\bs Z;\beta,e^{*},\mu^{*},\pi):\left|\beta-\beta^{*}(\pi)\right|\leq\delta,e\in\mathcal{F}_{\mathrm{DNN},e},\\
 & \qquad\mu\in\mathcal{F}_{\mathrm{DNN},\mu}^{2},\left\Vert e-e^{*}\right\Vert _{P,2}\leq\rho_{e,m}+t,\left\Vert \mu-\mu^{*}\right\Vert _{P,2}\leq C_{\mu}\rho_{\mu,m}+t,\pi\in\Pi\biggr\}
\end{align*}
is a function class with envelope $C$ and $\psi(\bs Z;\beta,e,\mu,\pi)$
is defined as
\begin{align*}
\psi(\bs Z;\beta,e,\mu,\pi) & :=\frac{\pi(\bs X)T}{e(\bs X)}\left\{ \mathcal{L}_{1}^{\prime}(Y-\beta)-\mu_{11}(\bs X)\right\} +\pi(\bs X)\mu_{11}(\bs X)\\
 & \qquad+\frac{(1-\pi(\bs X))(1-T)}{1-e(\bs X)}\left\{ \mathcal{L}_{1}^{\prime}(Y-\beta)-\mu_{10}(\bs X)\right\} +(1-\pi(\bs X))\mu_{10}(\bs X).
\end{align*}
By the bounded difference inequality, we have
\begin{equation}
P\left(\sup_{f\in\mathcal{M}_{4}(t)}\left|\left\{ P_{m}-P\right\} f(\bs Z)\right|-\e\left[\sup_{f\in\mathcal{M}_{4}(t)}\left|\left\{ P_{m}-P\right\} f(\bs Z)\right|\right]\geq t\right)\leq c_{1}\exp(-c_{2}mt^{2})\label{eq:bounded diff ineq for M4}
\end{equation}
for all $0<t<c_{3}$ and large $m$. Applying Lemma~\ref{lem:transformed uniform entropy} to the two summands in $\psi(\cdot;\beta,e,\mu,\pi)-\psi(\cdot;\beta,e^{*},\mu^{*},\pi)$, using the monotonicity of $\mathcal{L}_{1}^{\prime}$ and noting that the local $L_{2}$ restrictions only form subclasses, gives, uniformly in $t$,
  \[
\sup_{Q}\log N\left(C\epsilon,\mathcal{M}_{4}(t),\left\Vert \cdot\right\Vert _{Q,2}\right)\leq C\left\{ \mathrm{VC}(\Pi)+v_{e,m}+v_{\mu,m}\right\}\log\left(a/\epsilon\right)
\]
for all $0<\epsilon<1$, where $a>1$ is a constant. Note that
by the boundedness of the functions in $\mathcal{F}_{\mathrm{DNN},e}$
and $\mathcal{F}_{\mathrm{DNN},\mu}$, \(\left|\psi(\bs Z;\beta,e,\mu,\pi)-\psi(\bs Z;\beta,e^{*},\mu^{*},\pi)\right|\lesssim\left\Vert \mu(\bs X)-\mu^{*}(\bs X)\right\Vert +\left|e(\bs X)-e^{*}(\bs X)\right|\) almost surely. Then we have
\begin{align*}
\sigma^{2} & :=\sup_{f\in\mathcal{M}_{4}(t)}\e\left[f(\bs Z)^{2}\right]\lesssim\sup_{\left\Vert \mu-\mu^{*}\right\Vert _{P,2}\leq C_{\mu}\rho_{\mu,m}+t}\e\left[\left\Vert \mu-\mu^{*}\right\Vert ^{2}\right]+\sup_{\left\Vert e-e^{*}\right\Vert _{P,2}\leq \rho_{e,m}+t}\e\left[\left|e-e^{*}\right|^{2}\right]\\
 & \lesssim\left(C_{\mu}\rho_{\mu,m}+t\right)^{2}+\left(\rho_{e,m}+t\right)^{2}\leq\left(C_{\mu}\rho_{\mu,m}+\rho_{e,m}+2t\right)^{2}.
\end{align*}
Applying \citet[Lemma 6.2]{chernozhukov2018Double} gives that
\begin{align*}
\e\left[\sup_{f\in\mathcal{M}_{4}(t)}\left|\left\{ P_{m}-P\right\} f(\bs Z)\right|\right]
&\leq C\left\{ \sqrt{\frac{\mathrm{VC}(\Pi)}{m}}+\sqrt{\frac{v_{e,m}}{m}}+\sqrt{\frac{v_{\mu,m}}{m}}\right\}\left\{ C_{\mu}\rho_{\mu,m}+\rho_{e,m}+2t\right\}\\
&\leq Ca_{m}(\Pi)^{2}+Ct^{2}.
\end{align*}
This, combined with (\ref{eq:Q3-M4-bound}) and (\ref{eq:bounded diff ineq for M4})
yields that
\(\mathcal{Q}_{3}\leq\sup_{f\in\mathcal{M}_{4}(t)}\left|\left\{ P_{m}-P\right\} f(\bs Z)\right|\leq Ca_{m}(\Pi)^{2}+Ct\)
w.p. at least $1-3c_{1}\exp(-c_{2}mt^{2})$ for all $0<t<c_{3}$
and large $m$.

Let $\mathcal{M}_{5}:=\{\bs Z\mapsto\psi(\bs Z;\beta,e^{*},\mu^{*},\pi):\left|\beta-\beta^{*}(\pi)\right|\leq\delta,\pi\in\Pi\}$
be a function class with envelope $C$. By the VC-subgraph property of $\Pi$, the monotonicity of $\mathcal{L}_{1}^{\prime}$, overlap, the one-dimensional variation of $\mu_{1t}^{*}(\cdot;\beta)$ controlled by Assumption~\ref{assu:regularity on mu}, and Lemma~\ref{lem:transformed uniform entropy}, \(\sup_{Q}\log N(C\epsilon,\mathcal{M}_{5},\left\Vert \cdot\right\Vert _{Q,2})\lesssim\mathrm{VC}(\Pi)\log(a/\epsilon)\) for all \(0<\epsilon<1\), where $a>1$ is a constant. Applying the bounded difference
inequality and \citet[Lemma 6.2]{chernozhukov2018Double} we have
  \[
\mathcal{Q}_{4}\leq\sup_{f\in\mathcal{M}_{5}}\left|\left\{ P_{m}-P\right\} f(\bs Z)\right|\leq\e\left[\sup_{f\in\mathcal{M}_{5}}\left|\left\{ P_{m}-P\right\} f(\bs Z)\right|\right]+t\leq C\sqrt{\frac{\mathrm{VC}(\Pi)}{m}}+t
\]
w.p. at least $1-c_{1}\exp(-c_{2}mt^{2})$ for all $0<t<c_{3}$
and large $m$.

Aggregating the results for $\mathcal{Q}_{3}$ and $\mathcal{Q}_{4}$
we have
  \[
\sup_{\pi\in\Pi}\sup_{\left|\beta-\beta^{*}(\pi)\right|\leq\delta}\left|\widetilde{Q}_{m}^{\prime}(\beta;\widehat{e},\widehat{\mu},\pi)-Q^{\prime}(\beta;\widehat{e},\widehat{\mu},\pi)\right|\leq C\sqrt{\frac{\mathrm{VC}(\Pi)}{m}}+Ca_{m}(\Pi)^{2}+Ct
\]
w.p. at least $1-4c_{1}\exp(-c_{2}mt^{2})$ for all $0<t<c_{3}$
and large $m$.

\textbf{Step 2.3: Bound $\sup_{\pi\in\Pi}\sup_{\left|\beta-\beta^{*}(\pi)\right|\leq\delta}\left|Q^{\prime}(\beta;\widehat{e},\widehat{\mu},\pi)-Q^{\prime}(\beta;e^{*},\mu^{*},\pi)\right|$.}
We let
\begin{align*}
\psi_{1}(\bs Z;\beta,e,\mu,\pi)&:=\frac{\pi(\bs X)T}{e(\bs X)}\left\{ \mathcal{L}_{1}^{\prime}(Y-\beta)-\mu_{11}(\bs X)\right\} +\pi(\bs X)\mu_{11}(\bs X),\\
\psi_{0}(\bs Z;\beta,e,\mu,\pi)&:=\frac{(1-\pi(\bs X))(1-T)}{1-e(\bs X)}\left\{ \mathcal{L}_{1}^{\prime}(Y-\beta)-\mu_{10}(\bs X)\right\} +(1-\pi(\bs X))\mu_{10}(\bs X).
\end{align*}
By Lemma \ref{lem:converge of propensity} and Lemma \ref{lem:converge of regression},
we have
\begin{align}
 & \sup_{\pi\in\Pi}\sup_{\left|\beta-\beta^{*}(\pi)\right|\leq\delta}\left|Q^{\prime}(\beta;\widehat{e},\widehat{\mu},\pi)-Q^{\prime}(\beta;e^{*},\mu^{*},\pi)\right|\nonumber\\
\leq & \underbrace{\sup_{\pi\in\Pi}\sup_{\left|\beta-\beta^{*}(\pi)\right|\leq\delta}\sup_{\left\Vert e-e^{*}\right\Vert _{P,2}\leq\rho_{e,m}+t}\sup_{\left\Vert \mu-\mu^{*}\right\Vert _{P,2}\leq C_{\mu}\rho_{\mu,m}+t}\left|\e\psi_{1}(\bs Z;\beta,e,\mu,\pi)-\e\psi_{1}(\bs Z;\beta,e^{*},\mu^{*},\pi)\right|}_{\mathcal{Q}_{5}}\nonumber \\
 & +\underbrace{\sup_{\pi\in\Pi}\sup_{\left|\beta-\beta^{*}(\pi)\right|\leq\delta}\sup_{\left\Vert e-e^{*}\right\Vert _{P,2}\leq\rho_{e,m}+t}\sup_{\left\Vert \mu-\mu^{*}\right\Vert _{P,2}\leq C_{\mu}\rho_{\mu,m}+t}\left|\e\psi_{0}(\bs Z;\beta,e,\mu,\pi)-\e\psi_{0}(\bs Z;\beta,e^{*},\mu^{*},\pi)\right|}_{\mathcal{Q}_{6}}\label{eq:Q5-Q6-decomposition}
\end{align}
w.p. at least $1-2c_{1}\exp(-c_{2}mt^{2})$ for all $0<t<c_{3}$
and large $m$.

We analyze $\mathcal{Q}_{5}$. We let $\Delta_{e}(\bs X)=e(\bs X)^{-1}-e^{*}(\bs X)^{-1}$ and $\Delta_{\mu}=\mu_{11}(\bs X)-\mu_{11}^{*}(\bs X)$, and
\begin{align*}
f(r) & =r\e\left[\pi(\bs X)e^{*}(\bs X)\Delta_{e}\mathcal{L}_{1}^{\prime}(Y^{*}(1)-\beta)\right]\\
 & \qquad-r\e\left[\pi(\bs X)e^{*}(\bs X)\left\{ \Delta_{e}\mu_{11}^{*}(\bs X)+\frac{\Delta_{\mu}}{e^{*}(\bs X)}+r\Delta_{e}\Delta_{\mu}\right\} \right]+r\e\left[\pi(\bs X)\Delta_{\mu}\right]
\end{align*}
for $r\in[0,1]$. Then
\(\e\psi_{1}(\bs Z;\beta,e,\mu,\pi)-\e\psi_{1}(\bs Z;\beta,e^{*},\mu^{*},\pi)=f(1)-f(0)=\int_{0}^{1}f^{\prime}(r)dr\), where
\begin{align*}
f^{\prime}(r) & =\e\left[\pi(\bs X)e^{*}(\bs X)\Delta_{e}\mathcal{L}_{1}^{\prime}(Y^{*}(1)-\beta)\right]-\e\left[\pi(\bs X)e^{*}(\bs X)\Delta_{e}\mu_{11}^{*}(\bs X)+\pi(\bs X)e^{*}(\bs X)\frac{\Delta_{\mu}}{e^{*}(\bs X)}\right]\\
 & \qquad-2\e\left[\pi(\bs X)e^{*}(\bs X)\Delta_{e}\Delta_{\mu}\right]r+\e\left[\pi(\bs X)\Delta_{\mu}\right]\\
 & =\e\left[\pi(\bs X)e^{*}(\bs X)\Delta_{e}\left\{ \mathcal{L}_{1}^{\prime}(Y^{*}(1)-\beta)-\mu_{11}^{*}(\bs X)\right\} \right]-2\e\left[\pi(\bs X)e^{*}(\bs X)\Delta_{e}\Delta_{\mu}\right]r\\
 & =\e\left[\pi(\bs X)e^{*}(\bs X)\Delta_{e}\left\{ \mu_{11}^{*}(\bs X;\beta)-\mu_{11}^{*}(\bs X;\beta^{*}(\pi))\right\} \right]-2\e\left[\pi(\bs X)e^{*}(\bs X)\Delta_{e}\Delta_{\mu}\right]r.
\end{align*}
Note that by Assumptions \ref{assu:Overlap assumption and smoothness},
\ref{assu:regularity-conditions-on-L} and \ref{assu:regularity on mu}
we have
\begin{align*}
\sup_{r\in[0,1]}\left|f^{\prime}(r)\right| & \lesssim\left|\e\left[\Delta_{e}\left\{ \mu_{11}^{*}(\bs X;\beta)-\mu_{11}^{*}(\bs X;\beta^{*}(\pi))\right\} \right]\right|+\left|\e\left[\Delta_{e}\Delta_{\mu}\right]\right|\\
 & \lesssim\left\Vert e-e^{*}\right\Vert _{P,2}\left\Vert \mu_{11}^{*}(\bs X;\beta)-\mu_{11}^{*}(\bs X;\beta^{*}(\pi))\right\Vert _{P,2}+\left\Vert e-e^{*}\right\Vert _{P,2}\left\Vert \mu-\mu^{*}\right\Vert _{P,2}\\
 & \lesssim\left\{ \rho_{e,m}+t\right\} L_{\mu}\delta+\left\{ \rho_{e,m}+t\right\} \left\{ C_{\mu}\rho_{\mu,m}+t\right\} \leq Ca_{m}(\Pi)\delta+Ca_{m}(\Pi)^{2}+Ct
\end{align*}
for all $0<t<c_{3}$ provided that $\left|\beta-\beta^{*}(\pi)\right|\leq\delta$,
$\left\Vert e-e^{*}\right\Vert _{P,2}\leq\rho_{e,m}+t$ and $\left\Vert \mu-\mu^{*}\right\Vert _{P,2}\leq C_{\mu}\rho_{\mu,m}+t$.
Therefore, we can obtain that
  \begin{align*}
\mathcal{Q}_{5} & \leq\sup_{\pi\in\Pi}\sup_{\left|\beta-\beta^{*}(\pi)\right|\leq\delta}\sup_{\left\Vert e-e^{*}\right\Vert _{P,2}\leq\rho_{e,m}+t}\sup_{\left\Vert \mu-\mu^{*}\right\Vert _{P,2}\leq C_{\mu}\rho_{\mu,m}+t}\int_{0}^{1}\left|f^{\prime}(r)\right|dr\leq Ca_{m}(\Pi)\delta+Ca_{m}(\Pi)^{2}+Ct
\end{align*}
for all $0<t<c_{3}$ and large $m$. Similarly, we can
also derive that $\mathcal{Q}_{6}\leq Ca_{m}(\Pi)\delta+Ca_{m}(\Pi)^{2}+Ct$
for all $0<t<c_{3}$ and large $m$. Recalling (\ref{eq:Q5-Q6-decomposition}),
we have
  \[
\sup_{\pi\in\Pi}\sup_{\left|\beta-\beta^{*}(\pi)\right|\leq\delta}\left|Q^{\prime}(\beta;\widehat{e},\widehat{\mu},\pi)-Q^{\prime}(\beta;e^{*},\mu^{*},\pi)\right|\leq Ca_{m}(\Pi)\delta+Ca_{m}(\Pi)^{2}+Ct
\]
w.p. at least $1-2c_{1}\exp(-c_{2}mt^{2})$ for all $0<t<c_{3}$
and large $m$.

\textbf{Step 2.4: Aggregating the results.} Combining the results
obtained in Steps 2.1--2.3 we have
  \[
\sup_{\pi\in\Pi}\widetilde{\Delta}(\delta;\pi)\leq C\sqrt{\frac{\mathrm{VC}(\Pi)}{m}}+Ca_{m}(\Pi)^{2}+Ca_{m}(\Pi)\delta+Ct
\]
w.p. at least $1-c_{1}\exp(-c_{2}mt^{2})$ for all $0<t<c_{3}$,
$0<\delta<\min\{(3\underline{Q}^{\prime\prime})/(4Q_{lip}^{\prime\prime}),c_{0}\}$
and large $m$. Because $\mathrm{VC}(\Pi)\geq1$ and $\mathrm{VC}(\Pi)/m\to0$, we have \(\mathrm{VC}(\Pi)/m\leq\sqrt{\mathrm{VC}(\Pi)/m}\) and \(m^{-1/2}\leq\sqrt{\mathrm{VC}(\Pi)/m}\).
Since $\rho_{e,m}=o(m^{-1/4})$ and $\rho_{\mu,m}=o(m^{-1/4})$, the second-order terms satisfy
\begin{align*}
  &\rho_{e,m}\sqrt{\mathrm{VC}(\Pi)/m}=o(\sqrt{\mathrm{VC}(\Pi)/m}),\quad \rho_{\mu,m}\sqrt{\mathrm{VC}(\Pi)/m}=o(\sqrt{\mathrm{VC}(\Pi)/m}),\\  &\rho_{e,m}\rho_{\mu,m}=o(m^{-1/2})=O(\sqrt{\mathrm{VC}(\Pi)/m}),\quad \rho_{e,m}^{2}+\rho_{\mu,m}^{2}=o(m^{-1/2})=O(\sqrt{\mathrm{VC}(\Pi)/m}).
\end{align*}
Hence $a_{m}(\Pi)^{2}\leq C\sqrt{\mathrm{VC}(\Pi)/m}$ for some constant $C>0$. Also, after reducing $c_{3}$ if necessary, we restrict to $0<\delta<c_{3}\le1$. Then the inequality $xy\leq \eta y+x^{2}/(4\eta)$ for $x,y\ge0$ and $y\le1$ gives
\(Ca_{m}(\Pi)\delta\leq\frac{\underline{Q}^{\prime\prime}}{8}\delta+C_{\eta}a_{m}(\Pi)^{2}\leq\frac{\underline{Q}^{\prime\prime}}{8}\delta+C\sqrt{\mathrm{VC}(\Pi)/m}\).
Therefore, for all large $m$ and all $0<t<c_{3}$,
$\sup_{\pi\in\Pi}\widetilde{\Delta}(\delta;\pi)\leq C_{0}\sqrt{\mathrm{VC}(\Pi)/m}+\underline{Q}^{\prime\prime}\delta/8+C_{1}t$.
Taking $t=\delta$ and choosing a constant $C_{2}>0$ large enough, we obtain that for
$\delta^{\prime}:=C_{2}\delta+C_{2}\sqrt{\mathrm{VC}(\Pi)/m}$, it holds w.p. at least $1-c_{1}\exp(-c_{2}m\delta^{2})$ that
$\sup_{\pi\in\Pi}\widetilde{\Delta}(\delta^{\prime};\pi)\leq\underline{Q}^{\prime\prime}\delta^{\prime}/4$.

Recall the definitions of $h(\delta;\pi)$ and $\Delta(\delta;\pi)$.
Applying Lemma \ref{lem:Nearness of the argmins of the convex functions},
we have
\[
\begin{aligned}
P\left(\sup_{\pi\in\Pi}\left|\widehat{\beta}_{I}(\pi)-\beta^{*}(\pi)\right|>\delta^{\prime}\right)&\leq P\left(\Delta(\delta^{\prime};\pi)\geq h(\delta^{\prime};\pi),\ \exists\pi\in\Pi\right)\\
&\leq P\left(\sup_{\pi\in\Pi}\widetilde{\Delta}(\delta^{\prime};\pi)\geq\frac{\underline{Q}^{\prime\prime}}{4}\delta^{\prime}\right)\leq c_{1}\exp(-c_{2}m\delta^{2})
\end{aligned}
\]
for all $0<\delta<c_{3}$ and large $m$,
where the second inequality follows from (\ref{eq:debiased-beta-lower-bound})
and (\ref{eq:proof-of-beta-bound-for-Delta-1}). Since $\delta^{\prime}\leq C\delta+C\sqrt{\mathrm{VC}(\Pi)/m}$,
the stated result follows.
 \end{proof}
\begin{lem}
[Convergence rate of $\widehat{W}_{I}(\pi)$ with unknown propensity score]\label{lem:Convergence rate of WH_hat-unknown propensity} Suppose
that Assumptions \ref{assu:unconfounded}, \ref{assu:Overlap assumption and smoothness},
\ref{assu:regularity-conditions-on-L},
\ref{assu:regularity-assumptions-on-U},
\ref{assu:smoothness for mu}, \ref{assu:regularity on mu}, and
\ref{assu:regularity-assumptions-on-nuisance}
hold and $I\subset\left\{ 1,\ldots,N\right\} $ is an index set with
$\left|I\right|=m$. For any policy class $\Pi\subset\Pi_{\infty}$
with VC dimension $\mathrm{VC}(\Pi)\geq1$, if $\mathrm{VC}(\Pi)/m\to0$
as $m\to\infty$, then $\widehat{W}_{I}(\pi)$ defined in (\ref{eq:def-of-W-hat-pi-I-unknown-propensity})
satisfies
\[
P\left(\sup_{\pi\in\Pi}\left|\widehat{W}_{I}(\pi)-W(\pi)\right|\geq\delta+C_{1}\sqrt{\frac{\mathrm{VC}(\Pi)}{m}}\right)\leq C_{2}\exp(-C_{3}m\delta^{2})
\]
for all $\delta\geq0$ and $m>C_{4}$, where $C_{1},\ldots,C_{4}>0$
are finite constants independent of $\delta$, $m$, $I$, and $\Pi$.
\end{lem}
\begin{proof}
Without loss of generality, we assume $I=\{1,\ldots,m\}$. For any
$\bs{\beta}\in\R^{p}$, we let
\begin{align*}
\psi_{1}(\bs Z;\bs{\beta},e,\mu,\pi)&:=\frac{\pi(\bs X)T}{e(\bs X)}\left\{ U(Y,\bs X,\bs{\beta})-\mu_{01}(\bs X)\right\} +\pi(\bs X)\mu_{01}(\bs X),\\
\psi_{0}(\bs Z;\bs{\beta},e,\mu,\pi)&:=\frac{(1-\pi(\bs X))(1-T)}{1-e(\bs X)}\left\{ U(Y,\bs X,\bs{\beta})-\mu_{00}(\bs X)\right\} +(1-\pi(\bs X))\mu_{00}(\bs X),\\
\widehat{\Psi}_{m}(\bs{\beta};\pi)&:=\frac{1}{m}\sum_{i\in I}\widehat{w}_{I,i}(\pi)\left\{ \frac{\pi(\bs X_{i})T_{i}}{\widehat{e}_{I}(\bs X_{i})}+\frac{(1-\pi(\bs X_{i}))(1-T_{i})}{1-\widehat{e}_{I}(\bs X_{i})}\right\} U(Y_{i},\bs X_{i},\bs{\beta}),\\
\widetilde{\Psi}_{m}(\bs{\beta};e,\mu,\pi)&=\frac{1}{m}\sum_{i\in I}\psi_{1}(\bs Z_{i};\bs{\beta},e,\mu,\pi)+\frac{1}{m}\sum_{i\in I}\psi_{0}(\bs Z_{i};\bs{\beta},e,\mu,\pi).
\end{align*}
Set $\Psi(\bs{\beta};e,\mu,\pi):=\e[\psi_{1}(\bs Z;\bs{\beta},e,\mu,\pi)+\psi_{0}(\bs Z;\bs{\beta},e,\mu,\pi)]$.
Let $a_{m}(\Pi):=\sqrt{\mathrm{VC}(\Pi)/m}+\rho_{e,m}+\rho_{\mu,m}$.
Write $\widehat{e}\assign\widehat{e}_{I}$. In this decomposition, $\widehat{\mu}(\bs X)$ denotes the stacked vector formed from $\widehat{\mu}_{I,00}(\bs X;\widehat{\bs{\beta}}_{I}^{\mathrm{init}}(\pi))$ and $\widehat{\mu}_{I,01}(\bs X;\widehat{\bs{\beta}}_{I}^{\mathrm{init}}(\pi))$, and $\mu^{*}(\bs X)$ is defined analogously with $\mu_{00}^{*}(\bs X;\bs{\beta}^{*}(\pi))$ and $\mu_{01}^{*}(\bs X;\bs{\beta}^{*}(\pi))$. By Lemma \ref{lem:Convergence rate of beta_H-unknown propensity}, we can decompose $\widehat{W}_{I}(\pi)-W(\pi)$ as
 \begin{align*}
 & \sup_{\pi\in\Pi}\left|\widehat{W}_{I}(\pi)-W(\pi)\right|=\sup_{\pi\in\Pi}\left|\widehat{\Psi}_{m}(\widehat{\bs{\beta}}_{I}(\pi);\pi)-\Psi(\bs{\beta}^{*}(\pi);e^{*},\mu^{*},\pi)\right|\\
\leq & \sup_{\pi\in\Pi}\left|\widehat{\Psi}_{m}(\widehat{\bs{\beta}}_{I}(\pi);\pi)-\widetilde{\Psi}_{m}(\widehat{\bs{\beta}}_{I}(\pi);\widehat{e},\widehat{\mu},\pi)\right|+\sup_{\pi\in\Pi}\left|\widetilde{\Psi}_{m}(\widehat{\bs{\beta}}_{I}(\pi);\widehat{e},\widehat{\mu},\pi)-\Psi(\widehat{\bs{\beta}}_{I}(\pi);\widehat{e},\widehat{\mu},\pi)\right|\\
 & +\sup_{\pi\in\Pi}\left|\Psi(\widehat{\bs{\beta}}_{I}(\pi);\widehat{e},\widehat{\mu},\pi)-\Psi(\widehat{\bs{\beta}}_{I}(\pi);e^{*},\mu^{*},\pi)\right|+\sup_{\pi\in\Pi}\left|\Psi(\widehat{\bs{\beta}}_{I}(\pi);e^{*},\mu^{*},\pi)-\Psi(\bs{\beta}^{*}(\pi);e^{*},\mu^{*},\pi)\right|\\
\leq & \underbrace{\sup_{\pi\in\Pi}\sup_{\left\Vert \bs{\beta}-\bs{\beta}^{*}(\pi)\right\Vert \leq\delta+C\sqrt{\mathrm{VC}(\Pi)/m}}\left|\widehat{\Psi}_{m}(\bs{\beta};\pi)-\widetilde{\Psi}_{m}(\bs{\beta};\widehat{e},\widehat{\mu},\pi)\right|}_{\mathcal{Q}_{1}}\\
 & +\underbrace{\sup_{\pi\in\Pi}\sup_{\left\Vert \bs{\beta}-\bs{\beta}^{*}(\pi)\right\Vert \leq\delta+C\sqrt{\mathrm{VC}(\Pi)/m}}\left|\widetilde{\Psi}_{m}(\bs{\beta};\widehat{e},\widehat{\mu},\pi)-\Psi(\bs{\beta};\widehat{e},\widehat{\mu},\pi)\right|}_{\mathcal{Q}_{2}}\\
 & +\underbrace{\sup_{\pi\in\Pi}\sup_{\left\Vert \bs{\beta}-\bs{\beta}^{*}(\pi)\right\Vert \leq\delta+C\sqrt{\mathrm{VC}(\Pi)/m}}\left|\Psi(\bs{\beta};\widehat{e},\widehat{\mu},\pi)-\Psi(\bs{\beta};e^{*},\mu^{*},\pi)\right|}_{\mathcal{Q}_{3}}\\
 & +\underbrace{\sup_{\pi\in\Pi}\left|\Psi(\widehat{\bs{\beta}}_{I}(\pi);e^{*},\mu^{*},\pi)-\Psi(\bs{\beta}^{*}(\pi);e^{*},\mu^{*},\pi)\right|}_{\mathcal{Q}_{4}}
\end{align*}
 w.p. at least $1-3c_{1}\exp\left(-c_{2}m\delta^{2}\right)$
for all $0<\delta<c_{3}$ and large $m$. The bounds for
$\mathcal{Q}_{1},\mathcal{Q}_{2},\mathcal{Q}_{3}$ follow from the same
localized classes as in Steps 2.1--2.3 in the proof of
Lemma~\ref{lem:Convergence rate of beta_H-unknown propensity}, with
$\mathcal{L}_{1}^{\prime}$ replaced by $U$ and with
$\|\bs{\beta}-\bs{\beta}^{*}(\pi)\|\leq\delta+C\sqrt{\mathrm{VC}(\Pi)/m}$.
Thus,
  \[
\mathcal{Q}_{1}+\mathcal{Q}_{2}+\mathcal{Q}_{3}\leq C\sqrt{\frac{\mathrm{VC}(\Pi)}{m}}+Ca_{m}(\Pi)^{2}+Ca_{m}(\Pi)\left\{\delta+C\sqrt{\frac{\mathrm{VC}(\Pi)}{m}}\right\}+Ct
\]
w.p. at least $1-c_{1}\exp(-c_{2}mt^{2})$ for all $0<t<c_{3}$,
$0<\delta<c_{4}$ and large $m$. By letting $t=\delta$
  and using the same second-order dominance relations verified in the
proof of Lemma~\ref{lem:Convergence rate of beta_H-unknown propensity},
we have \(\mathcal{Q}_{1}+\mathcal{Q}_{2}+\mathcal{Q}_{3}\leq C\sqrt{\mathrm{VC}(\Pi)/m}+\delta\)
w.p. at least $1-c_{1}\exp(-c_{2}m\delta^{2})$ for all $0<\delta<c_{3}$
and large $m$.

Now, we analyze $\mathcal{Q}_{4}$. Note that \(\Psi(\bs{\beta};e^{*},\mu^{*},\pi)=\e[\{\pi(\bs X)T/e^{*}(\bs X)+(1-\pi(\bs X))(1-T)/(1-e^{*}(\bs X))\} U(Y,\bs X,\bs{\beta})]=\Psi(\bs{\beta};\pi)\), where $\Psi(\bs{\beta};\pi)$ is defined in Assumption~\ref{assu:regularity-assumptions-on-U}.
Then $\mathcal{Q}_{4}=\sup_{\pi\in\Pi}\left|\Psi(\widehat{\bs{\beta}}_{I}(\pi);\pi)-\Psi(\bs{\beta}^{*}(\pi);\pi)\right|$.
On the event that $\widehat{\bs{\beta}}_{I}(\pi)$ lies in the $c_{0}$-neighborhood of $\bs{\beta}^{*}(\pi)$ uniformly over $\pi\in\Pi$, the mean value theorem and Assumption~\ref{assu:regularity-assumptions-on-U}(iii) give
$\mathcal{Q}_{4}\leq\overline{\Psi}^{\prime}\sup_{\pi\in\Pi}\Vert \widehat{\bs{\beta}}_{I}(\pi)-\bs{\beta}^{*}(\pi)\Vert$.
Applying Lemma~\ref{lem:Convergence rate of beta_H-unknown propensity}, we obtain
$P(\mathcal{Q}_{4}\geq\delta+C\sqrt{\mathrm{VC}(\Pi)/m})\leq c_{1}\cdot\exp(-c_{2}m\delta^{2})$
for all $0<\delta<c_{3}$ and large $m$.

Combining the results for $\mathcal{Q}_{1}+\mathcal{Q}_{2}+\mathcal{Q}_{3}$
and $\mathcal{Q}_{4}$ we have
  \[
P\left(\sup_{\pi\in\Pi}\left|\widehat{W}_{I}(\pi)-W(\pi)\right|\geq\delta+C\sqrt{\frac{\mathrm{VC}(\Pi)}{m}}\right)\leq C_{2}^{\prime}\exp(-C_{3}^{\prime}m\delta^{2})
\]
for all $0<\delta<C^{\prime}$ and large $m$, where
$C^{\prime},C_{1}^{\prime},\ldots,C_{3}^{\prime}>0$ are constants.
Note that $\sup_{\pi\in\Pi}\left|\widehat{W}_{I}(\pi)-W(\pi)\right|\leq C$
almost surely for some constant $C>0$, so the probability below is zero when $\delta\geq2C$. For $0<\delta<2C$, the preceding bound gives
\[
P\left(\sup_{\pi\in\Pi}\left|\widehat{W}_{I}(\pi)-W(\pi)\right|\geq\delta+C\sqrt{\frac{\mathrm{VC}(\Pi)}{m}}\right)
\leq C_{2}^{\prime}\exp\left\{ -C_{3}^{\prime}\frac{\left(C^{\prime}\right)^{2}}{4C^{2}}m\delta^{2}\right\}.
\]
Together with the zero-probability case, this bound holds for all $\delta>0$ and large $m$. This completes the proof.
\end{proof}

\subsection{Proof of Theorem \ref{thm:oracle holdout unknown propensity}}

We first verify Assumption~\ref{ass:high-level} for the welfare function
$W(\pi)$. For any $\pi_{1},\pi_{2}\in\Pi_{\infty}$, write
$d_{\Delta}:=P(\pi_{1}(\bs X)\neq\pi_{2}(\bs X))$ and $w_{\pi}:=\pi(\bs X)T/e^{*}(\bs X)+(1-\pi(\bs X))(1-T)/(1-e^{*}(\bs X))$.
By overlap, $\left|w_{\pi}\right|\le\kappa^{-1}$ and $\left|w_{\pi_{1}}-w_{\pi_{2}}\right|\le\kappa^{-1}1\{\pi_{1}(\bs X)\neq\pi_{2}(\bs X)\}$.
Recall that $Q_{j}(\beta;\pi):=\e[w_{\pi}\mathcal{L}_{j}(Y-\beta)]$.
For $j=1,\ldots,p$, since $\beta^{*}_{j}(\pi_{2})$ minimizes
$Q_{j}(\beta;\pi_{2})$, we have $Q_{j}^{\prime}(\beta^{*}_{j}(\pi_{2});\pi_{2})=0$.
Moreover, for any fixed $\beta$, the derivative of $Q_j(\beta;\pi)$
with respect to $\beta$ depends on $\pi$ only through $w_{\pi}$. Hence,
by the preceding display and Assumption~\ref{assu:regularity-conditions-on-L}(iii),
\begin{equation}
\label{eq:Qprime-policy-distance-bound}
\begin{aligned}
\left|Q_{j}^{\prime}(\beta^{*}_{j}(\pi_{2});\pi_{1})\right|
&=\left|Q_{j}^{\prime}(\beta^{*}_{j}(\pi_{2});\pi_{1})-Q_{j}^{\prime}(\beta^{*}_{j}(\pi_{2});\pi_{2})\right|\\
&\le \e\left[\left|w_{\pi_{1}}-w_{\pi_{2}}\right|\left|\mathcal{L}_{j}^{\prime}(Y-\beta^{*}_{j}(\pi_{2}))\right|\right]\le C_{L}d_{\Delta},
\end{aligned}
\end{equation}
where $C_{L}>0$ is a constant.
Choose $\eta>0$ small enough so that $\eta\le c_{0}/\sqrt{p}$, where
$c_{0}$ is the neighborhood radius in
Assumption~\ref{assu:regularity-conditions-on-L}, and
$Q_{j}^{\prime\prime}(\beta;\pi_{1})\ge\underline{Q}^{\prime\prime}/2$
whenever $|\beta-\beta^{*}_{j}(\pi_{1})|\le\eta$. Such an $\eta$ exists by
Assumption~\ref{assu:regularity-conditions-on-L}(ii).
Let $\delta_{0}:=\min\{1,\underline{Q}^{\prime\prime}\eta/(2C_{L})\}$.
If $d_{\Delta}<\delta_{0}$ and $|\beta^{*}_{j}(\pi_{2})-\beta^{*}_{j}(\pi_{1})|>\eta$, then
consider the point $\bar{\beta}_{j}$ between $\beta^{*}_{j}(\pi_{1})$
and $\beta^{*}_{j}(\pi_{2})$ such that
$|\bar{\beta}_{j}-\beta^{*}_{j}(\pi_{1})|=\eta$. Since
$Q_{j}^{\prime}(\beta^{*}_{j}(\pi_{1});\pi_{1})=0$, the local lower
bound on $Q_{j}^{\prime\prime}(\cdot;\pi_{1})$ gives \(\left|Q_{j}^{\prime}(\bar{\beta}_{j};\pi_{1})\right|=\left|\int_{\beta^{*}_{j}(\pi_{1})}^{\bar{\beta}_{j}}Q_{j}^{\prime\prime}(u;\pi_{1})du\right|\ge\underline{Q}^{\prime\prime}\eta/2\).
If $\beta^{*}_{j}(\pi_{2})>\beta^{*}_{j}(\pi_{1})$, then
$\beta^{*}_{j}(\pi_{2})\ge\bar{\beta}_{j}$ and
$Q_{j}^{\prime}(\bar{\beta}_{j};\pi_{1})\ge\underline{Q}^{\prime\prime}\eta/2$.
Since $Q_{j}^{\prime}(\cdot;\pi_{1})$ is nondecreasing, it follows
that $Q_{j}^{\prime}(\beta^{*}_{j}(\pi_{2});\pi_{1})\ge\underline{Q}^{\prime\prime}\eta/2$.
If instead $\beta^{*}_{j}(\pi_{2})<\beta^{*}_{j}(\pi_{1})$, then
$\beta^{*}_{j}(\pi_{2})\le\bar{\beta}_{j}$ and
$Q_{j}^{\prime}(\bar{\beta}_{j};\pi_{1})\le-\underline{Q}^{\prime\prime}\eta/2$,
so monotonicity gives
$Q_{j}^{\prime}(\beta^{*}_{j}(\pi_{2});\pi_{1})\le-\underline{Q}^{\prime\prime}\eta/2$.
In both cases,
$|Q_{j}^{\prime}(\beta^{*}_{j}(\pi_{2});\pi_{1})|\ge\underline{Q}^{\prime\prime}\eta/2$,
which contradicts (\ref{eq:Qprime-policy-distance-bound}), because
$d_{\Delta}<\delta_{0}\le\underline{Q}^{\prime\prime}\eta/(2C_{L})$ implies
$C_{L}d_{\Delta}<\underline{Q}^{\prime\prime}\eta/2$. Hence $|\beta^{*}_{j}(\pi_{2})-\beta^{*}_{j}(\pi_{1})|\le\eta$
whenever $d_{\Delta}<\delta_{0}$.

Now consider the case $d_{\Delta}<\delta_{0}$.
On the interval between
$\beta^{*}_{j}(\pi_{1})$ and $\beta^{*}_{j}(\pi_{2})$, the local lower
bound $Q_{j}^{\prime\prime}(\beta;\pi_{1})\ge\underline{Q}^{\prime\prime}/2$
therefore applies. By the mean value theorem applied to
$Q_{j}^{\prime}(\cdot;\pi_{1})$, together with
$Q_{j}^{\prime}(\beta^{*}_{j}(\pi_{1});\pi_{1})=0$ and
(\ref{eq:Qprime-policy-distance-bound}), \(\frac{\underline{Q}^{\prime\prime}}{2}|\beta^{*}_{j}(\pi_{2})-\beta^{*}_{j}(\pi_{1})|\le |Q_{j}^{\prime}(\beta^{*}_{j}(\pi_{2});\pi_{1})-Q_{j}^{\prime}(\beta^{*}_{j}(\pi_{1});\pi_{1})|\le C_{L}d_{\Delta}\).
Thus $|\beta^{*}_{j}(\pi_{2})-\beta^{*}_{j}(\pi_{1})|\le C d_{\Delta}$ for $j=1,\ldots,p$, and therefore $\|\bs{\beta}^{*}(\pi_{2})-\bs{\beta}^{*}(\pi_{1})\|\le C d_{\Delta}$.
The choice $\eta\le c_{0}/\sqrt{p}$ also ensures that the line segment
between $\bs{\beta}^{*}(\pi_{1})$ and $\bs{\beta}^{*}(\pi_{2})$ lies
inside the $c_{0}$-neighborhood of $\bs{\beta}^{*}(\pi_{1})$ whenever
$d_{\Delta}<\delta_{0}$.
Using $W(\pi)=\Psi(\bs{\beta}^{*}(\pi);\pi)$, Assumption~\ref{assu:regularity-assumptions-on-U}(iii),
and the boundedness of $U$ in Assumption~\ref{assu:regularity-assumptions-on-U}(i), we obtain
\[
\begin{aligned}
|W(\pi_{1})-W(\pi_{2})|
&\leq |\Psi(\bs{\beta}^{*}(\pi_{1});\pi_{1})-\Psi(\bs{\beta}^{*}(\pi_{2});\pi_{1})|+|\Psi(\bs{\beta}^{*}(\pi_{2});\pi_{1})-\Psi(\bs{\beta}^{*}(\pi_{2});\pi_{2})|\le C d_{\Delta}
\end{aligned}
\]
whenever $d_{\Delta}<\delta_{0}$.

If $d_{\Delta}\ge\delta_{0}$, then
the bound $\left|W(\pi)\right|\le C$ follows from
$|w_{\pi}|\le\kappa^{-1}$ and Assumption~\ref{assu:regularity-assumptions-on-U}(i),
so $|W(\pi_1)-W(\pi_2)|\le 2C\le (2C/\delta_0)d_{\Delta}$. Therefore there
exists a constant $C_{W}>0$ such that $\left|W(\pi_{1})-W(\pi_{2})\right|\le C_{W}P(\pi_{1}(\bs X)\neq\pi_{2}(\bs X))$
for all $\pi_{1},\pi_{2}\in\Pi_{\infty}$. This proves Assumption~\ref{ass:high-level}.

By Lemma~\ref{lem:Convergence rate of WH_hat-unknown propensity}, the
debiased empirical welfare function $\widehat{W}_{I}(\pi)$ satisfies
Assumption~\ref{assu:general ass uniform convergence} for policy classes
with VC dimension at least one. For Assumption~\ref{assu:general assu on hatW},
repeating the same proof for a fixed policy $\pi$ yields the simpler bound
$P(|\widehat{W}_{I}(\pi)-W(\pi)|\geq\delta+C/\sqrt{m})\leq C_{1}\exp(-C_{2}m\delta^{2})$
for all $\delta>0$ and large $m$. Together with the preceding
Lipschitz bound for $W$, this verifies all claims of
Theorem~\ref{thm:oracle holdout unknown propensity}. This completes the proof.

\makeatletter
\def\bibsection{\pdfbookmark[1]{\bibname}{bibliography-section}\section*{\bibname}}
\makeatother
\bibliographystyle{abbrvnat}
\bibliography{r}

\end{document}